\pdfoutput=1
\documentclass{article}

\usepackage[pass]{geometry}
\usepackage[T1]{fontenc}
\usepackage[utf8]{inputenc}
\usepackage{lmodern}
\usepackage{amsmath,amssymb,amsthm}
\usepackage{mathtools}
\mathtoolsset{showonlyrefs}
\usepackage{graphicx}
\usepackage{booktabs}
\usepackage{algorithm}
\usepackage{algpseudocode}
\usepackage{placeins}
\usepackage{natbib}
\usepackage{xurl}
\usepackage{hyperref}
\hypersetup{
    pdftitle={Approximate Bayesian Fusion},
    pdfauthor={Filippo Pagani; Adam M. Johansen; Gareth O. Roberts},
    pdfkeywords={Divide-and-conquer, distributed inference, sequential Monte Carlo, Brownian bridge, error bounds},
    colorlinks=true,
    linkcolor=blue,
    citecolor=blue,
    urlcolor=blue}
\renewcommand\harvardurl[1]{\textbf{URL:} \url{#1}}

\newtheorem{assumption}{Assumption}
\newtheorem{lemma}{Lemma}
\newtheorem{proposition}{Proposition}
\newtheorem{definition}{Definition}
\newtheorem{theorem}{Theorem}
\newtheorem{remark}{Remark}

\title{\bf Approximate Bayesian Fusion}
\author{Filippo Pagani\thanks{Corresponding author: filippo.pagani@warwick.ac.uk}, Adam M. Johansen and Gareth O. Roberts \\[1.5ex]
    Department of Statistics, University of Warwick}
\date{}

\begin{document}

\maketitle

\begin{abstract} \noindent
    Distributed inference plays an important role in modern computational statistics and machine learning. In the frequentist setting, methods can often aggregate results by simply averaging estimators computed on data subsets. However, distributed Bayesian inference is more challenging, as it requires reconstructing the full posterior by combining subposterior distributions computed on different machines. Bayesian Fusion performs distributed Bayesian inference exactly, but at a sizeable computational cost. We propose a method that uses numerical approximations to reduce the cost of the most computationally expensive steps of Bayesian Fusion, and expands its domain of applicability. We bound the numerical error as a function of the step size of the approximation. Under regularity conditions, our approach can approximate the target posterior distribution to an arbitrary degree of accuracy, and is a practical and efficient alternative for large-scale distributed Bayesian inference. In numerical experiments on a logistic regression with more than $300{,}000$ observations, our method reaches the accuracy reported for Divide-and-Conquer Bayesian Fusion about two orders of magnitude faster. On a proportional hazards model with non-Gaussian posteriors, which lacks the lower bounds that Bayesian Fusion requires, our method reproduces funnel features of the posterior that simpler methods like Consensus Monte Carlo miss.
\end{abstract}

\noindent%

\section{Introduction}

Modern statistical analyses increasingly require inference from data distributed across machines, institutions, or other locations. This may arise from computational constraints, the scale of the dataset, or restrictions on data sharing. Distributed inference addresses this by analysing data subsets, or shards, separately and then combining the results. In Bayesian inference, this is challenging because independent shard analyses produce subposteriors that must be recombined to recover the full posterior, which is a non-trivial task.

A wide range of methods has been proposed for this problem. In this work we do not consider methods based primarily on subsampling the full dataset \citep{welling2011, nemeth2021, bierkens2019}, but we focus on approaches that initially perform inference independently on each shard and subsequently combine the resulting distributions or samples. 
These approaches range from fast approximations to exact but computationally expensive methods.

At one end of the spectrum, approximate methods simplify either the subposteriors or their combination. Consensus Monte Carlo (CMC) \citep{scott2016cmc} combines subposterior draws by weighted averaging. It is exact when the subposteriors are Gaussian, but its error can be arbitrarily large when they are not.
Density-based approaches instead estimate or smooth individual subposteriors before combining them, with kernel density estimates \citep{neiswanger2014} or the Weierstrass transform \citep{wang2014}, but density estimation becomes increasingly difficult as the dimension grows. More flexible approximations have also been developed using diffusion-based models of shard-level subposteriors \citep{trojan2024}.

At the other end of the spectrum, exact methods target the full posterior itself. Monte Carlo Fusion (MCF) \citep{dai2019mf} uses a rejection-sampling construction to simulate from the exact posterior distribution given samples from each of the subposteriors.
However, this approach is computationally expensive, particularly 
as the dimension or the number of subposteriors increases. Bayesian Fusion (BF) \citep{dai2023bf} improves computational performance over MCF by embedding fusion within a sequential Monte Carlo framework. Furthermore, Divide-and-Conquer Bayesian Fusion \citep{chan2023dcf} recursively combines subposteriors through a tree-structured scheme, which improves efficiency further. 
These methods are exact, or consistent when implemented with SMC, but they remain expensive, scale poorly, and require model-specific bounds that make them hard to apply.

There currently does not exist methodology for the fusion problem which is both scalable and equipped with a controllable approximation error.

In this work, we develop and analyse Approximate Bayesian Fusion (ABF). ABF retains the sequential Monte Carlo structure of BF, but replaces the unbiased estimates of its importance weights with a simple numerical approximation. The contributions of this article are:
\begin{enumerate}
    \item A fusion method that is robust, scalable, and needs no model-specific bounds. It applies to any model whose subposteriors can be sampled and whose log densities can be differentiated twice.
    \item Control of the approximation error. Under regularity conditions, the total variation distance between the exact and the approximate fusion distributions decreases as the square root of the time step of the approximation (Theorem~\ref{th:bound_strong}). Unlike for CMC and similar methods, the approximation error can therefore be reduced at the cost of more computation.
    \item Numerical experiments. On the logistic regression of \cite{chan2023dcf}, ABF reaches the accuracy reported for Divide-and-Conquer Bayesian Fusion about two orders of magnitude faster (Section~\ref{sec:nycflights}). On a proportional hazards model where the bounds required by BF do not exist, so that BF cannot be applied, ABF recovers the funnel shapes of the posterior that CMC misses (Section~\ref{sec:ph}).
\end{enumerate}
Our error bounds cover the logistic regression of Section~\ref{sec:logistic_comparison}, but not the other examples, which use an adaptive step size and divide-and-conquer fusion. These examples show that ABF works well beyond the conditions of the theory.

The rest of the manuscript is organised as follows. Section~\ref{sec:BF} reviews Bayesian Fusion, on which ABF is built, and introduces the notation used in the rest of the paper. Section~\ref{sec:ABF} introduces ABF and its numerical approximation, and Section~\ref{sec:controlling_error} bounds the error of this approximation as a function of the step size. Section~\ref{sec:implementation} describes how we choose the parameters of ABF and how we simulate it efficiently, including its divide-and-conquer version. Section~\ref{sec:experiments} compares ABF with other methods on four examples, two simulated and two based on real data, and Section~\ref{sec:discussion} shows our conclusions.

\section{Bayesian Fusion} \label{sec:BF}

\subsection{The fusion target}

Suppose that the data $X = (X_1,\ldots,X_C)$ are split across $C$ machines, cores, or shards, and that the parameter of interest is $\vartheta \in \mathbb{R}^d$. Running inference separately on shard $c$ gives access to samples from a shard-specific subposterior $f_c(\vartheta) = f_c(\vartheta | X_c)$. Bayesian Fusion is used when these subposterior densities are constructed so that the target posterior density satisfies
\begin{equation}
    f(\vartheta) \propto \prod_{c=1}^C f_c(\vartheta) \, .
\end{equation}
This product representation can be constructed, for example, by distributing part of the factors composing the likelihood alongside a fractional prior across shards, or by assigning the prior to one of the factors and likelihood contributions to the others. In what follows we assume only that this product representation holds.

The aim of Bayesian Fusion is to obtain samples from $f$ given samples from the $C$ component densities $f_1,\ldots,f_C$. The construction in \cite{dai2023bf} can be understood in three parts. 
First, one draws an initial sample from each of the $C$ subposteriors. 
Second, these  $C$ samples are propagated according to an interacting diffusion whose components coalesce at a common value at time $T$. 
Third, a change of measure reweights these paths so that, under the resulting fusion measure, this common terminal value has density $f$. In the SMC implementation, the propagation and reweighting in the latter two parts are performed sequentially over a temporal partition.

\subsection{The coalescing diffusion}

For each shard $c=1,\ldots,C$, Bayesian Fusion defines a process $\theta_t^{(c)} \in \mathbb{R}^d$ on $[0,T]$. We write
\begin{equation}
    \theta_t = (\theta_t^{(1)}, \ldots, \theta_t^{(C)})
\end{equation}
for the full $Cd$-dimensional state at time $t$, which can be viewed as $C$ coupled copies of the original parameter, one for each shard. The process is initialised with $\theta_0^{(c)} \sim f_c$ independently, and is constructed so that all shard-specific components coalesce at time $T$:
\begin{equation}
    \theta_T^{(1)} = \ldots = \theta_T^{(C)} := \vartheta \, ,
\end{equation}
where the common terminal value $\vartheta$ is the fused sample.

More explicitly, each particle evolves according to the coupled system of stochastic differential equations
\begin{equation} \label{eq:sde}
    \mathrm{d} \theta_t^{(c)} = \frac{\bar{\theta}_t - \theta_t^{(c)}}{T-t} \, \mathrm{d}t + \mathrm{d} W_t^{(c)} \, , \quad \theta_0^{(c)} \sim f_c \, , \quad c=1,\ldots,C \, ,
\end{equation}
where $\{ W_t^{(c)} \}_{c=1}^C$ are independent $d$-dimensional Brownian motions and $\bar{\theta}_t = C^{-1} \sum_{c=1}^C \theta_t^{(c)}$ represents the mean across different shards.
The drift term pulls each shard-specific component $\theta_t^{(c)}$ towards the current mean across shards $\bar \theta_t$. Conditionally on $\theta_0$ and on the common endpoint $\vartheta$, the $C$ components are independent Brownian bridges from $\theta_0^{(c)}$ to $\vartheta$ \citep[proof of Theorem~1]{dai2023bf}, so all components coalesce at time $T$, as illustrated in Figure~\ref{fig:coalescence}.

\begin{figure}
    \centering
    \includegraphics[width=1\linewidth]{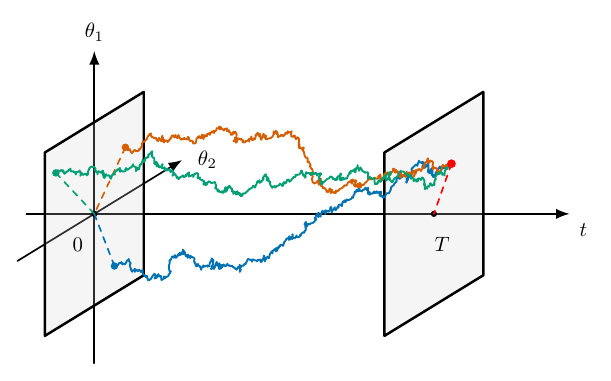}
    \caption{Schematic of $C=3$ separate shards collapsing on the same point in two dimensions at the coalescence time $T$ by following Brownian bridge trajectories.}
    \label{fig:coalescence}
\end{figure}

\subsection{The fusion weight}

On its own, the coalescing diffusion of Equation~\eqref{eq:sde} does not give the fused sample $\vartheta$ the density $f$. Bayesian Fusion corrects this by reweighting each path $\theta_{0:T}$ by the fusion weight
\begin{align} \label{eq:weight}
    \rho(\theta_{0:T})
    & := \rho_0 ( \theta_0) \, \exp \left \{ - \sum_{c=1}^C \int_0^T ( \phi^{(c)}(\theta_s^{(c)}) - \Phi^{(c)} ) \, \mathrm{d} s \right \} \, ,
\end{align}
and we proceed to explain the role of its terms one by one.

The first factor in the equation above is the initial weight
\begin{equation}
    \label{eq:init_weight}
    \rho_0 ( \theta_0 ) := \exp \left \{ - \frac{1}{2T} \sum_{c=1}^C \| \theta_0^{(c)} - \bar{\theta}_0 \|^2 \right \} \, ,
\end{equation}
where $\| \cdot \|$ is the Euclidean $2$-norm. It downweights initial shard configurations whose components are far apart relative to the coalescence time $T$. The second factor in Equation \eqref{eq:weight} updates the path weight\footnote{For the technical reader, this factor comes from the Dacunha-Castelle representation of the transition density of a Langevin diffusion with invariant density proportional to $f_c^2$; see the proof of Proposition~2 in \cite{dai2019mf}.} according to the trajectory between times $0$ and $T$.

The function $\phi^{(c)}$ contains the shard-specific posterior information needed for this path correction. For a general $\theta \in \mathbb{R}^d$, it is defined as
\begin{align}
    \phi^{(c)}(\theta) 
    & := \frac{1}{2} \| \nabla \log f_c(\theta) \|_2^2 + \frac{1}{2} \sum_{r=1}^d \frac{\partial^2}{\partial \theta_r^2} \log f_c(\theta) \, . \label{eq:phi}
\end{align}
The constant $\Phi^{(c)}$ denotes a lower bound for $\phi^{(c)}$ whose value need not be known to implement the algorithm described in this work.

This construction is justified by Theorem~1 of \cite{dai2023bf}, which states that 
under the law of the coalescing diffusion reweighted by $\rho$, the common terminal value $\vartheta$ has density $f$. Since each $\phi^{(c)}$ depends only on the subposterior $f_c$, Bayesian Fusion turns samples from the $C$ subposteriors into samples from the full posterior without evaluating $f$. The difficulty lies in computing the weights, which we address in the next subsection.

\subsection{Sequential Monte Carlo implementation}

Bayesian Fusion approximates the fusion posterior by a weighted particle system. Suppose that we have access to $N$ samples from each subposterior $f_c$, either by direct sampling or from a previous MCMC run.
Each SMC particle contains one draw from each of the $C$ shards and therefore has dimension $Cd$. The algorithm propagates these particles forward on a temporal partition
\begin{equation} \label{eq:partition}
    \mathcal{P} = \{ t_0, t_1, \ldots, t_n \quad | \quad 0 := t_0 < t_1 < \ldots < t_n := T \, , \quad t_j - t_{j-1} = \Delta_j  \} \, .
\end{equation}
Between successive times in this partition, each particle is propagated according to the transition law induced by Equation \eqref{eq:sde}:
\begin{equation} \label{eq:transitions_theta}
    \theta_{t_{j}} \mid \theta_{t_{j-1}} \sim \mathcal{N} ( M_{t_{j-1}, t_{j}}, V_{t_{j-1}, t_{j}}) \, .
\end{equation}
The mean vector is $M_{t_{j-1}, t_{j}} = ( M_{t_{j-1}, t_{j}}^{(1)},\ldots,M_{t_{j-1}, t_{j}}^{(C)} )$, with
\begin{equation} \label{eq:transitionsM}
    M_{t_{j-1}, t_{j}}^{(c)} = \frac{T - t_{j}}{T - t_{j-1}} \theta_{t_{j-1}}^{(c)} + \frac{t_{j} - t_{j-1}}{T - t_{j-1}} \bar{\theta}_{t_{j-1}} \, ,
\end{equation}
and the covariance is $V_{t_{j-1}, t_{j}} = \Sigma_{t_{j-1}, t_{j}} \otimes I_d$, where $\otimes$ denotes the Kronecker product, and
\begin{equation} \label{eq:transitionsS}
    \Sigma_{t_{j-1}, t_{j}} = \frac{(t_{j} - t_{j-1}) (T - t_{j})}{(T - t_{j-1})} I_C + \frac{(t_{j} - t_{j-1})^2}{C (T - t_{j-1})} \mathbf{1}_C \mathbf{1}_C^\top \, ,
\end{equation}
where $\mathbf{1}_C$ represents a vector of ones of length $C$. This transition is exact, so the coalescing diffusion is simulated without discretisation error.

At every transition, the algorithm updates particle weights. For a path segment $\theta_{[t_{j-1},t_j]} = \{ \theta_s : s \in [t_{j-1}, t_j] \}$, define the incremental weight function as
\begin{equation} \label{eq:weight2}
    \rho_j(\theta_{[t_{j-1},t_j]}) = \exp \left \{ -\sum_{c=1}^C \int_{t_{j-1}}^{t_{j}} ( \phi^{(c)}(\theta_s^{(c)}) - \Phi^{(c)} ) \, \mathrm{d} s \right \} \, .
\end{equation}
For particle $i$, the unnormalised weight after step $j$ is
\begin{equation}
    w_{j,i} := W_{j-1,i} \, \rho_j(\theta_{[t_{j-1},t_j],i}) \, ,
\end{equation}
where $W_{j-1,i}$ is the normalised weight of the particle after step $j-1$, after any resampling, and $W_{0,i} \propto \rho_0(\theta_{0,i})$. Without resampling, $w_{j,i}$ is proportional to $\rho_0(\theta_{0,i}) \prod_{k=1}^{j} \rho_k(\theta_{[t_{k-1},t_k],i})$. The weights $w_{j,1},\ldots,w_{j,N}$ are then normalised across particles, 
and resampling is performed when the effective sample size (ESS) drops below a chosen threshold $\tau N$, where $\tau \in (0,1]$. The final output is a weighted particle approximation to the fusion posterior $f$.

Bayesian Fusion is exact as $N \to \infty$, but its weights are costly to compute. Each incremental weight requires a path integral along the bridge between two grid points, which is not simulated, so Bayesian Fusion uses the conditional expectation of $\rho_j$ given $\theta_{t_{j-1}}$ and $\theta_{t_j}$. \cite{dai2023bf} estimate it using the layered Brownian bridge (LBB) construction \citep{beskos2008,pollock2016exact}, which preserves exactness by producing unbiased estimates. This estimator can be expensive, and implementing it requires model-specific bounds for $\phi^{(c)}$ on each layer.

\section{Approximate Bayesian Fusion} \label{sec:ABF}

The core idea behind Approximate Bayesian Fusion is to construct a principled approximate method by replacing the exact, but computationally expensive weight update in Equation \eqref{eq:weight2} with an inexpensive numerical approximation. 

Conceptually, we replace the integral of $\phi^{(c)}$ along the bridge path by a low-order quadrature rule, in our case a simple trapezium approximation using evaluations at the endpoints of the time interval $[t_{j-1}, t_j]$, which is equivalent to approximating $\phi^{(c)}$ by a linear function.
The approximation becomes more accurate as the mesh size $\Delta_j = t_j - t_{j-1}$ decreases, so that the bridge has limited time to wander, and $\phi^{(c)}$ varies little over the interval. This yields an inexpensive approximation of the path integral, and removes the need for the layered Brownian bridge machinery used to build unbiased estimators of the weights in Equation \eqref{eq:weight}, which are computationally costly and require model-specific bounds on $\phi^{(c)}$ for each layer.
As an additional consequence of avoiding the layered Brownian bridge construction, there is no need for us to calculate the minimum of $\phi^{(c)}$ and use it in the weight update in Equation \eqref{eq:weight} as it can be absorbed into the normalising constant of the weights of the particle system, and we can simply ignore the $\Phi^{(c)}$ term in the integral. 

In practice, when using a simple trapezium approximation to the path integral calculated at every SMC transition, the weight update in Equation \eqref{eq:weight2} can be approximated as
\begin{equation}
    \int_{t_{j-1}}^{t_j} \phi^{(c)}(\theta_s^{(c)}) \, \mathrm{d} s \approx \frac{1}{2} \big ( \phi^{(c)}(\theta_{t_{j-1}}^{(c)}) + \phi^{(c)}(\theta_{t_j}^{(c)}) \big ) (t_j - t_{j-1}) \, ,
\end{equation}
for all $j = 1, \ldots, n$ and $c = 1, \ldots, C$. A pseudo-algorithm for Approximate Bayesian Fusion is given in Algorithm \ref{alg:abf-trapezium}.

\begin{algorithm}[!ht] 
\caption{Approximate Bayesian Fusion}
\label{alg:abf-trapezium}
\begin{algorithmic}[1]
\Require Subposterior particles $\{\theta_{0,i}^{(c)}\}_{i=1,\ldots,N}^{c=1,\ldots,C}$, potentials $\{\phi^{(c)}\}_{c=1}^C$, coalescence time $T$, step size $\Delta$, resampling threshold $\tau$.

\Statex
\For{$i=1,\ldots,N$} \Comment{Set up initial weights}
    \State Set the initial particle $\theta_{t_0,i} \gets (\theta_{0,i}^{(1)},\ldots,\theta_{0,i}^{(C)})$.
    \State Set the initial weights $w_{0,i} \gets \rho_0(\theta_{t_0,i})$ using Equation~\eqref{eq:init_weight}.
\EndFor
\State Normalise $W_{0,i} \gets w_{0,i}/\sum_{k=1}^N w_{0,k}$.

\Statex
\For{$j=1,\ldots,n$} \Comment{SMC iteration}
    \If{$\mathrm{ESS}(W_{j-1,1:N}) < \tau N$} \Comment{Resampling step}
        \State Resample ancestor indices $A_1,\ldots,A_N$ from $W_{j-1,1:N}$.
        \State Set $\theta_{t_{j-1},i} \gets \theta_{t_{j-1},A_i}$ and $W_{j-1,i} \gets 1/N$ for all $i$.
    \EndIf
    
    \Statex
    \For{$i=1,\ldots,N$} \Comment{Propagation step and weight update}
        \State Propagate $\theta_{t_j,i}$ from the Gaussian transition in Equations~\eqref{eq:transitions_theta}--\eqref{eq:transitionsS}.
        \State Approximate the path integral over $[t_{j-1},t_j]$ by
        \Statex
        \begin{equation*}
            \ell_{j,i}
            \gets
            \frac{\Delta}{2}
            \sum_{c=1}^C
            \left\{
                \phi^{(c)}(\theta_{t_{j-1},i}^{(c)})
                +
                \phi^{(c)}(\theta_{t_j,i}^{(c)})
            \right\}.
        \end{equation*}
        \State Update weights $w_{j,i} \gets W_{j-1,i} \, \exp\{-\ell_{j,i}\}$.
    \EndFor
    \State Normalise $W_{j,i} \gets w_{j,i}/\sum_{k=1}^N w_{j,k}$.
\EndFor

\Statex
\State Set $Y_i \gets \theta_{t_n,i}^{(1)}$ for $i=1,\ldots,N$.
\State \Return weighted fused particles $\{(Y_i,W_{n,i})\}_{i=1}^N$.
\end{algorithmic}
\end{algorithm}

The resampling scheme for SMC particles that we employ in practice is residual resampling, with the residual part drawn by systematic resampling \citep{douc2005comparison, hol2006resampling}.

The coalescence time $T$, the step size and the resampling threshold $\tau$ required by Algorithm~\ref{alg:abf-trapezium} are chosen as described in Section~\ref{sec:implementation}, which also describes how we simulate ABF efficiently.

The approximation of the integral introduces bias in the weight update, and therefore the resulting algorithm is no longer exact in the limit as $N \to \infty$. However, the bias vanishes as $\max_j \Delta_j \to 0$. In Section~\ref{sec:controlling_error} we show how the bias can be controlled and, particularly, that the distance between the exact target distribution and its numerical approximation is bounded by a term of order $\sqrt{\Delta}$, where $\Delta = \Delta_j$ for all $j$.

\section{Controlling the numerical error} \label{sec:controlling_error}

We now study the error introduced by the numerical approximations used in ABF. Under suitable regularity conditions, we show that the discrepancy between the true posterior and the distribution targeted by the approximate procedure is controlled by the SMC step size $\Delta$. In particular, the approximation error is bounded by a term of order $\sqrt{\Delta}$. 
All the proofs pertaining to this section are gathered in Appendix \ref{app:proofs}. 

Let $\mathbb{P}$ be the law of the process $\theta_t$, described by the SDE \eqref{eq:sde}, and let $\mathbb{E} = \mathbb{E}_\mathbb{P}$ represent the expected value with respect to $\mathbb{P}$. 
As in \cite{dai2023bf} and \cite{chan2023dcf}, we work on a time partition $\mathcal{P} = \{ t_j \}_{j=0}^{n}$, which for simplicity we take to be uniform, $t_j = j \Delta$. The results extend to a deterministic non-uniform partition, with $\Delta$ replaced by the largest step $\max_j \Delta_j$, because the proofs bound the error interval by interval, and each bound increases with the length of the interval.

\noindent Let us make the following assumptions:

\begin{assumption}[Regularity] \label{as:smoothness}
    For each $c \in \{1,\ldots,C\}$, the density $f_c$ is positive on $\mathbb{R}^d$, $\log f_c$ is twice continuously differentiable on $\mathbb{R}^d$, and $f_c$ has finite mean $\mu^{(c)}$ and finite covariance matrix $Q^{(c)}$. Moreover, for each $r = 1, \ldots, d$, the second derivative $\partial^2 f_c / \partial \theta_r^2$ is integrable on $\mathbb{R}^d$, and $\partial f_c(\theta) / \partial \theta_r \to 0$ as $\theta_r \to \pm \infty$.
\end{assumption}

The differentiability is needed for the $\phi^{(c)}$ function to exist. The conditions on the derivatives of $f_c$ are used in the proof of Remark~\ref{rem:Phi_neg}, and the finite moments are used in the lemmas of Appendix~\ref{app:proofs}.

\begin{assumption}[Integrability] \label{as:fubini}
  For each $c \in \{1,\ldots,C\}$, $\phi^{(c)}$ is a Borel measurable function, $\theta_s^{(c)}$ is the solution of the SDE in Equation~\eqref{eq:sde}, and 
  \begin{equation}
      \mathbb{E} \left [ \int_0^T \left | \phi^{(c)}(\theta_s^{(c)}) \right | \, \mathrm{d}s \right ] < \infty \, .
  \end{equation}
\end{assumption}
Assumption \ref{as:fubini} is used to apply Fubini's theorem, in the proofs of all the main results.

\begin{assumption}[Lipschitzness] \label{as:phi_lip}
  There exists a constant $L$ such that, for all $c \in \{1,\ldots,C\}$ and all $\theta,\theta' \in \mathbb{R}^d$,
  \begin{equation}
        |\phi^{(c)}(\theta) - \phi^{(c)}(\theta')| \le L \| \theta-\theta' \| \, ,
  \end{equation}
  where $\| \cdot \|$ is the Euclidean norm.
\end{assumption}
Assumption~\ref{as:phi_lip} is not strictly necessary. It is made for simplicity of exposition, and can be relaxed to weaker conditions without changing the core argument, i.e.\ the scaling in $\Delta$. We discuss this in more detail in Appendix~\ref{app:weaker_assumptions}.

\begin{assumption}[Lower-boundedness] \label{as:phi_min}
    For each $c \in \{1,\ldots,C\}$, $\phi^{(c)}$ is lower bounded: $\exists \Phi^{(c)}$ s.t.\ $\inf_\theta \phi^{(c)}(\theta) \ge \Phi^{(c)} > -\infty$.
\end{assumption}

\noindent Assumption \ref{as:phi_min} and the value of $\Phi^{(c)}$ are used to produce a sample from the layered Brownian bridge in \cite{dai2019mf, dai2023bf, chan2023dcf}. Though it is not necessary to know the value of $\Phi^{(c)}$ for running the Approximate Bayesian Fusion algorithm, its existence guarantees that our importance weights have finite variance. It is also useful when proving Theorem \ref{th:bound_strong}, i.e.\ our theoretical bound on the numerical error on the fusion measure. 

In order to fix notation, let us now give a definition of the trapezium rule tailored to our problem.

\begin{definition}[Trapezium rule] \label{def:trapezium}

    Let $\theta_{t_{0:n}} = \{ \theta_{t} \}_{t \in \mathcal{P}}$. The trapezium rule approximation 
    \begin{equation}
      I_\Delta^{{(c)}} \approx \int_0^T \phi^{(c)}(\theta_s^{(c)}) \, \mathrm{d} s
    \end{equation}
    is defined as 
    \begin{align}
      I_{\Delta}^{(c)} 
      & = \sum_{j=0}^{n-1} \frac{\Delta}{2} \left ( \phi^{(c)}(\theta_{t_{j}}^{(c)}) + \phi^{(c)}(\theta_{t_{j+1}}^{(c)}) \right ) 
      =: \int_0^T \phi^{(c)}_\Delta (s, \theta_{t_{0:n}}^{(c)}) \, \mathrm{d} s \, ,
    \end{align}
    where $\phi_\Delta^{(c)}$ is the piecewise linear interpolation in time of the values $\phi^{(c)}(\theta_{t_j}^{(c)})$ along the discretised path:
    \begin{equation}
        \phi_\Delta^{(c)} (s, \theta_{t_{0:n}}^{(c)}) := \sum_{j=0}^{n-1} \left ( \phi^{(c)} (\theta_{t_j}^{(c)}) + \frac{(s-t_j)}{\Delta} \left (\phi^{(c)} (\theta_{t_{j+1}}^{(c)}) - \phi^{(c)} (\theta_{t_j}^{(c)}) \right ) \right ) \mathbf{1}_{[t_j, t_{j+1})}(s) \, ,
    \end{equation}
    with the final interval being closed, i.e.\ $\mathbf{1}_{[t_{n-1}, t_{n}]}(s)$. Therefore, for $s \in \mathcal{P}$,
    \begin{equation} \label{eq:approx_at_grid}
        \phi_\Delta^{(c)}(s, \theta_{t_{0:n}}^{(c)}) = \phi^{(c)}(\theta_s^{(c)}) \, ,
    \end{equation}
    and for $s = t_j + h$, with $s \in [0,T] \setminus \mathcal{P}$ and $0 \le h \le \Delta$,
    \begin{equation}
        \phi_\Delta^{(c)} (t_j+h, \theta_{t_{0:n}}^{(c)}) = \phi^{(c)}(\theta_{t_{j}}^{(c)}) + \frac{h}{\Delta} (\phi^{(c)}(\theta_{t_{j+1}}^{(c)}) - \phi^{(c)}(\theta_{t_{j}}^{(c)})) \, .
    \end{equation}
\end{definition}

We first establish that the strong error between the exact and approximated integral in the exponential in Equation \eqref{eq:weight} is bounded from above by an expression that, as a function of the step size, scales as $\sqrt{\Delta}$.

\begin{proposition}[Strong error bound] \label{prop:strong_error}
    Let the strong error be defined as
    \begin{align}
        \epsilon_\Delta := \mathbb{E} \left | \sum_{c=1}^C \int_0^T \phi^{(c)}(\theta_s^{(c)}) \, \mathrm{d} s - \sum_{c=1}^C \int_0^T \phi_\Delta^{(c)}(s, \theta_{t_{0:n}}^{(c)}) \, \mathrm{d} s \right | \, .
    \end{align}
    Then under Assumptions \ref{as:smoothness}, \ref{as:fubini}, and \ref{as:phi_lip}, there exists a bound $E_\Delta \ge \epsilon_\Delta$ such that
    \begin{align}
        E_\Delta = 
        L C \sqrt{d} \left ( \Delta \sqrt{\frac{T}{C}} + \Delta \bar{\kappa} + \frac{7}{6} T \sqrt{\Delta} \right ) \, ,
    \end{align}
    where $\bar{\kappa} = C^{-1} \sum_{c=1}^C \kappa^{(c)}$, with
    \begin{equation} \label{eq:kappa}
        \kappa^{(c)} := d^{-1/2} \left ( \sqrt{ \mathrm{tr}(Q^{(c)}) + \| \mu^{(c)} \|^2 } + \sqrt{ \frac{1}{C} \mathrm{tr}(\bar{Q}) + \| \bar{\mu} \|^2 } \right ) \, ,
    \end{equation}
    $\bar{\mu} := C^{-1} \sum_{c=1}^C \mu^{(c)}$, and $\bar{Q} := C^{-1} \sum_{c=1}^C Q^{(c)}$, where $\mu^{(c)}$ and $Q^{(c)}$ are the mean and covariance matrix of $f_c$ in Assumption~\ref{as:smoothness}.
\end{proposition}

\begin{proof}
    See Section \ref{sec:strong_error_proof}.
\end{proof}

Let the fusion measure $\mathbb{F}$ be the probability measure induced by the Radon--Nikodym derivative
\begin{equation}
    \frac{\mathrm{d} \mathbb{F}}{\mathrm{d} \mathbb{P}} (\theta) \propto \rho_0 ( \theta_0) \, \exp \left \{ - \sum_{c=1}^C \int_0^T \phi^{(c)}(\theta_s^{(c)}) \, \mathrm{d} s \right \} \, ,
\end{equation}
and let its numerical approximation $\mathbb{F}_\Delta$ be analogously defined as
\begin{equation}
    \frac{\mathrm{d} \mathbb{F}_\Delta}{\mathrm{d} \mathbb{P}} (\theta) \propto \hat{\rho} (\theta_{0:T}) = \rho_0 (\theta_{0}) \, \exp \left \{ - \sum_{c=1}^C I_\Delta^{(c)} \right \} \, .
\end{equation}
We are now able to translate the result from Proposition \ref{prop:strong_error} into a bound on the distance between the fusion measure
$\mathbb{F}$, and the approximate fusion measure that arises from our numerical approximation $\mathbb{F}_\Delta$. Notably, this bound also scales as $\sqrt{\Delta}$ as $\Delta \downarrow 0$. Theorem \ref{th:bound_strong} below is the main theoretical result of this work.

\begin{theorem}[TV distance bound] \label{th:bound_strong}
    Let Assumptions \ref{as:smoothness}--\ref{as:phi_min} hold.    
    Let $\mathbb{F}$ be the fusion measure on $[0,T]$, and let $\mathbb{F}_\Delta$ be the approximate fusion measure on $[0,T]$ obtained using the numerical approximation to calculate the weights. The total variation distance between these distributions can be bounded as follows. 
    \begin{align}
        \| \mathbb{F} - \mathbb{F}_\Delta \|_{TV}
        & \le \frac{1}{Z} \mathrm{e}^{- C \, \bar{\Phi}  \, T} \; L C \sqrt{d} \left ( \Delta \sqrt{\frac{T}{C}} + \Delta \bar{\kappa} + \frac{7}{6} T \sqrt{\Delta} \right ) \, ,
    \end{align}
    with $\bar{\Phi} = C^{-1} \sum_{c=1}^C \Phi^{(c)}$, where $\Phi^{(c)}$ is the lower bound in Assumption~\ref{as:phi_min}, $\bar{\kappa}$ is defined in Proposition~\ref{prop:strong_error}, $L$ is the Lipschitz constant from Assumption \ref{as:phi_lip}, and $Z$ is the normalising constant of $\mathbb{F}$,
    \begin{equation}
        Z = \mathbb{E} \left [ \rho_0(\theta_0) \exp \left \{ - \sum_{c=1}^C \int_0^T \phi^{(c)}(\theta_s^{(c)}) \, \mathrm{d} s \right \} \right ] \, .
    \end{equation}
\end{theorem}

\begin{proof}
    See Section~\ref{sec:bound_strong_proof}.
\end{proof}

Theorem~\ref{th:bound_strong} compares distributions of whole paths, but what matters in practice is the distribution of the fused sample $\vartheta$. 
Under $\mathbb{F}$, $\vartheta$ has density $f$ \citep{dai2023bf}, and keeping only the endpoint of a path cannot increase the total variation distance between two distributions. Hence the total variation distance between $f$ and the distribution of $\vartheta$ under $\mathbb{F}_\Delta$, which is the target of ABF, satisfies the bound of Theorem~\ref{th:bound_strong}.

\begin{remark} \label{rem:Phi_neg}
    The exponent $- C \, \bar{\Phi}  \, T$ is always positive, because every lower bound $\Phi^{(c)}$ is strictly negative: $\phi^{(c)}$ has mean zero under $f_c$, and it cannot be non-negative everywhere. The proof is in Section~\ref{sec:Phi_neg_proof}.
\end{remark}

The bound of Proposition~\ref{prop:strong_error} is of order $\sqrt{\Delta}$ because its proof adds the absolute values of the local errors on each interval. 
The weak error, in which the local errors can cancel, is smaller. To study it, we make the following assumption on the regularity of $\phi^{(c)}$, which replaces Assumption \ref{as:phi_lip}.

\begin{assumption}[Lipschitz gradient] \label{as:weak_error_regularity}
    For each $c \in \{1,\ldots,C\}$, $\phi^{(c)}$ is twice continuously differentiable, and there exists a constant $U \ge 0$ such that, for all $\theta, \theta' \in \mathbb{R}^d$,
    \begin{equation}
        \| \nabla \phi^{(c)}(\theta) - \nabla \phi^{(c)}(\theta') \| \le U \| \theta - \theta' \| \, .
    \end{equation}
\end{assumption}

\noindent Using Assumption \ref{as:weak_error_regularity}, we can prove the following result.

\begin{proposition}[Weak error bound] \label{prop:weak_error}
    Let Assumptions \ref{as:smoothness}, \ref{as:fubini} and \ref{as:weak_error_regularity} be satisfied, and let the weak error be defined as 
    \begin{equation} \label{eq:weak_error}
        \eta_\Delta := \left | \mathbb{E} \left [ \sum_{c=1}^C \int_0^T \phi^{(c)}(\theta_s^{(c)}) \, \mathrm{d} s - \sum_{c=1}^C \int_0^T \phi_\Delta^{(c)}(s, \theta_{t_{0:n}}^{(c)}) \, \mathrm{d} s \right ] \right | \, .
    \end{equation}
    Then, for small $\Delta$, there exists a bound $H_\Delta \ge \eta_\Delta$ such that
    \begin{equation}
        H_\Delta = \frac{1}{2} U C T \Delta d + \mathcal{O}(\Delta^2) \, ,
    \end{equation}
    where the constant $U$ is defined in Assumption \ref{as:weak_error_regularity} above. 
\end{proposition}

\begin{proof}
    See Section \ref{sec:weak_error_proof}.
\end{proof}

Proposition \ref{prop:weak_error} shows that for small $\Delta$, the weak error scales linearly in the step size $\Delta$.
Although the rate in $\Delta$ is faster, it is not clear how to turn Proposition~\ref{prop:weak_error} into a bound on the distance between the exact and approximate fusion measures, as Theorem~\ref{th:bound_strong} does for the strong error.

\section{Practical implementation} \label{sec:implementation}

\subsection{ABF parameter tuning}
\label{sec:abftuning}

Approximate Bayesian Fusion requires a coalescence time $T$, a mesh $0 = t_0 < t_1 < \cdots < t_n = T$ for the SMC transitions, and a resampling threshold $\tau$. The threshold $\tau$ plays its usual role: particles are resampled whenever the ESS falls below $\tau N$, and we use $\tau = 0.5$ in our experiments. We choose $T$ first, through the target acceptance probability $a_0$ of a rejection sampler, and then the mesh, through a target ratio $\zeta_\star$ for the conditional effective sample size (CESS) of each transition.

\subsubsection{Tuning the coalescence time \texorpdfstring{$T$}{T}} \label{sec:coalescence_time}

The coalescence time enters the initial weight $\rho_0$ in Equation~\eqref{eq:init_weight}, which penalises the dispersion of the shard components of a particle relative to $T$. A small $T$ gives negligible weight to all but the least dispersed tuples, so the initial weights can degenerate. A large $T$ admits more dispersed tuples, but it requires a longer bridge and more SMC transitions.

Rather than weighting the initial particles by $\rho_0$, we select them by rejection sampling. Each proposal is a tuple formed by drawing one particle independently from each shard, and it is accepted with probability $\rho_0$, so that the $N$ accepted tuples carry equal weights. This replaces the initial weights in Algorithm~\ref{alg:abf-trapezium} by $W_{0,i} = 1/N$.

We choose $T$ so that the acceptance probability of this sampler equals a target $a_0 \in (0,1)$. Writing $\rho_0 = \exp \{ - \mathcal{E} / T \}$, with energy $\mathcal{E} := \frac{1}{2} \sum_{c=1}^C \| \theta^{(c)} - \bar{\theta} \|^2$, we compute the energies $\mathcal{E}_1, \ldots, \mathcal{E}_M$ of $M$ proposed tuples and solve
\begin{equation}
    \widehat{a}(T) := \frac{1}{M} \sum_{q=1}^M \exp \{ - \mathcal{E}_q / T \} = a_0
\end{equation}
by bracketing and bisection. Appendix~\ref{app:rejection_acceptance_tuning} gives the details, and Algorithms~\ref{alg:tune-coalescence-time} and~\ref{alg:rejection-initialisation} there summarise the calibration and the initialisation.

Since $\widehat{a}$ is increasing in $T$ (see Proposition~\ref{prop:rejection_jensen}), smaller values of $a_0$ give smaller values of $T$, at the cost of more proposals: the sampler needs on average $N / a_0$ proposals to accept $N$ tuples. If one proposal takes on average a time $t_{\mathrm{iter}}(d,C)$, estimated from a short pilot run, a time budget $\mathcal{B}$ for this step gives
\begin{equation}
    a_0 = \frac{N \, t_{\mathrm{iter}}(d,C)}{\mathcal{B}} \, .
\end{equation}
In our experiments, $a_0 \in [0.01, 0.1]$ worked well across a range of problems, and the algorithm was not particularly sensitive to the value within this range.

\subsubsection{Choosing the step size \texorpdfstring{$\Delta$}{Delta}} \label{sec:adaptive_step_size}

The error analysis of Section~\ref{sec:controlling_error} is carried out on a uniform mesh $t_j = j \Delta$. Such a step size can be calibrated once $T$ is known: we can take the smallest integer $n$ such that the ESS after the first SMC transition is at least $\zeta_\star N$, found by bisection, and set $\Delta = T / n$. 

The weight degeneracy caused by a transition of given length, however, tends to change along the bridge, so a step size that suits the first transition is not necessarily the same for all the others. We therefore let the step size adapt, as suggested by \cite{chan2023dcf}, though with a different rule, and require the target $\zeta_\star$ of every transition rather than of the first alone. 
The analysis of Section~\ref{sec:controlling_error} does not cover the resulting mesh, and we use the rule below as a practical heuristic.

Let $\widehat{\zeta}_j \in (0,1]$ denote the CESS ratio of transition $j$ \citep{zhou2016toward}, which measures the weight degeneracy introduced by that transition alone.
For a short transition, the CESS ratio behaves as $\exp \{ - \mathrm{Var} ( \sum_{c=1}^C \phi^{(c)} ) \, \Delta_j^2 \}$, where the variance of the summed potential is taken across the particles before the transition; Appendix~\ref{app:adaptive_step_size} derives this approximation and the rule below. Assuming that this variance changes little from one transition to the next, the step size that should attain $\zeta_\star$ at the next transition is
\begin{equation} \label{eq:adaptive_rule_main}
    \widehat{\Delta}_{j+1} = \Delta_j \sqrt{ \frac{\log \zeta_\star}{\log \widehat{\zeta}_j} } \, .
\end{equation}
Both logarithms are negative, so the step shrinks when the observed CESS ratio falls below the target and grows when it exceeds it. Unlike the bisection used to choose the steps of adaptive tempering, this prediction requires no additional propagation of the particles. Appendix~\ref{app:adaptive_step_size} also describes the choice of the first step, the limits on how much the step can change between transitions, and other implementation details.

The rule differs from the adaptive mesh of \citet[Section~4.2.2]{chan2023dcf}, which takes each step from the asymptotic lower bound on the CESS of Bayesian Fusion in their Theorem~15. That bound is derived for Gaussian subposteriors, holds in the limit $N \to \infty$ and $\Delta_j \to 0$, and depends on the particles only through the average squared distance of their shard components from the subposterior means, rather than through the spread of the potentials across particles. It bounds the CESS of the BF weights, whose unbiased estimators add variance, and can therefore be conservative for the ABF weights. Appendix~\ref{app:tuning_calcs_bf} computes the steps given by this rule along ABF runs on the logistic regression of Section~\ref{sec:logistic_comparison}, and finds them about an order of magnitude smaller than those given by Equation~\eqref{eq:adaptive_rule_main}.

In practice, the choice of $\zeta_\star$ is relatively simple. Values between $0.5$ and $0.99$ worked well in our experiments, with higher values needed for harder problems. Increasing $\zeta_\star$ from $0.5$ to $0.9$ roughly doubled the runtime, and increasing it from $0.9$ to $0.99$ roughly tripled it. 
Such an increase is affordable when, as in our examples, a single run takes seconds or minutes. 
Since the discretisation error grows with the step size, we recommend starting from $\zeta_\star = 0.9$. If the step size of a fusion collapses, so that it takes many more steps than the others, the particles have most likely entered a region where $\phi^{(c)}$ changes very quickly, and $\zeta_\star$ should be increased, for example to $0.99$. On the proportional hazards model of Section~\ref{sec:ph}, a particularly pathological model, values below $0.95$ caused a bias.

\subsection{Efficient simulation} \label{sec:bridge_transition_summary}

We use three devices to make ABF faster in practice: a factorised form of the bridge transitions, preconditioning, and divide-and-conquer fusion. The first is an exact rewriting of the transitions and is always used; the other two, both from \cite{chan2023dcf}, are used in the more difficult examples of Section~\ref{sec:experiments}.

The Gaussian transition in Equation \eqref{eq:transitions_theta} can be simulated without explicitly forming or factorising the covariance matrix in Equation \eqref{eq:transitionsS}. Following Corollary~2 of \cite{dai2023bf}, the transition admits a representation in terms of a shared and a shard-specific Gaussian component. For a transition from time $s$ to time $t$, let 
$\alpha = (T-t)/(T-s)$, so that $1 - \alpha =(t-s)/(T-s)$.
Then for particle $i$, shard $c$, and coordinate $r$,
\begin{equation} \label{eq:efficient_bridge_update_summary}
    \theta_{t,i,r}^{(c)}
    =
    \alpha \theta_{s,i,r}^{(c)}
    +
    (1 - \alpha) \bar{\theta}_{s,i,r}
    +
    \sqrt{t-s}\sqrt{C^{-1}(1 - \alpha)}\,\chi_{i,r}
    +
    \sqrt{t-s}\sqrt{\alpha}\,\psi_{i,r}^{(c)} \, ,
\end{equation}
with
$\bar{\theta}_{s,i,r}=C^{-1} \sum_{c=1}^C \theta_{s,i,r}^{(c)}$,
where $\chi_{i,r}$ is standard normal and shared across shards for a given particle and coordinate, and the $\psi_{i,r}^{(c)}$ are independent standard normal variables, independent also of $\chi_{i,r}$.
The shared $\chi_{i,r}$ term induces the required covariance between shards, while the $\psi_{i,r}^{(c)}$ terms provide the shard-specific variation.

When a subposterior is strongly correlated, or has very different scales in different coordinates, the Brownian motion in Equation~\eqref{eq:sde}, whose covariance is the identity, is poorly matched to its shape. Following \cite{chan2023dcf}, the Brownian motion of shard $c$ can instead be given a covariance matrix $\Lambda_c$ that approximates the covariance of $f_c$. We set $\Lambda_c$ to the empirical covariance matrix $\hat{Q}^{(c)}$ of the draws from shard $c$, or to its diagonal. The bridge transition then becomes, in vector form,
\begin{equation} \label{eq:efficient_bridge_update_precond}
    \theta_{t,i}^{(c)} = \alpha \theta_{s,i}^{(c)} + (1 - \alpha) \tilde{\theta}_{s,i} + \sqrt{t-s} \sqrt{1 - \alpha} \, \bar{\Lambda}^{1/2} \chi_i + \sqrt{t-s} \sqrt{\alpha} \, \Lambda_c^{1/2} \psi_i^{(c)} \, ,
\end{equation}
where
\begin{equation}
    \bar{\Lambda} := \Big ( \sum_{c=1}^C \Lambda_c^{-1} \Big )^{-1} \, , \qquad \tilde{\theta}_{s,i} := \bar{\Lambda} \sum_{c=1}^C \Lambda_c^{-1} \theta_{s,i}^{(c)} \, ,
\end{equation}
$\chi_i$ and $\psi_i^{(c)}$ are independent $\mathcal{N}(0, I_d)$ vectors, and $\Lambda^{1/2}$ denotes the Cholesky factor of a matrix $\Lambda$. The mean across shards is thus replaced by an average weighted by the precisions $\Lambda_c^{-1}$, and the noise terms are scaled by $\bar{\Lambda}^{1/2}$ and $\Lambda_c^{1/2}$. When $\Lambda_c = I_d$ for all $c$, we have $\bar{\Lambda} = C^{-1} I_d$ and $\tilde{\theta}_{s,i} = \bar{\theta}_{s,i}$, and Equation~\eqref{eq:efficient_bridge_update_precond} reduces to Equation~\eqref{eq:efficient_bridge_update_summary}. Equation~\eqref{eq:efficient_bridge_update_precond} is Corollary~7.3.2 of \citet{chan2022thesis}. The initial weight $\rho_0$ changes accordingly, as detailed in Algorithm~\ref{alg:approx-bayesian-fusion}, and the functions $\phi^{(c)}$ become
\begin{equation} \label{eq:phi_precond}
    \phi^{(c)}(\theta) := \frac{1}{2} \nabla \log f_c(\theta)^\top \Lambda_c \nabla \log f_c(\theta) + \frac{1}{2} \operatorname{Tr} \big ( \Lambda_c \nabla^2 \log f_c(\theta) \big ) \, ,
\end{equation}
which reduce to Equation~\eqref{eq:phi} when $\Lambda_c = I_d$ \citep{chan2023dcf}. Algorithm~\ref{alg:approx-bayesian-fusion} in Appendix~\ref{app:complete_algo} gives the complete algorithm, with preconditioning, the rejection sampler of Section~\ref{sec:coalescence_time} and the adaptive step size of Section~\ref{sec:adaptive_step_size}.

Fusing all $C$ subposteriors in a single step becomes harder as $C$ grows, because an initial particle receives a non-negligible weight $\rho_0$ only if all of its $C$ components are close to each other. Divide-and-conquer fusion \citep{chan2023dcf} instead fuses the subposteriors two at a time along a binary tree: the $C$ subposteriors are fused in pairs into $C/2$ intermediate posteriors, each represented by the particles returned by its fusion, these are fused in pairs into $C/4$, and so on, until a single fusion returns the full posterior. Each of the $C - 1$ fusions then involves only two components, and the fusions at the same level of the tree are independent of each other. Each fusion starts from the weighted particles returned by its two children: the preconditioners $\Lambda_c$ are computed from their weighted moments, the coalescence time is calibrated afresh, and the rejection sampler of Section~\ref{sec:coalescence_time} draws each component of a proposal from the particles of the corresponding child, according to their weights. Our implementation requires $C$ to be a power of two.

\section{Experiments and examples} \label{sec:experiments}

In this section we compare ABF with other methods on four examples. The first is a small Bayesian logistic regression, on which we compare ABF with BF and study how the accuracy and computational cost of both depend on the step size. The remaining examples are a five-dimensional Rosenbrock distribution, a logistic regression on the nycflights13 data, and a proportional hazards model, which test ABF on strongly non-Gaussian targets and on larger datasets. In these examples, ABF uses the adaptive step size of Section~\ref{sec:adaptive_step_size}, together with the preconditioning and divide-and-conquer construction of Section~\ref{sec:bridge_transition_summary}. We compare ABF with CMC in all examples and, on the nycflights13 data, with the results that \citet{chan2023dcf} report for Divide-and-Conquer Bayesian Fusion and for the other methods of their Figure~7.

Consensus Monte Carlo \citep{scott2016cmc} combines the samples from the subposteriors by weighted averaging. Let $\theta_i^{(c)}$ be the $i$-th draw from shard $c$, and let $\hat{Q}^{(c)}$ be the empirical covariance matrix of the draws from that shard. CMC combines corresponding draws according to
\begin{equation}
    \theta_i^{\mathrm{CMC}} = \left( \sum_{c=1}^C \bigl(\hat{Q}^{(c)}\bigr)^{-1} \right)^{-1} \sum_{c=1}^C \bigl(\hat{Q}^{(c)}\bigr)^{-1} \theta_i^{(c)} \, .
\end{equation}
Thus, a subposterior with a smaller covariance receives greater weight. When the subposteriors are Gaussian and their covariance matrices are known, this aggregation produces exact draws from their product. With estimated covariance matrices, it is asymptotically exact as the number of subposterior draws increases. For non-Gaussian subposteriors, however, the weighted average may fail to preserve features such as skewness, curved dependence, and heavy tails. The aggregation step is computationally inexpensive and highly scalable, but CMC has no tuning parameter that directly trades additional computation for a reduction in this aggregation error. Increasing the number of draws reduces Monte Carlo error, but does not generally remove the error caused by non-Gaussianity.

We now assess the accuracy using the integrated absolute distance of \cite{chan2023dcf}. For one run of a method, let $\widehat f_{r}$ be the kernel density estimate of marginal $r$ obtained from its output, and let $\widehat f_{\mathrm{ref},r}$ be the corresponding estimate obtained from the full-data reference sample. We define
\begin{equation}\label{eq:logit_iad}
    \operatorname{IAD} = \frac{1}{2d} \sum_{r=1}^d \int_{\mathbb{R}} \left| \widehat f_{r}(y) - \widehat f_{\mathrm{ref},r}(y) \right| \, \mathrm{d}y \, .
\end{equation}
For each marginal, the integral divided by two is the total-variation distance between the two estimated densities. Hence, the IAD is the average of these distances over the $d$ one-dimensional marginals. The densities are evaluated on a common grid, using bandwidths selected from the reference sample.

Each method is run for ten replicas, except on the Rosenbrock distribution, where we show a single run. For ABF and BF, the runtime of a replica includes the rejection-sampling initialisation and the SMC run, and excludes the generation of the shard and reference samples. Following \citet{chan2023dcf}, on the nycflights13 data and on the proportional hazards model all ABF replicas use the same shard draws, whereas each CMC replica uses its own independent shard draws, and its runtime includes only the aggregation step. The logistic regression of Section~\ref{sec:logistic_comparison} is a simpler example, whose shard draws are a single Random-Walk Metropolis run per shard, without a target effective sample size. There, all methods use the same shard draws, and each CMC replica aggregates them after permuting the draws of each shard at random, which is included in its runtime. Throughout the experiments, we set the target rejection-sampling acceptance probability to $a_0 = 0.01$ and the resampling threshold to $\tau = 0.5$.

\subsection{Comparison on Bayesian logistic regression}
\label{sec:logistic_comparison}

We compare ABF with BF and CMC on a simulated Bayesian logistic regression. For observation $i$, let $x_i = (1,z_i^\top)^\top \in \mathbb{R}^6$, where the five entries of $z_i$ are independent standard Gaussian random variables. The response is generated according to
\begin{equation}\label{eq:logit_model}
    p_i = \frac{\exp(x_i^\top \beta)}{1 + \exp(x_i^\top \beta)} \, , \qquad Y_i \mid x_i, \beta \sim \operatorname{Bernoulli}(p_i) \, ,
\end{equation}
where $\beta = (\beta_1,\ldots,\beta_6)^\top$, $\beta_j = (-1)^j / \sqrt{6}$, and the prior on $\beta_j$ is $\propto 1, \forall j$. The data are distributed over $C = 4$ shards, with $500$ observations per shard and $m = 2{,}000$ observations in total. For each shard, a Random-Walk Metropolis sampler provides $N = 10{,}000$ subposterior draws, and the reference sample consists of $200{,}000$ draws from a Random-Walk Metropolis run on the full data. With a flat prior, the first three derivatives of each $\log f_c$ are bounded, so that $\phi^{(c)}$ is bounded below and Lipschitz. Since ABF runs here on a uniform mesh, with a single fusion and no preconditioning, this example satisfies the assumptions of Theorem~\ref{th:bound_strong}.

To study the effect of the step size, we run ABF and BF on uniform meshes with $\Delta = T/n$, where $n$ is chosen by the bisection of Section~\ref{sec:adaptive_step_size} for several target $\zeta_\star$. Both methods use the rejection sampler of Section~\ref{sec:coalescence_time}, which calibrates the coalescence time $T = 0.0353$ from the shard draws alone and provides the initial particles.
We use $\zeta_\star \in \{0.1, 0.2, 0.3, 0.4\}$ for BF and $\zeta_\star \in \{0.1, 0.2, 0.5, 0.8, 0.9, 0.99\}$ for ABF. Appendix~\ref{app:logit_details} reports the resulting meshes and the accuracy and runtime of each configuration.

\begin{figure}[!htb]
    \centering
    \includegraphics[width=1\linewidth]{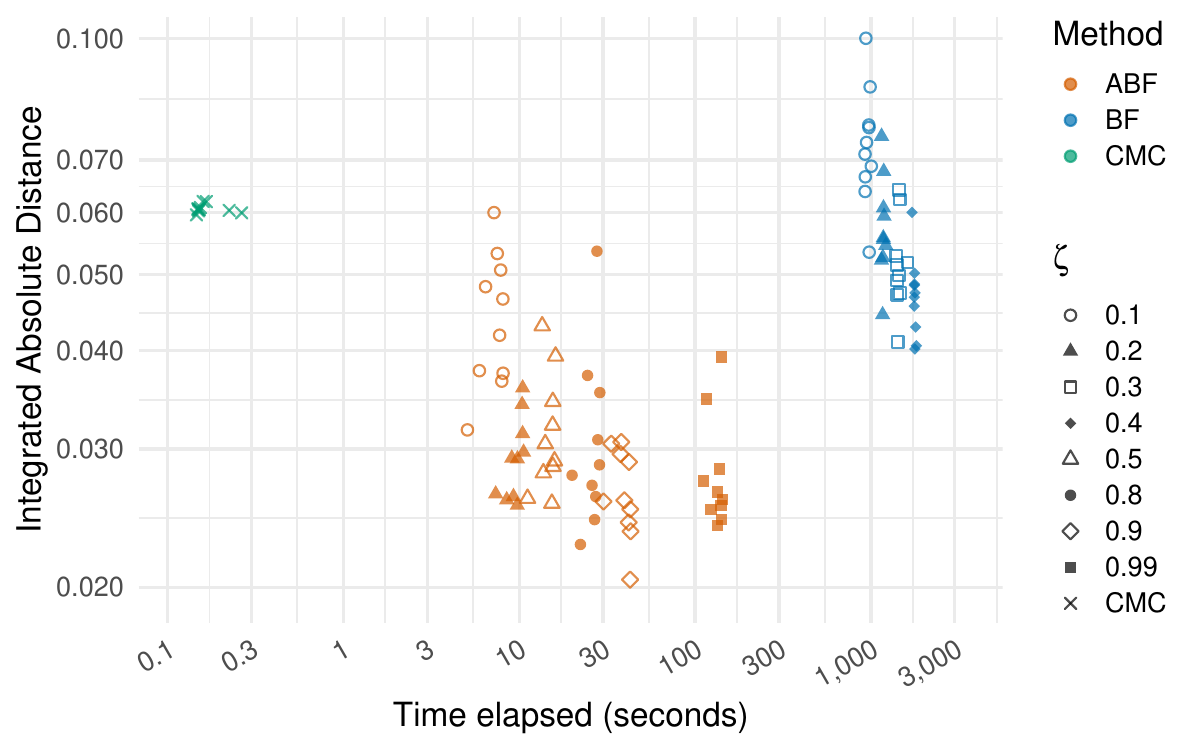}
    \caption{IAD against runtime per replica for the $6$-dimensional Bayesian logistic regression. Each point is one of ten replicas; colours indicate the method and shapes the target $\zeta_\star$. Both axes are logarithmic, and points nearer the lower-left corner are preferable.}
    \label{fig:logistic_iad_compare}
\end{figure}
Figure~\ref{fig:logistic_iad_compare} shows that, on average, every ABF configuration is both more accurate and faster than every BF configuration. 
The IAD of BF decreases as $\zeta_\star$ increases and the mesh is refined. 
As the mesh becomes finer, the weights become less variable and hence the effective particle population becomes larger.

The IAD of ABF is largest on the coarsest mesh, obtained with $\zeta_\star = 0.1$, decreases by about a third at $\zeta_\star = 0.2$, and then plateaus. We attribute this plateau to errors that do not depend on the mesh: the Monte Carlo error of a run with $N = 10{,}000$ particles; the error of the shard draws, which all replicas share and which come from Random-Walk Metropolis chains; and the error of the reference sample and of the kernel density estimates in the IAD. Once the discretisation bias of Section~\ref{sec:controlling_error} falls below these errors, a finer mesh only increases the cost.

These results reflect the different ways in which the two methods treat the path integral in the fusion weights. BF replaces it with an unbiased random estimator, so that the method is exact and its particle approximation converges to the posterior as $N \to \infty$, but the estimator adds variance to the weights. ABF replaces it with a deterministic quadrature rule, which introduces a bias but adds no randomness of its own. As a consequence, the ABF weights are much less variable than the BF weights, and the effective particle population is much larger for the same number of particles; Appendix~\ref{app:logit_details} compares the two on a common mesh.

CMC is the fastest method, but ABF is more accurate at every $\zeta_\star$. CMC does well on this example because the subposteriors and the target are all close to Gaussian.
The next examples show that the CMC approximation departs significantly from the target when the target is far from Gaussian.

\subsection{Rosenbrock distribution, \texorpdfstring{$5d$}{5d}-``Banana''} \label{sec:rosenbrock}

The Rosenbrock distribution is a common benchmark for Monte Carlo algorithms because its nonlinear dependence produces strongly curved, ``banana''-shaped contours \citep{pagani2022rosenbrock}. We use the five-dimensional hybrid Rosenbrock density
\begin{equation}
    f(\theta) \propto \exp \big \{
    - a ( \theta_1 - \mu )^2 - b_{12} ( \theta_2 - \theta_1^2 )^2 - b_{13} ( \theta_3 - \theta_2^2 )^2 - b_{22} ( \theta_4 - \theta_1^2 )^2 - b_{23} ( \theta_5 - \theta_4^2 )^2
    \big \} \, ,
\end{equation}
where $\theta = (\theta_1, \ldots, \theta_5)^\top$. 
We set $\mu = 0$, $a = 2$, and $b_{12} = b_{13} = b_{22} = b_{23} = 4$. The target is split into $C = 4$ subposteriors of the same form, with $a = 0.5$ and $b_{12} = b_{13} = b_{22} = b_{23} = 1$, so that their product is $f$. To compare the tails of the distribution reliably, we draw $10^6$ independent samples from each subposterior, and use $10^6$ independent samples from $f$ as the reference. ABF uses divide-and-conquer fusion with full preconditioning (Section~\ref{sec:bridge_transition_summary}), and the adaptive step size of Section~\ref{sec:adaptive_step_size} with $\zeta_\star = 0.9$.

\begin{figure}[!htb]
    \centering
    \includegraphics[width=1\linewidth]{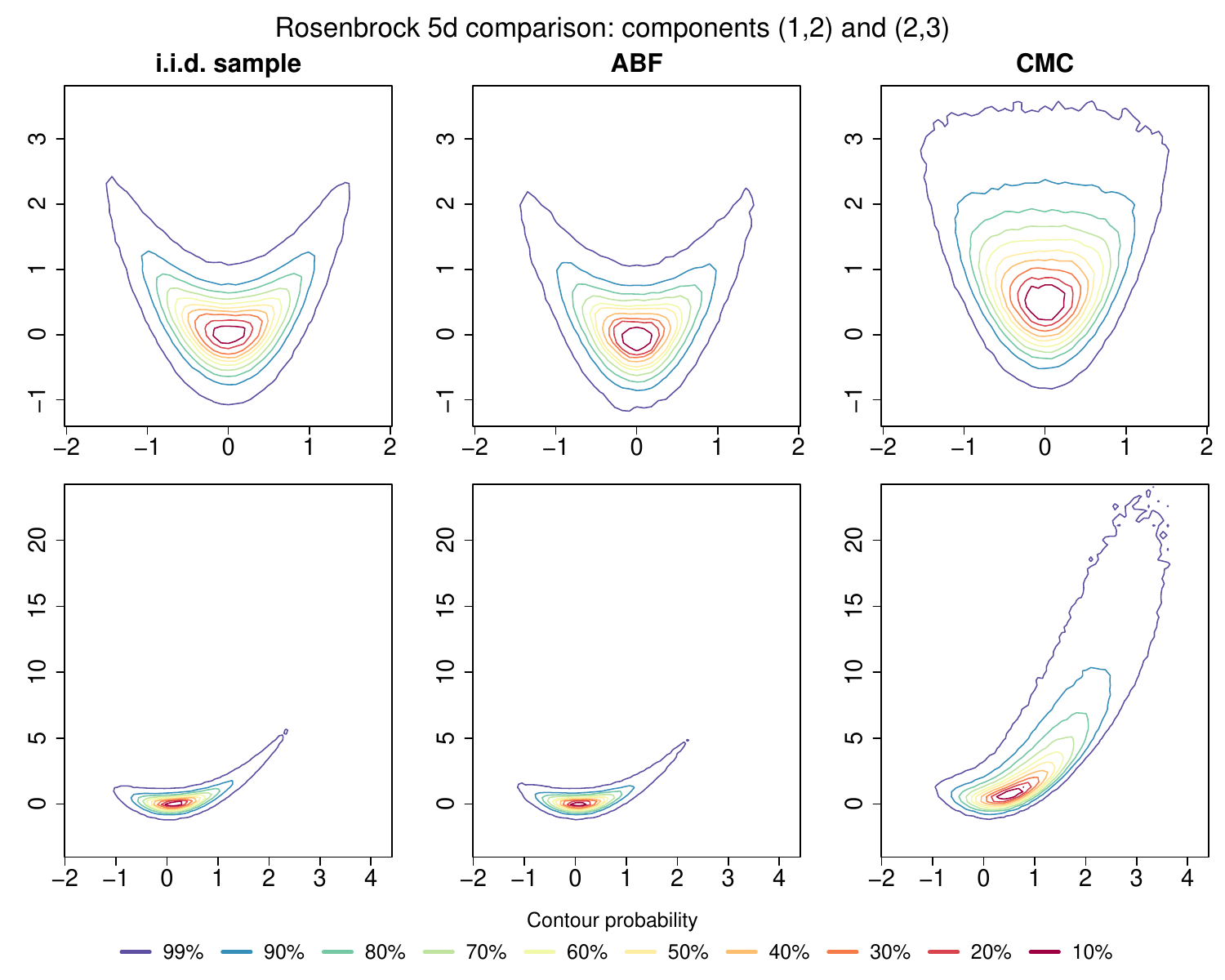}
    \caption{Two-dimensional marginals of components $(1,2)$, first row, and $(2,3)$, second row, of the five-dimensional hybrid Rosenbrock distribution, for an i.i.d.\ sample, ABF and CMC. ABF reproduces the curved dependence and the tails of the target, whereas CMC loses the curvature and shifts the mode.}
    \label{fig:rosen5d_cmc_compare}
\end{figure}

Figure~\ref{fig:rosen5d_cmc_compare} shows that ABF preserves the main features of the target: the curved dependence structure of the Rosenbrock distribution in both marginals, the position of the mode, and the length and position of the tails. The only visible difference is that the outer contour of ABF lies slightly inside that of the i.i.d.\ sample.

In contrast, CMC fails to capture the curved structure of the target, producing an upward shift of the mode and mass, and substantially greater dispersion for both pairs of components examined.

\subsection{Logistic regression on the nycflights13 data} \label{sec:nycflights}

We fit the logistic regression of \citet[Section 5.4.2]{chan2023dcf} to the nycflights13 data \citep{wickham2021nycflights13}, which record the flights that departed from New York City in 2013. 
After removing the flights with missing values, $m = 327{,}346$ flights remain. The model is given by Equation~\eqref{eq:logit_model} with $\beta \in \mathbb{R}^{d}$, $d=21$, and the prior is $\beta \sim \mathcal{N}(0, I_{d})$, so that the fractional prior of each shard is $\mathcal{N}(0, C I_{d})$, as in \citet{chan2023dcf}. 

The flights are shuffled and divided among $C = 64$ shards, so that each shard contains $5{,}114$ or $5{,}115$ flights. In $41$ shards no flight is operated by the carrier OO, and in one shard no flight is operated by the carrier HA, so that in these shards the coefficient of that carrier is determined by the fractional prior alone, resulting in noticeably different subposteriors. For each shard, Stan \citep{carpenter2017} provides the subposterior draws, and the reference sample consists of $256{,}000$ draws from $64$ Stan chains run on the full data; Appendix~\ref{app:nycflights_details} gives the settings.
ABF uses divide-and-conquer fusion on a binary tree with full preconditioning (Section~\ref{sec:bridge_transition_summary}), and the adaptive step size of Section~\ref{sec:adaptive_step_size}. Following the implementation in \citet{chan2023dcf}, the $m$ flights take only $461$ distinct values of $(x_i, Y_i)$ pairs, so the log-likelihood of each shard and its derivatives are evaluated as sums over these values, weighted by the number of flights of the shard that take each value, which leaves the target unchanged. We run ABF with $\zeta_\star \in \{0.5, 0.6, 0.7, 0.8, 0.9, 0.99\}$ and $N \in \{1{,}000, 10{,}000, 30{,}000\}$ particles, and CMC with the same values of $N$ as the number of draws per shard. The aggregation of CMC runs on a single core, as in \citet{chan2023dcf}. Each fusion uses $64$ threads, and the runtime of an ABF replica is the wall time of the whole tree. This matches the original experiment of \citet{chan2023dcf}, which was run on a $64$-core virtual machine \citep[Appendix A]{chan2022thesis}.

\begin{figure}[!htb]
    \centering
    \includegraphics[width=1\linewidth]{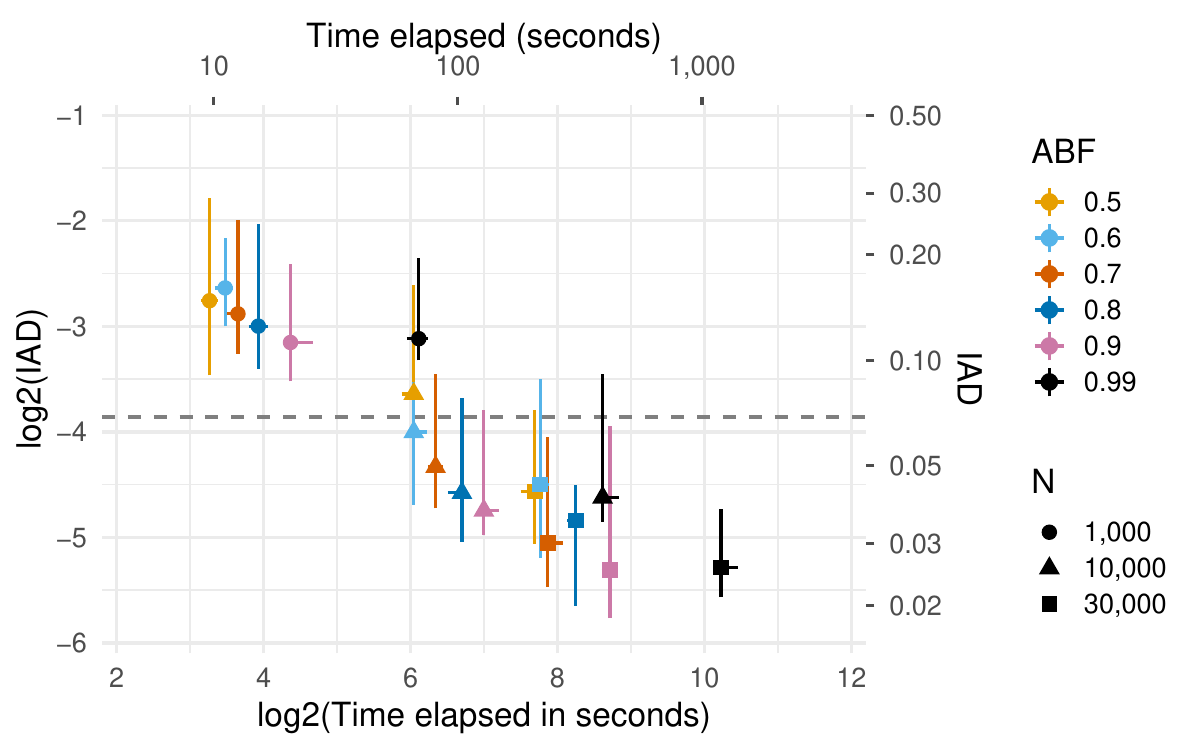}
    \caption{IAD against runtime of ABF for the logistic regression on the nycflights13 data, on the $\log_2$ scales of \citet[Figure 7]{chan2023dcf}. Colours indicate the target $\zeta_\star$, and marker shapes the number of particles $N$. The markers indicate the median runtime and the median IAD over ten replicas, and the bars extend from the minimum to the maximum. As in \citet[Figure 7]{chan2023dcf}, the dashed horizontal line marks the smallest mean IAD of CMC amongst all $N$; Figure~\ref{fig:nycflights_iad_compare_cmc} in Appendix~\ref{app:nycflights_details} also shows the CMC markers. Points nearer the lower-left corner are preferable.}
    \label{fig:nycflights_iad_compare}
\end{figure}

Figure~\ref{fig:nycflights_iad_compare} shows that the IAD of ABF decreases with $N$: for $N = 1{,}000$ the median IAD lies above the dashed line, which marks the smallest mean IAD of CMC amongst all $N$, while for $N = 30{,}000$ it lies between $0.025$ and $0.045$, depending on $\zeta_\star$, below the dashed line. This accuracy comes at a higher cost: for $N = 30{,}000$, the median runtimes are between $3$ and $7$ minutes for $\zeta_\star \leq 0.9$, and $20$ minutes for $\zeta_\star = 0.99$. CMC is much faster, with median runtimes between $0.12$ and $3.4$ seconds, but its median IAD levels off at about $0.07$ as $N$ increases (Appendix~\ref{app:nycflights_details}).

These results are consistent with Figure~7 of \citet{chan2023dcf}, which uses the same data and $C = 64$: the smallest mean IAD of CMC amongst all $N$ that they report agrees with ours, marked by the dashed line in Figure~\ref{fig:nycflights_iad_compare}. In their figure, CMC is both the most accurate and the fastest of the approximate methods, ahead of the kernel density and Weierstrass methods, so ABF with $N = 30{,}000$ is also more accurate than these. Their most accurate run of Divide-and-Conquer Bayesian Fusion reaches an IAD of about $0.025$ after about $18$ hours, while ABF with $\zeta_\star = 0.9$ and $N = 30{,}000$ reaches the same median IAD in about $7$ minutes, about two orders of magnitude faster. The hardware and the implementations differ, so this indicates only the order of magnitude of the difference.

\FloatBarrier

\subsection{Proportional hazards model} \label{sec:ph}

We run ABF on a proportional hazards model estimated on the \texttt{hds} data set of the R package \texttt{dynamichazard}, which contains records of the hard drives in Backblaze data centres \citep{backblaze-data, christoffersen2021, hds-doc}. 
We chose this data set because it combines left truncation and right censoring with group effects: drives enter the data at different ages, most have not failed by the end of the observation period, and the hazard varies between drive families. Together, these features produce a posterior with some strongly non-Gaussian shapes.

We keep one record per drive and only the $G = 4$ drive families, out of $79$, with the most drives. With only four families, the spread of the family effects is weakly identified. All drives of a family also have the same capacity, so the effect of capacity is identified only through the prior on the family effects. We chose this design deliberately. Appendix~\ref{sec:ph_model_prior} describes the resulting data. 

We fit a Weibull proportional hazards model with a random effect for each family. The hazard of drive $i$ at age $t$, measured in years of operation, is
\begin{equation}\label{eq:ph_hazard}
    h_i(t) = \gamma t^{\gamma - 1} \exp(\beta_0 + \beta_1 x_i + u_{g_i}) \, , \qquad u_g \mid \omega \overset{\mathrm{iid}}{\sim} \mathcal{N}(0, \omega^2) \, ,
\end{equation}
where $\gamma > 0$ is the Weibull shape, $x_i$ is the standardised log capacity of drive $i$, $g_i \in \{1, \ldots, G\}$ is its family, and $\omega$ is the standard deviation of the family effects. The lifetimes are left-truncated and right-censored. ABF and CMC operate on the $d = 8$ unconstrained coordinates $q = (\beta_0, \beta_1, \log \gamma, \log \omega, u_1, \ldots, u_4)^\top$. Appendix~\ref{sec:ph_model_prior} gives the likelihood, the prior, and the fractional prior of each shard.

The drives are divided among $C = 16$ shards, with the drives of each family and failure status spread as evenly as possible. 
For each shard, Stan \citep{carpenter2017} provides the subposterior draws, and the reference sample consists of $256{,}000$ draws from $64$ Stan chains run on the full data; Appendix~\ref{sec:ph_model_prior} gives the settings.
ABF uses divide-and-conquer fusion on a binary tree with $15$ fusions, full preconditioning (Section~\ref{sec:bridge_transition_summary}), the adaptive step size of Section~\ref{sec:adaptive_step_size} with $\zeta_\star \in \{0.95, 0.99\}$, and $N \in \{1{,}000, 10{,}000, 30{,}000\}$ particles. CMC uses the same values of $N$ as the number of draws per shard. Each ABF fusion uses $64$ threads, and the runtime of an ABF replica is the wall time of the whole tree. The IAD is computed on the parameters of Equation~\eqref{eq:ph_hazard}, $(\beta_0, \beta_1, \gamma, \omega, u_1, \ldots, u_4)$.

Family $4$ has only $17$ failures, about one per shard, so the data determine $u_4$ weakly, and its spread is set by $\omega$. This produces a funnel between $u_4$ and $\log \omega$, whose neck lies at values of $u_4$ and $\omega$ close to $0$. As $\omega \to 0$ with the family effects close to $0$, $\phi^{(c)} \to - \infty$, so $\phi^{(c)}$ is not bounded from below (see Appendix~\ref{sec:ph_model_prior}).

A lower bound $\Phi^{(c)}$ is assumed only in the error analysis of Section~\ref{sec:controlling_error} (Assumption~\ref{as:phi_min}), whose error bounds therefore do not apply to this model. ABF does not need the value of $\Phi^{(c)}$, because $\Phi^{(c)}$ only rescales the weights and cancels when they are normalised (see Section~\ref{sec:ABF}). Without a lower bound, however, the ABF weights are not bounded, and some care may be required when interpreting the results below. BF, in contrast, relies on lower bounds on $\phi^{(c)}$ in its layered Brownian bridge construction, and therefore cannot be applied to this model.

\begin{figure}[!htb]
    \centering
    \includegraphics[width=1\linewidth]{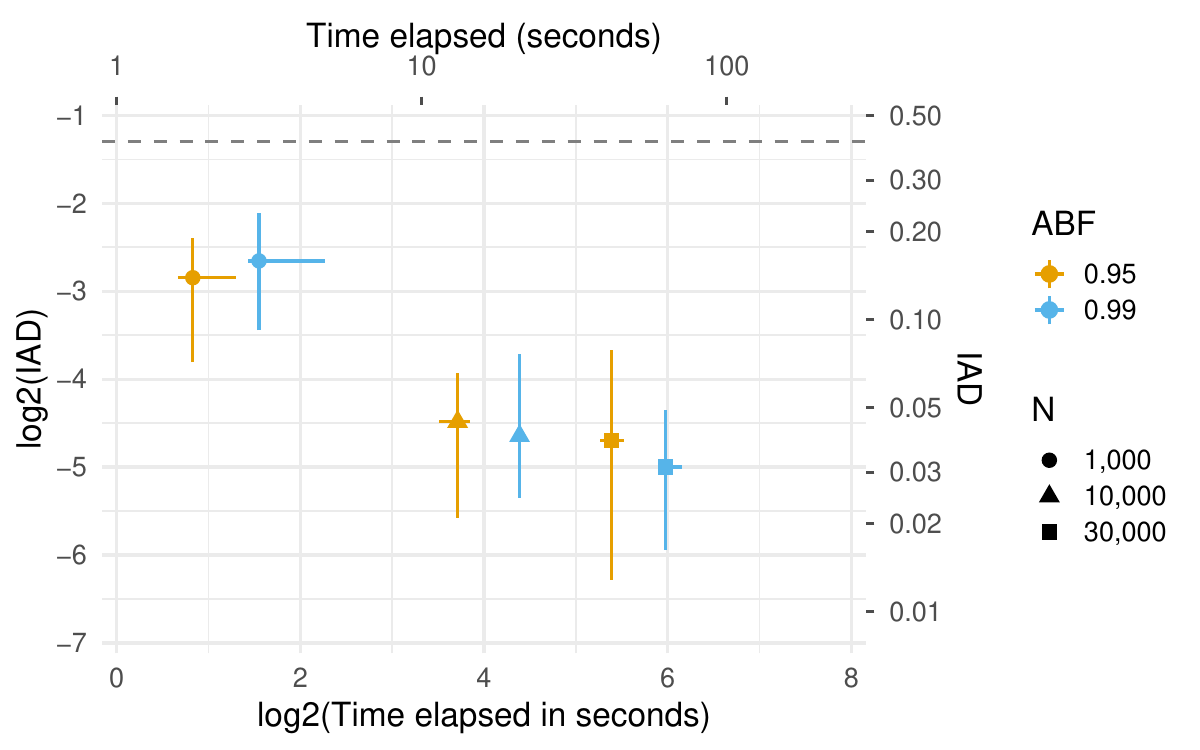}
    \caption{IAD against runtime of ABF for the proportional hazards model, for $\zeta_\star \in \{0.95, 0.99\}$, on the $\log_2$ scales of \citet[Figure 7]{chan2023dcf}. Colours indicate the target $\zeta_\star$, and marker shapes the number of particles $N$. The markers indicate the median runtime and the median IAD over ten replicas, and the bars extend from the minimum to the maximum. The dashed horizontal line marks the smallest mean IAD of CMC amongst all $N$. Points nearer the lower-left corner are preferable.}
    \label{fig:ph_iad}
\end{figure}

Figure~\ref{fig:ph_iad} shows the expected behaviour. For both values of $\zeta_\star$, the median IAD of ABF decreases as $N$ increases, because the Monte Carlo error decreases.
For $N = 1{,}000$, the Monte Carlo error dominates, and the two values of $\zeta_\star$ give similar errors. For $N = 30{,}000$, the median IAD is $0.039$ for $\zeta_\star = 0.95$ and $0.031$ for $\zeta_\star = 0.99$, close to the IAD of about $0.03$ that ABF reaches on the logistic regressions of Sections~\ref{sec:logistic_comparison} and~\ref{sec:nycflights}, although this target is much further from Gaussian. The IAD of CMC, in contrast, lies between $0.40$ and $0.43$ for every $N$, much worse than on the logistic regressions, as expected.
While CMC is very fast, at most $0.5$ seconds, the price of the accuracy of ABF is modest: for $N = 30{,}000$, the median runtime of ABF is $42$ seconds for $\zeta_\star = 0.95$ and $63$ seconds for $\zeta_\star = 0.99$.

The values of $\zeta_\star$ used above are higher than in the previous examples. With $\zeta_\star \in \{0.85, 0.9\}$, which we exclude from Figure~\ref{fig:ph_iad}, the longer steps caused a bias towards the neck of the funnel. There, $\phi^{(c)}$ behaves like $- \omega^{-2}$, and the trapezium rule, which evaluates $\phi^{(c)}$ only at the grid points, gives a very large weight to a particle that lands at a grid point deep in the neck, so that a few such particles dominate the sample and pull $\omega$ towards $0$ (Appendix~\ref{sec:ph_model_prior}). With $\zeta_\star = 0.85$ and $N = 30{,}000$, rerunning every fusion on a fixed grid obtained by splitting each step of an adaptive pilot run into four removes the error (Appendix~\ref{sec:ph_model_prior}), which confirms that it comes from the time discretisation: the criterion of Section~\ref{sec:adaptive_step_size} controls the variance of the weights over each step, but not this discretisation error. 

The IAD compares only the one-dimensional marginals. Figure~\ref{fig:ph_pairs} also compares two-dimensional marginals, one of which involves a quantity of practical interest calculated from the posterior sample: the median lifetime $m_0$ of a drive with $x_i = 0$ and $u_{g_i} = 0$. By Equation~\eqref{eq:ph_hazard}, the survival function of such a drive is $\exp(- e^{\beta_0} t^{\gamma})$, and setting it equal to $1/2$ at $t = m_0$ gives
$m_0 = ( \log 2 / e^{\beta_0} )^{1 / \gamma}$.

\begin{figure}[!htb]
    \centering
    \includegraphics[width=0.9\linewidth]{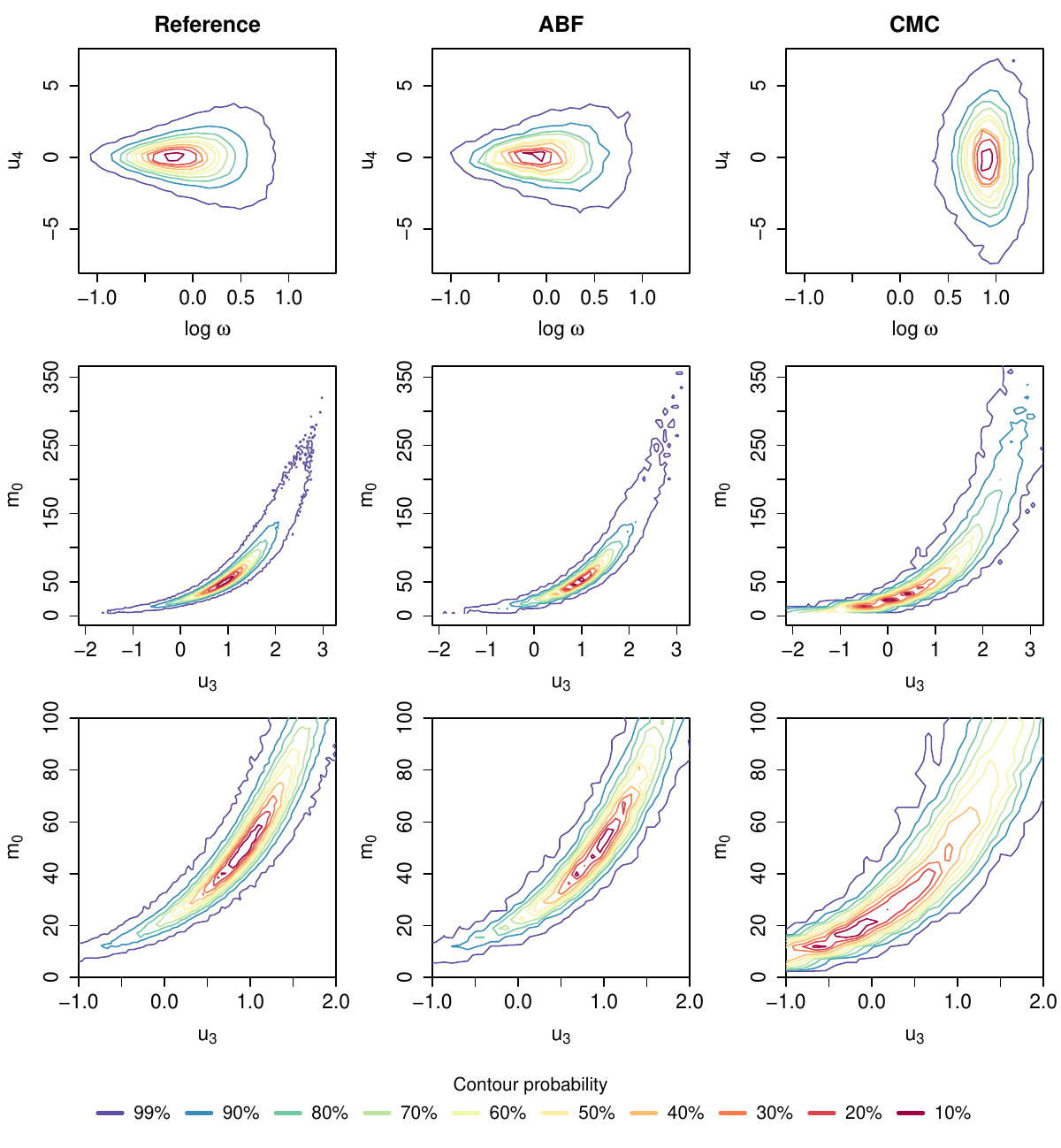}
    \caption{Two-dimensional marginals of the proportional hazards model for the reference sample, ABF and CMC. The first row shows the funnel between $\log \omega$ and $u_4$, a non-Gaussian geometry arising from real data: ABF reproduces its shape, whereas CMC misplaces the mode of $\log \omega$ and overstates the spread of $u_4$. The second row shows the same comparison for a quantity of practical interest, the median lifetime $m_0$,
    against $u_3$, and the third row magnifies the second around the mode. In the second row, the outer contours of CMC extend beyond the plotted region.}
    \label{fig:ph_pairs}
\end{figure}

Figure~\ref{fig:ph_pairs} compares the reference sample with the ABF replica for $\zeta_\star = 0.99$ and $N = 30{,}000$ whose IAD, $0.030$, is closest to the median, resampled according to its weights, and with the CMC replica with $N = 30{,}000$ and the same index. 
On the first row of the plot, the pair $(\log \omega, u_4)$ shows the funnel described above, while on the second row, the joint distribution of $(u_3, m_0)$ lies along a curved ridge with a long upper tail. ABF preserves both $2d$-marginals: the location and spread of $\log \omega$ and $u_4$, and the curvature of the ridge of $(u_3, m_0)$, the position of its mode, which the third row of Figure~\ref{fig:ph_pairs} shows magnified, and its upper tail all closely match the reference.

The subposteriors of $\log \omega$ and of $u_4$ are not Gaussian: each shard contains about one failure of family $4$, so each subposterior has the funnel geometry between $u_4$ and $\log \omega$. CMC misrepresents both the location and the uncertainty: it shifts the mode of $\log \omega$ well above that of the reference, so that their $95\%$ intervals barely overlap, and makes the marginal of $u_4$ more than twice as wide. For $(u_3, m_0)$, CMC moves the mode of $m_0$ to much shorter lifetimes, while its upper tail extends much further. 

\section{Conclusion} \label{sec:discussion}

In this work we introduced Approximate Bayesian Fusion (ABF), which adopts the sequential Monte Carlo structure of Bayesian Fusion (BF) but replaces the unbiased estimate of the path integral in the fusion weights, obtained through the layered Brownian bridge construction, by an inexpensive trapezium approximation. ABF therefore requires no model-specific bounds on $\phi^{(c)}$, and the weight update of each step only requires evaluating $\phi^{(c)}$ at the two ends of the step. The approximation introduces a bias, which we showed to be at most of order $\sqrt{\Delta}$ in total variation distance when $\phi^{(c)}$ is bounded below and Lipschitz (Theorem~\ref{th:bound_strong}).

We also implemented ABF within the divide-and-conquer framework of \citet{chan2023dcf}, with preconditioned Brownian bridges and a step size chosen online by targeting the conditional effective sample size of each transition.

In our experiments, ABF substantially outperforms BF in accuracy as well as speed.
On the $21$-dimensional logistic regression of \citet{chan2023dcf}, with more than $300{,}000$ observations, ABF reached the IAD reported for Divide-and-Conquer Bayesian Fusion by \citet{chan2023dcf} about two orders of magnitude faster. On the proportional hazards model, ABF reproduced features of the posterior such as the funnel shapes that CMC misses, and reached an IAD comparable to that on the logistic regressions, whose posteriors are closer to Gaussian. Notably, the proportional hazards model admits no lower bound on $\phi^{(c)}$, so BF cannot be applied. 
This suggests that ABF may be applicable to a wider class of models than the theory covers, while maintaining high accuracy and low runtime, provided that the steps are short enough.

Two extensions could make ABF more efficient and fully automatic. First, the step-size criterion controls the variance of the weights over each step, but not the discretisation error. A criterion that also accounts for the variation of $\phi^{(c)}$ along the paths could prevent this bias. Second, ABF requires only the gradient of $\log f_c$ and the trace of $\Lambda_c$ multiplied by its Hessian, which enter $\phi^{(c)}$ through Equation~\eqref{eq:phi_precond}. In our implementation, these were coded by hand for each model, but automatic differentiation could provide them for any model whose log density is written in a differentiable language, such as Stan \citep{carpenter2017}. ABF could then be applied to a new model with no model-specific work beyond that of sampling the subposteriors.

\section*{Acknowledgements}

The authors would like to thank Shenggang Hu, Alastair Crossley, and Murray Pollock for their encouragement and fruitful discussions.
FP, AMJ, and GOR acknowledge financial support from UKRI grant EP/Y014650/1 as part of the ERC Synergy project OCEAN.
Computing facilities were provided by the Scientific Computing Research Technology Platform of the University of Warwick.

\bibliography{bib}

\appendix

\newpage
\FloatBarrier

\section{Tuning of the coalescence time via rejection sampling}
\label{app:rejection_acceptance_tuning}

In Section~\ref{sec:coalescence_time}, a tuple $\theta^{(1:C)}$, formed by drawing one particle independently from each subposterior, is accepted with probability $\exp \{ - \mathcal{E} / T \}$, where $\mathcal{E}$ is the energy of the tuple, and $T$ is chosen as the root $\hat T$ of $\widehat{a}(T) = a_0$, where $\widehat{a}(T)$ is the average acceptance probability of $M$ proposed tuples with energies $\mathcal{E}_1, \ldots, \mathcal{E}_M$. Algorithms~\ref{alg:tune-coalescence-time} and~\ref{alg:rejection-initialisation} summarise the calibration and the initialisation. They are written without preconditioning. With preconditioning, the mean $\bar{\theta}$ is replaced by the precision-weighted centre $\tilde{\theta}$, and the energy by $\mathcal{E} = \frac{1}{2} \sum_{c=1}^C (\theta^{(c)} - \tilde{\theta})^\top \Lambda_c^{-1} (\theta^{(c)} - \tilde{\theta})$, as in Algorithm~\ref{alg:approx-bayesian-fusion} of Appendix~\ref{app:complete_algo}. When the particles of a shard carry weights, as in the divide-and-conquer tree, each component of a proposal is drawn according to these weights.

We find $\hat T$ by bisection. The initial bracket needs care, because the scale of $T$ depends on the model, the dimension $d$, and the number of shards $C$, and can vary by orders of magnitude across examples. As the lower endpoint we take a small positive value, and as the upper endpoint the root obtained when every energy is replaced by their mean $\bar{\mathcal{E}} = M^{-1} \sum_{q=1}^{M} \mathcal{E}_q$, that is, the solution of $\exp \{ - \bar{\mathcal{E}} / T \} = a_0$,
\begin{equation}
    \overline{T} = \frac{\bar{\mathcal{E}}}{| \log a_0 |} \, .
\end{equation}
Proposition~\ref{prop:rejection_jensen} below shows that $\overline{T}$ lies above $\hat T$ whenever the energies are not all equal, and the two coincide when they are, so this interval always brackets the root. Our implementation doubles the upper endpoint if it does not, but only as a safeguard against numerical error.

\begin{algorithm}[t]
\caption{Calibration of the coalescence time}
\label{alg:tune-coalescence-time}
\begin{algorithmic}[1]
\Require Shard particle collections $\{ \theta_i^{(c)} \}_{i = 1, \ldots, N}^{c = 1, \ldots, C}$, calibration sample size $M$, target acceptance probability $a_0$.
\Ensure Calibrated coalescence time $\hat T$.
\For{$q = 1, \ldots, M$}
    \State Draw $\theta_q^{(c)}$ independently from the empirical particle distribution of shard $c$, for $c = 1, \ldots, C$.
    \State Set $\bar{\theta}_q \gets C^{-1} \sum_{c=1}^C \theta_q^{(c)}$.
    \State Set $\mathcal{E}_q \gets \frac{1}{2} \sum_{c=1}^C \| \theta_q^{(c)} - \bar{\theta}_q \|^2$.
\EndFor
\State Define $\widehat{a}(T) \gets M^{-1} \sum_{q=1}^M \exp \{ - \mathcal{E}_q / T \}$.
\State Find $\hat T$ by bracketing and bisection such that $\widehat{a}(\hat T) = a_0$.
\State \Return $\hat T$.
\end{algorithmic}
\end{algorithm}

\begin{algorithm}[t]
\caption{Rejection sampler for ABF initialisation}
\label{alg:rejection-initialisation}
\begin{algorithmic}[1]
\Require Shard particle collections $\{ \theta_i^{(c)} \}_{i = 1, \ldots, N}^{c = 1, \ldots, C}$, coalescence time $\hat T$, required number of accepted tuples $N$.
\Ensure Accepted initial particles $\{ \theta_{t_0,i}^{(1:C)} \}_{i=1}^N$ with $t_0 = 0$.
\State Set $i \gets 1$.
\While{$i \le N$}
    \State Draw a proposal $\theta^{(c)}$ independently from the empirical particle distribution of shard $c$, for $c = 1, \ldots, C$.
    \State Set $\bar{\theta} \gets C^{-1} \sum_{c=1}^C \theta^{(c)}$.
    \State Set $\mathcal{E} \gets \frac{1}{2} \sum_{c=1}^C \| \theta^{(c)} - \bar{\theta} \|^2$.
    \If{the proposal is accepted, with probability $\exp \{ - \mathcal{E} / \hat T \}$}
        \State Set $\theta_{t_0,i}^{(1:C)} \gets \theta^{(1:C)}$.
        \State Set $i \gets i + 1$.
    \EndIf
\EndWhile
\State \Return $\{ \theta_{t_0,i}^{(1:C)} \}_{i=1}^N$.
\end{algorithmic}
\end{algorithm}

\begin{proposition} \label{prop:rejection_jensen}
    Let $\mathcal{E}_1, \ldots, \mathcal{E}_M$ be non-negative energies that are not all equal, and let $\bar{\mathcal{E}} = M^{-1} \sum_{q=1}^{M} \mathcal{E}_q$. For $T > 0$, define
    \begin{equation}
        \widehat{a}(T) = \frac{1}{M} \sum_{q=1}^{M} \exp \{ - \mathcal{E}_q / T \} \, , \qquad \bar{a}(T) = \exp \{ - \bar{\mathcal{E}} / T \} \, .
    \end{equation}
    Then
    \begin{equation}
        \widehat{a}(T) > \bar{a}(T) \qquad \text{for all } T > 0 \, .
    \end{equation}
    Moreover, if $a_0 \in (0,1)$ and the equations $\widehat{a}(T) = a_0$ and $\bar{a}(T) = a_0$ have positive solutions $\hat T$ and $\overline{T}$, respectively, then $\hat T < \overline{T}$.
\end{proposition}

\begin{proof}
    Fix $T > 0$ and define the function $g_T(e) = \exp \{ - e / T \}$ of an energy $e \ge 0$. Since
    \begin{equation}
        g_T''(e) = T^{-2} \exp \{ - e / T \} > 0 \, ,
    \end{equation}
    $g_T$ is strictly convex. Therefore, by Jensen's inequality,
    \begin{align}
        \bar{a}(T)
        & = g_T \left ( \frac{1}{M} \sum_{q=1}^{M} \mathcal{E}_q \right ) \\
        & < \frac{1}{M} \sum_{q=1}^{M} g_T(\mathcal{E}_q) \\
        & = \widehat{a}(T) \, .
    \end{align}
    This proves the first claim.

    To compare the corresponding roots, note that $\widehat{a}$ is strictly increasing on $(0, \infty)$, since
    \begin{equation}
        \widehat{a}'(T) = \frac{1}{M} \sum_{q=1}^{M} \frac{\mathcal{E}_q}{T^2} \exp \{ - \mathcal{E}_q / T \} > 0 \, ,
    \end{equation}
    where strict positivity follows from the non-negativity of the energies and the assumption that they are not all equal. At $T = \overline{T}$, the first part gives
    \begin{equation}
        \widehat{a}(\overline{T}) > \bar{a}(\overline{T}) = a_0 = \widehat{a}(\hat T) \, .
    \end{equation}
    Since $\widehat{a}$ is strictly increasing, it follows that $\overline{T} > \hat T$.
\end{proof}

\section{Adaptive step-size selection}
\label{app:adaptive_step_size}

This appendix derives the step-size rule of Section~\ref{sec:adaptive_step_size} and describes the choice of the first step, the safeguards, and the offline variant. Throughout, $\Delta_j := t_j - t_{j-1}$ denotes the length of transition $j$, and $W_{j-1,i}$, for $i = 1, \ldots, N$, denote the normalised particle weights immediately before the transition, after any resampling performed at its start, so that $\sum_{i=1}^N W_{j-1,i} = 1$.

\subsection{Conditional effective sample size}

Write
\begin{equation} \label{eq:summed_potential}
    \varphi_{j,i} := \sum_{c=1}^C \phi^{(c)} \big ( \theta_{t_j,i}^{(c)} \big )
\end{equation}
for the summed potential of particle $i$ at time $t_j$. The trapezium rule of Algorithm~\ref{alg:abf-trapezium} adds
\begin{equation} \label{eq:adaptive_log_weight_main}
    \ell_{j,i} := \frac{\Delta_j}{2} \left ( \varphi_{j-1,i} + \varphi_{j,i} \right )
\end{equation}
to the negative log weight of particle $i$, so that transition $j$ reweights that particle by $\exp \{ - \ell_{j,i} \}$. We measure the degeneracy introduced by the transition through the CESS ratio of \citet{zhou2016toward},
\begin{equation} \label{eq:app_cess}
    \widehat{\zeta}_j := \frac{ \left ( \sum_{i=1}^N W_{j-1,i} \, \mathrm{e}^{- \ell_{j,i}} \right )^2 }{ \sum_{i=1}^N W_{j-1,i} \, \mathrm{e}^{- 2 \ell_{j,i}} } \, .
\end{equation}
This quantity lies in $(0,1]$, and $N \widehat{\zeta}_j$ is the conditional ESS on the usual particle scale. Unlike the ESS \citep{kong1994sequential} of the updated weights, Equation~\eqref{eq:app_cess} ignores any degeneracy inherited from earlier transitions and reports only what the current one contributed. Criteria of this form are standard for choosing the spacing between successive targets in adaptive SMC; see, for example, \cite{jasra2011inference, zhou2016toward} and \citet[Chapter~17]{chopin2020smc}.

\subsection{The local approximation} \label{app:local_approx}

Consider a particle drawn from the current weighted particle system, and let $\ell_j$, $\varphi_{j-1}$ and $\varphi_j$ denote its increment and its summed potentials at the two ends of the step, dropping the particle index $i$. Throughout this subsection, $\mathbb{E}$ and $\mathrm{Var}$ are taken over this draw and over the bridge noise of the transition. Let
\begin{equation}
    \zeta_j := \frac{ \mathbb{E} [ \mathrm{e}^{- \ell_j} ]^2 }{ \mathbb{E} [ \mathrm{e}^{- 2 \ell_j} ] }
\end{equation}
be the population version of Equation~\eqref{eq:app_cess}. Transition $j$ starts at $t_{j-1}$ and has length $\Delta_j$. We first show that $\ell_j$ is approximately proportional to $\Delta_j$, and then that $\log \zeta_j$ is approximately $- \mathrm{Var}(\ell_j)$. Together, these give a model for how $\zeta_j$ depends on $\Delta_j$, with a coefficient that depends on where the transition starts on the bridge.

Over a transition of length $\Delta_j$ the bridge moves a distance of order $\sqrt{\Delta_j}$. Under Assumption~\ref{as:phi_lip}, or under the Lipschitz-gradient Assumption~\ref{as:weak_error_regularity}, the potentials change by the same order, so
\begin{equation}
    \varphi_j = \varphi_{j-1} + \mathcal{O} \big ( \sqrt{\Delta_j} \big ) \, ,
\end{equation}
where the remainder is random, so the order is meant in mean square, i.e.\ $\mathbb{E} [ ( \varphi_j - \varphi_{j-1} )^2 ] = \mathcal{O}(\Delta_j)$. The two endpoint potentials in Equation~\eqref{eq:adaptive_log_weight_main} therefore agree to that order, and their average may be replaced by $\varphi_{j-1}$, giving
\begin{align}
    \ell_j
    & = \frac{\Delta_j}{2} (\varphi_{j-1} + \varphi_{j}) \\
    & = \frac{\Delta_j}{2} (\varphi_{j-1} + \varphi_{j-1} + \mathcal{O}(\sqrt{\Delta_j}) ) \\
    & = \varphi_{j-1} \Delta_j + \frac{\Delta_j}{2} \mathcal{O}(\sqrt{\Delta_j}) \\
    & = \varphi_{j-1} \Delta_j + \mathcal{O} \big ( \Delta_j^{3/2} \big ) \, . \label{eq:app_linear_log_weight}
\end{align}
Taking variances,
\begin{equation} \label{eq:app_variance_delta}
    \mathrm{Var} ( \ell_j ) = v_j \Delta_j^2 + \mathcal{O} \big ( \Delta_j^3 \big ) \, , \qquad v_j := \mathrm{Var}(\varphi_{j-1}) \, ,
\end{equation}
where, since $\varphi_{j-1}$ does not involve the bridge noise, $v_j$ is a variance over the weighted particle system before the step only.
The remainder in Equation~\eqref{eq:app_variance_delta} is of order $\Delta_j^3$ under Assumption~\ref{as:weak_error_regularity}: conditionally on the position of the particle at $t_{j-1}$, the Gaussian noise of the bridge averages out of the first-order term of $\varphi_j - \varphi_{j-1}$, so that only a drift of order $\Delta_j$ remains. Under Assumption~\ref{as:phi_lip} alone, a potential with a kink can leave a conditional mean of order $\sqrt{\Delta_j}$, and the remainder can then only be bounded by $\mathcal{O}(\Delta_j^{5/2})$.
The coefficient $v_j$ depends on the particle system at $t_{j-1}$, and hence on where the transition starts on the bridge, but not on $\Delta_j$. It measures the local difficulty of the bridge, and it is through $v_j$ that the appropriate step size changes from one transition to the next.

Next, $\zeta_j$ is unchanged by a common shift of the increments, because replacing $\ell_j$ with $\ell_j + c$ produces a factor $\mathrm{e}^{- 2 c}$ which appears in both the numerator and the denominator. Let $Y := \ell_j - \mathbb{E}[\ell_j]$ be the centred increment, so that $\mathbb{E}[Y] = 0$ and $\mathbb{E}[Y^2] = \mathrm{Var}(\ell_j)$. Taking $c = - \mathbb{E}[\ell_j]$,
\begin{equation}
    \zeta_j = \frac{ \mathbb{E} [ \mathrm{e}^{- Y} ]^2 }{ \mathbb{E} [ \mathrm{e}^{- 2 Y} ] } \, .
\end{equation}
By Equation~\eqref{eq:app_linear_log_weight}, $Y \approx \Delta_j ( \varphi_{j-1} - \mathbb{E}[\varphi_{j-1}] )$, which is $\Delta_j$ times a quantity that does not depend on $\Delta_j$ and has standard deviation $\sqrt{v_j}$, so $Y = \mathcal{O}(\Delta_j)$. Centring the increments therefore means that the expansion below requires only their spread, of order $\sqrt{v_j} \, \Delta_j$, to be small, and not their mean, which can be large when the potentials are large.

For $s \in \{ 1, 2 \}$, expanding the exponential to second order and taking expectations, using $\mathbb{E}[Y] = 0$ and $\mathbb{E}[Y^2] = \mathrm{Var}(\ell_j)$,
\begin{align}
    \mathbb{E} [ \mathrm{e}^{- s Y} ] & = \mathbb{E} \left [ 1 - s Y + \frac{s^2 Y^2}{2} + \mathcal{O}(\Delta_j^3) \right ] \\
    & = 1 + \frac{s^2}{2} \mathrm{Var}(\ell_j) + \mathcal{O}(\Delta_j^3) \, .
\end{align}
Since $\mathrm{Var}(\ell_j) = \mathcal{O}(\Delta_j^2)$ by Equation~\eqref{eq:app_variance_delta}, the expansion $\log(1 + x) = x + \mathcal{O}(x^2)$ applied with $x = \mathcal{O}(\Delta_j^2)$ gives
\begin{equation}
    \log \mathbb{E} [ \mathrm{e}^{- s Y} ] = \frac{s^2}{2} \mathrm{Var}(\ell_j) + \mathcal{O}(\Delta_j^3) \, .
\end{equation}
Taking logarithms in the expression for $\zeta_j$ and substituting this expansion with $s = 1$ and $s = 2$,
\begin{align}
    \log \zeta_j & = 2 \log \mathbb{E} [ \mathrm{e}^{- Y} ] - \log \mathbb{E} [ \mathrm{e}^{- 2 Y} ] \\
    & = \mathrm{Var}(\ell_j) - 2 \mathrm{Var}(\ell_j) + \mathcal{O}(\Delta_j^3) \\
    & = - \mathrm{Var}(\ell_j) + \mathcal{O}(\Delta_j^3) \, , \label{eq:app_log_cess_expansion}
\end{align}
and hence
\begin{equation} \label{eq:app_cess_variance}
    \zeta_j \approx \exp \left \{ - \mathrm{Var}(\ell_j) \right \} \, .
\end{equation}
The remainder in Equation~\eqref{eq:app_log_cess_expansion} comes from the third and higher moments of $Y$, and is small whenever the step is short, provided that these moments are finite. The approximation is exact when $\ell_j$ is Gaussian: then $Y \sim \mathcal{N}(0, \mathrm{Var}(\ell_j))$ and $\mathbb{E} [ \mathrm{e}^{- s Y} ] = \exp \{ s^2 \mathrm{Var}(\ell_j) / 2 \}$, so that $\log \zeta_j = - \mathrm{Var}(\ell_j)$ with no remainder.

Combining Equations~\eqref{eq:app_cess_variance} and~\eqref{eq:app_variance_delta} gives the local model used in Section~\ref{sec:adaptive_step_size},
\begin{equation} \label{eq:app_local_model}
    \zeta_j \approx \exp \left \{ - v_j \Delta_j^2 \right \} \, .
\end{equation}
To quantify its accuracy, we substitute Equation~\eqref{eq:app_variance_delta} into Equation~\eqref{eq:app_log_cess_expansion},
\begin{align}
    \log \zeta_j & = - \mathrm{Var} ( \ell_j ) + \mathcal{O}(\Delta_j^3) \\
    & = - v_j \Delta_j^2 + \mathcal{O}(\Delta_j^3) + \mathcal{O}(\Delta_j^3) \\
    & = - v_j \Delta_j^2 + \mathcal{O}(\Delta_j^3) \, . \label{eq:app_local_model_error}
\end{align}
Both remainders on the second line are of order $\Delta_j^3$, but they have different sources. The first comes from Equation~\eqref{eq:app_variance_delta}: it accounts for the movement of the particles during the step. The second comes from the Taylor expansion in Equation~\eqref{eq:app_log_cess_expansion}: it accounts for the shape of the distribution of the increments.

The size of this absolute error is not informative on its own: as $\Delta_j \to 0$ it goes to zero, but so does the quantity being approximated, $\log \zeta_j \approx - v_j \Delta_j^2$, and a small error on a small quantity could still be a large fraction of it. Dividing the error by the retained term $v_j \Delta_j^2$, assuming $v_j > 0$ strictly, gives the relative error of $\log \zeta_j$ under Equation~\eqref{eq:app_local_model},
\begin{equation}
    \frac{ \mathcal{O}(\Delta_j^3) }{ v_j \Delta_j^2 } = \mathcal{O}(\Delta_j) \, .
\end{equation}
Both sources contribute at this order, so for short steps the model is accurate whenever the increments have finite moments. Under Assumption~\ref{as:phi_lip} alone, the relative error can only be bounded by $\mathcal{O}(\sqrt{\Delta_j})$.

Equation~\eqref{eq:app_local_model} is local in two senses. First, it neglects the remainder above, so it is accurate only for short steps. Second, the step-size rule below estimates $v_j$ from transition $j$ and uses the estimate to choose $\Delta_{j+1}$, which requires $v_{j+1} \approx v_j$, that is, that the particle system changes little over one transition.

\subsection{The step-size update}

After a completed transition of length $\Delta_j$, the observed CESS ratio $\widehat{\zeta}_j$ estimates $\zeta_j$, so Equation~\eqref{eq:app_local_model} gives $\log \widehat{\zeta}_j \approx - v_j \Delta_j^2$, and hence the estimate
\begin{equation} \label{eq:app_v_estimate}
    \widehat v_j := - \frac{ \log \widehat{\zeta}_j }{ \Delta_j^2 } \, .
\end{equation}
Requiring the next transition to attain the target $\zeta_\star$ under the same model, with $v_{j+1}$ replaced by $\widehat v_j$, that is $\zeta_\star = \exp \{ - \widehat v_j \widehat{\Delta}_{j+1}^2 \}$, and solving for $\widehat{\Delta}_{j+1}$,
\begin{align}
    \widehat{\Delta}_{j+1} & = \sqrt{ \frac{ - \log \zeta_\star }{ \widehat v_j } } \\
    & = \Delta_j \sqrt{ \frac{ \log \zeta_\star }{ \log \widehat{\zeta}_j } } \, ,
\end{align}
which is Equation~\eqref{eq:adaptive_rule_main}. In adaptive tempering, the incremental weights of a candidate step are available in closed form, and the step that attains the target can be found by bisection \citep{jasra2011inference, zhou2016toward}. In ABF they depend on the propagated paths, so each candidate step would require propagating the particles again, and Equation~\eqref{eq:adaptive_rule_main} instead predicts the next step from the completed transition. This has the form of the elementary step-size controller of ODE solvers, with the CESS in place of the estimate of the local error \citep[Section~II.4]{hairer1993solving}. Both logarithms are negative, so the ratio is positive, and the proposed step is smaller than $\Delta_j$ when $\widehat{\zeta}_j < \zeta_\star$ and larger when $\widehat{\zeta}_j > \zeta_\star$.

The same quantity $v_j$ can be estimated directly from the spread of the observed increments. With
\begin{equation}
    \overline{\ell}_j := \sum_{i=1}^N W_{j-1,i} \, \ell_{j,i} \, , \qquad \widehat{\mathrm{Var}}(\ell_j) := \sum_{i=1}^N W_{j-1,i} \left ( \ell_{j,i} - \overline{\ell}_j \right )^2 \, ,
\end{equation}
the alternative estimate is
\begin{equation} \label{eq:app_v_var}
    \widehat v_j^{\, \mathrm{var}} := \frac{ \widehat{\mathrm{Var}}(\ell_j) }{ \Delta_j^2 } \, .
\end{equation}
We use Equation~\eqref{eq:app_v_estimate} for the update, because it is calibrated directly against the quantity being targeted, and Equation~\eqref{eq:app_v_var} can be used as a check: the two estimates agree to the extent that Equation~\eqref{eq:app_cess_variance} holds, so a large discrepancy indicates that the third and higher moments of the increments are not negligible.

\subsection{The first transition}

Before the first transition there is no observed CESS, but Equation~\eqref{eq:app_variance_delta} identifies $v_1$ as the variance of the summed potential across the initial particle system, and the potentials $\varphi_{0,i}$ are already evaluated in order to form the first incremental weight. Writing $\widehat{\mathrm{Var}}(\varphi_0)$ for their weighted variance, solving $\zeta_\star = \exp \{ - \widehat{\mathrm{Var}}(\varphi_0) \Delta_1^2 \}$ for $\Delta_1$, and halving the result, gives the first step
\begin{equation} \label{eq:app_first_step}
    \Delta_1 = \frac{1}{2} \sqrt{ \frac{ - \log \zeta_\star }{ \widehat{\mathrm{Var}}(\varphi_0) } } \, .
\end{equation}

Two features of Equation~\eqref{eq:app_first_step} are worth noting. First, it is conservative by construction. The particles are usually most dispersed at $t = 0$ and coalesce as the run proceeds, so $v_j$ is largest at the start and Equation~\eqref{eq:app_first_step} typically returns the smallest step the run will require; the factor $1/2$ adds a further margin. Second, the direction of the error matters. A first step that is too small is corrected geometrically, at a cost of $\log(\Delta^\star / \Delta_1) / \log g_{\max}$ transitions, where $\Delta^\star$ is the step the rule settles on, and those transitions are inexpensive because they cover little of $[0,T]$. A first step that is too large degenerates the weights immediately, and the rule can only shrink the step by a factor $g_{\min}$ per transition, by which time the information lost cannot be recovered by resampling.

\subsection{Safeguards}

Because Equation~\eqref{eq:app_local_model} is a local approximation, we limit how much the step may change between adjacent transitions, as is standard in step-size control for ODE solvers \citep[Section~II.4]{hairer1993solving},
\begin{equation} \label{eq:app_clipping}
    \widehat{\Delta}_{j+1}^{\mathrm{clip}} := \min \left \{ g_{\max} \Delta_j , \, \max \left \{ g_{\min} \Delta_j , \, \widehat{\Delta}_{j+1} \right \} \right \} \, ,
\end{equation}
with $0 < g_{\min} < 1 < g_{\max}$. We use $g_{\min} = 0.5$ and $g_{\max} = 1.5$, so that a single observation can at most halve the step or increase it by one half. Equation~\eqref{eq:app_clipping} also covers the degenerate case in which $\widehat{\zeta}_j$ is numerically equal to one, so that $\widehat v_j = 0$, and Equation~\eqref{eq:adaptive_rule_main} would propose an unbounded increase.

The step is also kept above a small floor $\Delta_{\min} > 0$ and within the bridge interval, and is extended to reach $T$ exactly when the remaining time is within rounding tolerance of it.

\subsection{Offline adaptation}

Equations~\eqref{eq:adaptive_rule_main} and~\eqref{eq:app_first_step} make the mesh a function of the particle system, as in other adaptive SMC methods \citep{beskos2016convergence}. This dependence can be removed by a pilot run, whose mesh is then held fixed in the production run; in our experiments the two variants gave indistinguishable accuracy, so we adapt online and avoid the pilot run.

\section{Detailed ABF algorithm} \label{app:complete_algo}

\newgeometry{top=0.1cm,bottom=0.1cm}

\begin{algorithm}[t]
\footnotesize
\caption{Detailed Approximate Bayesian Fusion}
\label{alg:approx-bayesian-fusion}
\begin{algorithmic}[1]
\Require Shard particles $\{\theta_{t_0,i}^{(c)}\}_{i=1,\ldots,N}^{c=1,\ldots,C}$ with $t_0=0$, preconditioned potentials $\{\phi^{(c)}\}_{c=1}^C$, calibration sample size $M$, target rejection acceptance $a_0$,
target CESS ratio $\zeta_\star$, resampling threshold $\tau$.

\Statex
\State \textbf{Set the preconditioners.}
Set $\Lambda_c$ to the empirical covariance $\hat{Q}^{(c)}$ of the draws from shard $c$, or to its diagonal, for $c = 1, \ldots, C$, and set $\bar{\Lambda} = ( \sum_{c=1}^C \Lambda_c^{-1} )^{-1}$. For a tuple $\theta^{(1:C)}$, write $\tilde{\theta} = \bar{\Lambda} \sum_{c=1}^C \Lambda_c^{-1} \theta^{(c)}$ for its precision-weighted centre. Setting $\Lambda_c = I_d$ gives ABF without preconditioning.

\Statex
\State \textbf{Tune coalescence time.}
Draw calibration tuples $\theta_q^{(1:C)}$, $q=1,\ldots,M$, and compute $\mathcal{E}_q 
= \frac{1}{2}\sum_{c=1}^C (\theta_q^{(c)}-\tilde \theta_q)^\top \Lambda_c^{-1} (\theta_q^{(c)}-\tilde \theta_q)$.
Find $\hat T$ via bisection s.t.\ $\widehat a(\hat T)=M^{-1}\sum_{q=1}^{M}\exp\{-\mathcal{E}_q/\hat T\}=a_0$.

\Statex
\State \textbf{Run the rejection sampler and initialise SMC.}
Draw independent shard-particle combinations $\theta_i^{(1:C)}$ and accept a proposal with probability $a(\theta_i^{(1:C)};\hat T)=\exp\{-\mathcal{E}(\theta_i^{(1:C)})/\hat T\}$; write the accepted tuple as $\{\theta_{t_0,i}^{(c)}\}_{c=1}^C$.
\State Set $W_i \gets 1/N$, $\quad w_i \gets W_i$, $\quad \varphi_{0,i} \gets \sum_{c=1}^C \phi^{(c)}(\theta_{t_0,i}^{(c)})$

\Statex
\State \textbf{Choose the first step.}
Set $\Delta_1$ by Equation~\eqref{eq:app_first_step} and $j \gets 1$.

\Statex
\While{$t_{j-1} + \Delta_j < \hat T$}
    \If{$\mathrm{ESS}(W_{1:N}) < \tau N$}
        \State Resample ancestor indices $A_1,\ldots,A_N$ from $W_{1:N}$
        \State Set $\theta_{t_{j-1},i}^{(c)} \gets \theta_{t_{j-1},A_i}^{(c)}$ for all $i,c$
        \State Set $W_i \gets 1/N$ and $w_i \gets W_i$
        \State Update $\varphi_{j-1,i} \gets \varphi_{j-1,A_i}$ 
    \EndIf

    \State Set $t_j \gets t_{j-1} + \Delta_j$
    \State Set $\alpha = (\hat T-t_j)/(\hat T-t_{j-1})$, so that $1 - \alpha = \Delta_j / (\hat T-t_{j-1})$

    \For{$i=1,\ldots,N$}
        \State Compute the current weighted centre, $\tilde \theta_{t_{j-1},i} = \bar{\Lambda} \sum_{c=1}^C \Lambda_c^{-1} \theta_{t_{j-1},i}^{(c)}$
        \State Draw common noise $\chi_{j,i} \sim \mathcal N(0,I_d)$
        
        \For{$c=1,\ldots,C$}
            \State Draw shard-specific $\psi_{j,i}^{(c)} \sim \mathcal N(0,I_d)$

            \State Propagate the bridge:
            \begin{equation*}
                \theta_{t_j,i}^{(c)} = \alpha \theta_{t_{j-1},i}^{(c)} + (1 - \alpha) \tilde \theta_{t_{j-1},i} + \sqrt{\Delta_j}\sqrt{1 - \alpha}\,\bar{\Lambda}^{1/2}\chi_{j,i} + \sqrt{\Delta_j}\sqrt{\alpha}\,\Lambda_c^{1/2}\psi_{j,i}^{(c)}
            \end{equation*}
        \EndFor

        \State Evaluate $\varphi_{j,i} = \sum_{c=1}^C \phi^{(c)}\!\left(\theta_{t_j,i}^{(c)}\right)$

        \State Compute the increment by the trapezium rule, $\ell_{j,i} = \frac{\Delta_j}{2} \left( \varphi_{j-1,i} + \varphi_{j,i} \right)$, and update $\log w_i \gets \log w_i - \ell_{j,i}$.
    \EndFor
    \State Compute $\widehat{\zeta}_j$ by Equation~\eqref{eq:app_cess}, from the normalised weights $W_i$ of the previous step and the increments $\ell_{j,i}$ of this step
    \State Set $\Delta_{j+1}$ by Equation~\eqref{eq:adaptive_rule_main}, clipped as in Equation~\eqref{eq:app_clipping} and bounded below by $\Delta_{\min}$
    \State Normalise $W_i = w_i / \sum_{k=1}^N w_k$ for $i=1,\ldots,N$
    \State Set $j \gets j + 1$
\EndWhile
\State Set $n \gets j$ and $\Delta_n \gets \hat T - t_{n-1}$

\Statex
\State \textbf{Final fusion step.}
\For{$i=1,\ldots,N$}
    \State Compute $\tilde \theta_{t_{n-1},i} = \bar{\Lambda} \sum_{c=1}^C \Lambda_c^{-1} \theta_{t_{n-1},i}^{(c)}$

    \State Draw $\chi_{n,i}\sim \mathcal N(0,I_d)$ and set $Y_i = \tilde \theta_{t_{n-1},i} + \sqrt{\Delta_n} \, \bar{\Lambda}^{1/2} \chi_{n,i}$

    \State Evaluate the final fused potential $\varphi_{n,i} = \sum_{c=1}^C \phi^{(c)}(Y_i)$.
    \State Compute $\ell_{n,i} = \frac{\Delta_n}{2} \left( \varphi_{n-1,i} + \varphi_{n,i} \right)$ and update $\log w_i \gets \log w_i - \ell_{n,i}$.
\EndFor
\State Normalise $W_i = w_i / \sum_{k=1}^N w_k$ for $i=1,\ldots,N$

\State \Return $\{(Y_i,W_i)\}_{i=1}^N$.
\end{algorithmic}
\end{algorithm}

\restoregeometry

\section{Step size comparison} \label{app:tuning_calcs_bf}

This appendix compares the step sizes chosen by the rule of Section~\ref{sec:adaptive_step_size} with those recommended by the asymptotic analysis of \cite{dai2023bf}, Section~3. Choosing the two constants of that analysis as in Section~\ref{sec:tuning_asymptotics} at every step gives the adaptive mesh of \citet[Section~4.2.2 and Proposition~19]{chan2023dcf}, here without preconditioning. Section~\ref{sec:tuning_asymptotics} describes their recommendation and how we evaluate it, and Section~\ref{sec:tuning_comparison} compares the two rules along ABF runs on the logistic regression of Section~\ref{sec:logistic_comparison}.

\subsection{The asymptotic step size} \label{sec:tuning_asymptotics}

\cite{dai2023bf} assume that the subposteriors are Gaussian with a common isotropic variance, that is, that the $c$-th subposterior is
\begin{equation} \label{eq:assumptions_normality}
    \mathcal{N} ( \mathbf{a}^{(c)}, \, s^2 I_d ) \, , \qquad s^2 := \frac{C b}{m} \, ,
\end{equation}
where $\mathbf{a}^{(c)} \in \mathbb{R}^d$, $m$ is the total number of observations, and $b > 0$ is a constant. The variance of a subposterior therefore decreases as $\mathcal{O}(C / m)$, since each shard holds $m / C$ observations.

For the transition from $t_{j-1}$ to $t_j = t_{j-1} + \Delta_j$, let
\begin{equation} \label{eq:nu_j}
    \nu_j := \frac{1}{C} \sum_{c=1}^C \| \theta_{t_{j-1}}^{(c)} - \mathbf{a}^{(c)} \|^2
\end{equation}
be the average squared distance of the shard components of a particle at time $t_{j-1}$ from their subposterior means. Theorem~5 of \cite{dai2023bf} states that, for positive constants $k_3$ and $k_4$, if
\begin{equation} \label{eq:deltaT}
    \Delta_j \le \widetilde{\Delta}_j := \min \left \{ \frac{b^2 k_3 C}{m^2 \, \mathbb{E} \nu_j} \, , \; \left ( \frac{k_4 b^2 C}{2 d m^2} \right )^{1/2} \right \} \, ,
\end{equation}
where the expectation is conditional on the particle system at time $t_{j-1}$, then the CESS ratio of the transition is at least $e^{-k_3 - k_4}$ in the limit $N \to \infty$ followed by $\Delta_j \to 0$. Their Remarks~3 and~5 recommend taking $\Delta_j = \widetilde{\Delta}_j$ at every step, with $k_3$ and $k_4$ such that $e^{-k_3 - k_4} = \zeta_\star$ for a target CESS ratio $\zeta_\star \in (0,1)$.

To write Equation~\eqref{eq:deltaT} in terms of $s^2$, we use the definition of $s^2$ in Equation~\eqref{eq:assumptions_normality} in each of the two terms,
\begin{align}
    \frac{b^2 k_3 C}{m^2} & = \left ( \frac{C b}{m} \right )^2 \frac{k_3}{C} = \frac{s^4 k_3}{C} \, , \\
    \frac{k_4 b^2 C}{2 d m^2} & = \left ( \frac{C b}{m} \right )^2 \frac{k_4}{2 C d} = \frac{s^4 k_4}{2 C d} \, ,
\end{align}
where we multiplied and divided by $C$ in both lines. Substituting into Equation~\eqref{eq:deltaT}, and taking $s^4$ out of the square root, gives
\begin{equation} \label{eq:deltaT_s2}
    \widetilde{\Delta}_j = \min \left \{ \frac{s^4 k_3}{C \, \mathbb{E} \nu_j} \, , \; s^2 \left ( \frac{k_4}{2 C d} \right )^{1/2} \right \} \, .
\end{equation}

The constraint $k_3 + k_4 = - \log \zeta_\star$ leaves one constant free. The first term of Equation~\eqref{eq:deltaT_s2} increases with $k_3$ and the second with $k_4$, so the minimum is largest when the two terms are equal, and we choose $k_3$ and $k_4$ in this way. This is the choice most favourable to Bayesian Fusion, as it gives the largest step that the constraint allows. Write
\begin{equation} \label{eq:kappa_j}
    \kappa_j := \frac{C \, ( \mathbb{E} \nu_j )^2}{2 d s^4} \, .
\end{equation}
Equating the two terms of Equation~\eqref{eq:deltaT_s2}, squaring both sides, and multiplying both sides by $C^2 (\mathbb{E} \nu_j)^2 / s^8$, we obtain
\begin{align}
    \frac{s^4 k_3}{C \, \mathbb{E} \nu_j} & = s^2 \left ( \frac{k_4}{2 C d} \right )^{1/2} \\
    \frac{s^8 k_3^2}{C^2 ( \mathbb{E} \nu_j )^2} & = \frac{s^4 k_4}{2 C d} \\
    k_3^2 & = \kappa_j k_4 \, ,
\end{align}
where the last line uses the definition of $\kappa_j$ in Equation~\eqref{eq:kappa_j}. Substituting $k_3 = - \log \zeta_\star - k_4$ and expanding the square, we obtain
\begin{align}
    ( - \log \zeta_\star - k_4 )^2 & = \kappa_j k_4 \\
    k_4^2 - ( \kappa_j - 2 \log \zeta_\star ) \, k_4 + ( \log \zeta_\star )^2 & = 0 \, .
\end{align}
Since $\zeta_\star \in (0,1)$, we have $\log \zeta_\star < 0$. The discriminant of this quadratic, $( \kappa_j - 2 \log \zeta_\star )^2 - 4 ( \log \zeta_\star )^2 = \kappa_j \, ( \kappa_j - 4 \log \zeta_\star )$, is positive for $\kappa_j > 0$, and its roots have sum $\kappa_j - 2 \log \zeta_\star > 0$ and product $( \log \zeta_\star )^2 > 0$. The roots are therefore distinct and positive, and since their product is $( \log \zeta_\star )^2$, exactly one of them is smaller than $- \log \zeta_\star$. Since $k_3 = - \log \zeta_\star - k_4$ must be positive, $k_4$ is the smaller root,
\begin{equation} \label{eq:k3_k4}
    k_4 = \frac{1}{2} \left ( \kappa_j - 2 \log \zeta_\star - \sqrt{\kappa_j \, ( \kappa_j - 4 \log \zeta_\star )} \right ) \, , \qquad k_3 = - \log \zeta_\star - k_4 \, .
\end{equation}
With this choice, the two terms in Equation~\eqref{eq:deltaT_s2} are equal.

The quantities in Equation~\eqref{eq:deltaT_s2} are estimated from the particles. We estimate $\mathbf{a}^{(c)}$ by the mean of the draws from shard $c$, and $\mathbb{E} \nu_j$ by the average of $\nu_j$ over the particles at time $t_{j-1}$, weighted by their normalised weights, which is the weighted version of the estimator in Remark~4 of \cite{dai2023bf}. Since $\mathrm{tr} ( s^2 I_d ) = d s^2$, we estimate $s^2$ by
\begin{equation} \label{eq:s2hat}
    \hat{s}^2 = \frac{1}{C d} \sum_{c=1}^C \mathrm{tr} ( \hat{Q}^{(c)} ) \, ,
\end{equation}
where $\hat{Q}^{(c)}$ is the sample covariance of the draws from shard $c$.

\subsection{Comparison on the logistic regression} \label{sec:tuning_comparison}

We compare the two rules along ABF runs on the logistic regression of Section~\ref{sec:logistic_comparison}, with $d = 6$, $C = 4$, $m = 2{,}000$ observations, and $N = 10{,}000$ particles. We run ABF with the steps chosen by each rule in turn, both with target $\zeta_\star = 0.4$: after each transition, the next step is the one recommended by Equation~\eqref{eq:adaptive_rule_main}, without the bounds $g_{\min}$ and $g_{\max}$, or by Equation~\eqref{eq:deltaT_s2}. Both rules use the coalescence time $T = 0.0353$, and start from the same initial particles and the same first step $\Delta_1 = 0.0039$, which is the uniform step chosen by the bisection of Section~\ref{sec:adaptive_step_size} for $\zeta_\star = 0.4$. We repeat the runs of each rule ten times.

\begin{figure}[!htb]
    \centering
    \includegraphics[width=1\linewidth]{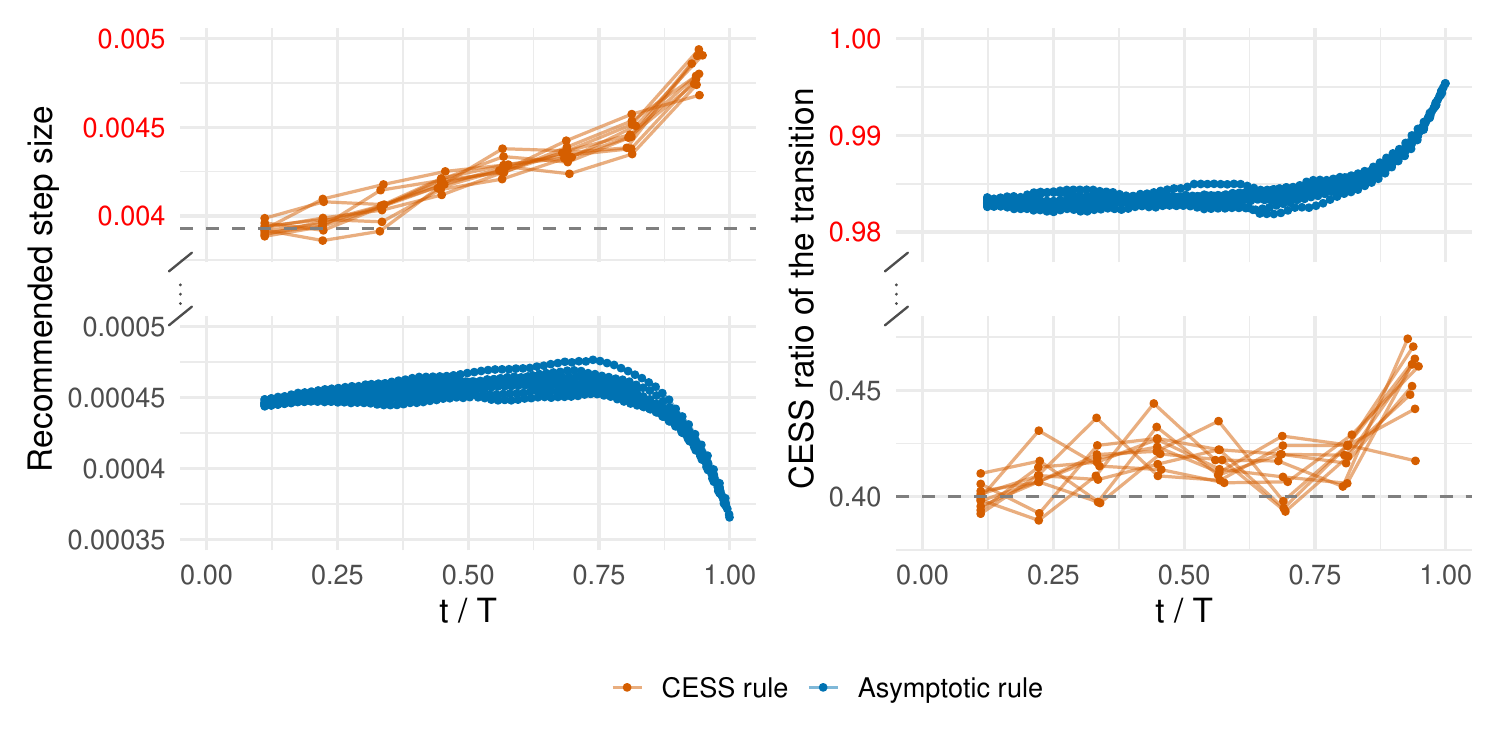}
    \caption{Step sizes along ABF runs on the logistic regression of Section~\ref{sec:logistic_comparison}, in which the steps are chosen by Equation~\eqref{eq:adaptive_rule_main} (CESS rule) or by Equation~\eqref{eq:deltaT_s2} (asymptotic rule), both with target $\zeta_\star=0.4$ and the same first step $\Delta_1 = 0.0039$. Left: after each transition, the step that the rule recommends for the next transition, against the bridge time $t / T$; the dashed line marks the first step $\Delta_1$. Right: the CESS ratio of each transition; the dashed line marks the target $0.4$. The first transition, which is the same for both rules, is shown only for the CESS rule. In both panels the vertical axis is broken, and each of its two parts has its own linear scale. Each line is one of the ten replicas of each rule.}
    \label{fig:step_size_compare}
\end{figure}

Figure~\ref{fig:step_size_compare} shows the results. The first transition is the same for both rules, and its CESS ratio is close to the target. With Equation~\eqref{eq:adaptive_rule_main}, the CESS ratio of every later transition also stays close to the target (right panel), and the steps increase gradually along the bridge (left panel). The steps recommended by the asymptotic rule are about an order of magnitude smaller, change little along most of the bridge, and decrease towards its end. With these steps, the CESS ratio of every later transition is close to one, well above the target, and a run takes about $70$ steps, against $9$ with Equation~\eqref{eq:adaptive_rule_main}.

The asymptotic rule bounds the CESS of the Bayesian Fusion weights in a limit, and on this example it is conservative for ABF: its steps attain a CESS ratio far above the target, at the cost of many more steps. Equation~\eqref{eq:adaptive_rule_main} instead uses the CESS that the completed transition actually attained, and therefore recommends a step close to the one that attains the target.

\section{Proofs from Section \ref{sec:controlling_error}} \label{app:proofs}

\subsection{Technical lemmas for the strong error bound}

The following are some technical lemmas that are needed in the main proof, and were separated for readability.

\begin{lemma} \label{lem:theta_0}
    Let $\{\theta_{0}^{(c)}\}_{c=1}^C$ be a collection of mutually independent $d$-dimensional random vectors with marginal densities $\{f_c\}_{c=1}^C$, where $f_c$ is an unknown distribution with finite mean $\mu^{(c)}$ and covariance $Q^{(c)}$, as per Assumption~\ref{as:smoothness}. Then
    \begin{equation} \label{eq:theta_0}
        \mathbb{E} \| \theta_0^{(c)} - \bar{\theta}_0 \| \le \sqrt{d} \, \kappa^{(c)} \, ,
    \end{equation}
    where $\bar{\theta}_0 = C^{-1} \sum_{c=1}^C \theta_0^{(c)}$, and $\kappa^{(c)}$, $\bar{\mu}$ and $\bar{Q}$ are defined in Proposition~\ref{prop:strong_error}.
\end{lemma}

\noindent We introduced the constant $\kappa^{(c)}$ in Equation \eqref{eq:kappa} as a quantity that is stable as the dimension $d$ goes to infinity. This allows us to explicitly express the natural dependence on dimensionality in Equation \eqref{eq:theta_0}.

\begin{proof}
    The expected value and variance of $\bar{\theta}_0$ are:
    \begin{align}
    \mathbb{E} \bar{\theta}_0 & = \frac{1}{C} \sum_{c=1}^C \mathbb{E} \theta_0^{(c)} = \frac{1}{C} \sum_{c=1}^C \mu^{(c)} = \bar{\mu} \\
    \mathrm{Var}(\bar{\theta}_0) & = \frac{1}{C^2} \mathrm{Var} \left ( \sum_{c=1}^C \theta_0^{(c)} \right ) = \frac{1}{C^2} \sum_{c=1}^C \mathrm{Var} \left ( \theta_0^{(c)} \right ) = \frac{1}{C^2} \sum_{c=1}^C Q^{(c)} = \frac{\bar{Q}}{C} \, .
\end{align}
Applying the triangle inequality, Jensen's inequality, and using the formula $\mathbb{E}[\| \theta \|^2] = \mathrm{tr}( \mathrm{Var}(\theta) ) + \| \mathbb{E}[\theta] \|^2$ for a random vector $\theta$, we have
\begin{align}
    \mathbb{E} \| \theta_0^{(c)} - \bar{\theta}_0 \| 
    & \le \mathbb{E} \| \theta_0^{(c)} \| + \mathbb{E} \| \bar{\theta}_0 \| \\
    & \le \sqrt{ \mathbb{E} [ \| \theta_0^{(c)} \|^2 ] } + \sqrt{ \mathbb{E} [ \| \bar{\theta}_0 \|^2 ] } \\
    & = \sqrt{ \mathrm{tr}(Q^{(c)}) + \| \mu^{(c)} \|^2 } + \sqrt{ \frac{1}{C} \mathrm{tr}(\bar{Q}) + \| \bar{\mu} \|^2 } \\
    & := \sqrt{d} \, \kappa^{(c)} \, .
\end{align}
\end{proof}

\begin{lemma} \label{lem:theta_T_conditional}
    Let $\theta_t$ follow the SDE from Equation \eqref{eq:sde} on $t \in [0,T]$. Then, conditional on $\theta_0$,
    \begin{equation}
        \vartheta | \theta_0 \sim \mathcal{N} \left ( \bar \theta_0, \frac{T}{C} I_d \right ) \, , \quad \mbox{and} \quad (\vartheta - \bar \theta_0) | \theta_0 \sim \mathcal{N} \left ( 0, \frac{T}{C} I_d \right ) \, .
    \end{equation}
\end{lemma}

\begin{proof}
    Starting from $\theta_0^{(c)}$, the transition density of the $c$th Brownian motion at the common endpoint $\vartheta=y$ is proportional to
    \begin{equation}
        \exp \left \{ -\frac{1}{2T} \| y - \theta_0^{(c)} \|^2 \right \} \, .
    \end{equation}
    Since the $C$ Brownian motions are independent before conditioning on the common endpoint, the density of the shared endpoint, conditional on $\theta_0$, is therefore the product
    \begin{align}
        p(y \mid \theta_0) 
        & \propto \prod_{c=1}^C \exp\left\{-\frac{1}{2T}\|y-\theta_0^{(c)}\|^2\right\} \\
        & \propto \exp \left \{ -\frac{1}{2T} \sum_{c=1}^C \| y - \theta_0^{(c)} \|^2 \right \} \, .
    \end{align}
    Completing the square by adding and subtracting the mean $\bar{\theta}_0$ gives
    \begin{align}
        \sum_{c=1}^C \|y-\theta_0^{(c)}\|^2
        & = \sum_{c=1}^C \|y - \bar{\theta}_0 + \bar{\theta}_0 - \theta_0^{(c)}\|^2 \\
        & = \sum_{c=1}^C \|y - \bar \theta_0\|^2 + 2 (y-\bar \theta_0)^\top \sum_{c=1}^C (\bar \theta_0-\theta_0^{(c)}) \\
        & \qquad + \sum_{c=1}^C \| \bar \theta_0 - \theta_0^{(c)} \|^2 \\
        & = C \| y - \bar \theta_0 \|^2 + \sum_{c=1}^C \|\theta_0^{(c)}-\bar \theta_0\|^2 \, ,
    \end{align}
    where on the second line the middle term vanishes, as by the definition of $\bar{\theta}_0$, the sum in $c$ is equal to zero.
    On the last line, the second term does not depend on $y$, so it is absorbed into the normalising constant. Hence
    \begin{equation}
        p(y \mid \theta_0)
        \propto
        \exp\left\{
        -\frac{C}{2T}\|y-\bar \theta_0\|^2
        \right\} \, ,
    \end{equation}
    and therefore
    \begin{equation}
        \vartheta | \theta_0 \sim \mathcal{N} \left ( \bar \theta_0, \frac{T}{C} I_d \right ) \, , \quad \mbox{and} \quad (\vartheta - \bar \theta_0) | \theta_0 \sim \mathcal{N} \left ( 0, \frac{T}{C} I_d \right ) \, .
    \end{equation}
\end{proof}

\begin{lemma} \label{lem:theta_T}
    Let $\theta_t$ follow the SDE from Equation \eqref{eq:sde} on $t \in [0,T]$. Then
    \begin{equation}
        \mathbb{E} \| \vartheta - \bar{\theta}_0 \| \le \sqrt{ d \, \frac{T}{C}}
    \end{equation}
\end{lemma}

\begin{proof}
    By the law of total expectation,
    \begin{equation}
        \mathbb{E} \| \vartheta - \bar{\theta}_0 \| = \mathbb{E} \left [ \mathbb{E} \left [ \| \vartheta - \bar{\theta}_0 \| \; | \; \theta_0 \right ] \right ] \, .
    \end{equation}

    By Lemma \ref{lem:theta_T_conditional},
    \begin{equation}
        (\vartheta - \bar \theta_0) | \theta_0 \sim \mathcal{N} \left ( 0, \frac{T}{C} I_d \right ) \, .
    \end{equation}

    By Jensen's inequality and independence,
    \begin{align}
        \mathbb{E}\left [ \| \vartheta - \bar{\theta}_0 \| \; | \; \theta_0 \right ]
        & \le \sqrt{\mathbb{E}\left [ \| \vartheta - \bar{\theta}_0 \|^2 \; | \; \theta_0 \right ] } \\
        & = \sqrt{ \mathbb{E} \left [ \sum_{r=1}^d (\vartheta_r - \bar{\theta}_{0,r})^2 \; | \; \theta_0 \right ] }\\
        & = \sqrt{ \sum_{r=1}^d \mathrm{Var} \left ( (\vartheta_r - \bar{\theta}_{0,r}) \; | \; \theta_0 \right ) } \\
        & = \sqrt{ d \, \frac{T}{C}} \, .
    \end{align}
    Consequently,
    \begin{align}
        \mathbb{E} \| \vartheta - \bar{\theta}_0 \| 
        & \le \mathbb{E} \left [ \sqrt{ d \, \frac{T}{C}} \right ]
        = \sqrt{ d \, \frac{T}{C}} \, .
    \end{align}
\end{proof}

\begin{lemma} \label{lem:bridge_incr}
    Let $B_t$ be a $d$-dimensional Brownian bridge starting at $0$ at time $0$, and finishing at $0$ at time $T$. 
    For $0 \le s \leq t \le T$,
    \begin{equation}
        \mathbb{E} \left [ \|B_{t} - B_{s} \| \; | \; B_0, B_T \right ] \le \sqrt{d \, (t-s) } \, .
    \end{equation}
\end{lemma}

\begin{proof}
    Brownian bridge increments are normally distributed, with mean zero and the following variance. As the components are all independent, we have
    \begin{align}
        & \mathrm{Var}(B_{t} - B_s) \\
        & = \mathrm{Var}(B_{t}) + \mathrm{Var}(B_s) - 2 \mathrm{Cov}(B_{t},B_{s}) \\
        & = \left ( t - \frac{t^2}{T} \right ) I_d + \left ( s - \frac{s^2}{T} \right ) I_d - 2 \left ( s - \frac{s t}{T} \right ) I_d \\
        & = \left ( t+s-2s - \frac{1}{T} \left ( t^2 + s^2 -2st \right ) \right ) I_d \\
        & = \left ( (t-s) - \frac{(t-s)^2}{T} \right ) I_d \\
        & = (t-s) \left ( 1- \frac{(t-s)}{T} \right ) I_d \, .
    \end{align}
    Therefore,
    \begin{equation} \label{eq:bridge_incr}
        B_{t} - B_{s} \sim \mathcal{N} \left ( 0, \, (t-s) (T-t+s)/T \cdot I_d \right ) \, .
    \end{equation}
    Using Jensen's inequality and the variance calculated above, we have
    \begin{align}
        \mathbb{E} \left [ \|B_{t} - B_{s} \| \; | \; B_0, B_T \right ] 
        & \le \sqrt{ \mathbb{E} \left [ \|B_{t} - B_{s} \|^2 \; | \; B_0, B_T \right ] } \\
        & = \sqrt{(t-s) (1-(t-s)/T) \; \mathrm{tr}( I_d )} \\
        & = \sqrt{d \, (t-s) (1-(t-s)/T)} \\
        & \le \sqrt{d \, (t-s) } \, ,
    \end{align}
    where $\sqrt{t-s} \ge \sqrt{(t-s) (1-(t-s)/T)}$ for $0 \le t-s \le T$.
\end{proof}

\begin{lemma} \label{lem:theta_increment}
    Let $\theta_t$ follow the SDE from Equation \eqref{eq:sde} on $t \in [0,T]$. Then, 
    with $0 \le s \leq t \le T$,
    \begin{equation}
        \mathbb{E} \| \theta_{t}^{(c)} - \theta_{s}^{(c)} \| \le (t-s) \sqrt{\frac{d}{CT}} + \frac{(t-s)}{T} \sqrt{d} \, \kappa^{(c)} + \sqrt{d \, (t-s)}
    \end{equation}
\end{lemma}

\begin{proof}
    Conditioning on the endpoints $\theta_0, \vartheta$, the process $\theta_t$ is a vector of $d$ independent Brownian bridges on $[0,T]$, and can be decomposed as 
    \begin{equation} \label{eq:theta_increment}
        \theta_{t}^{(c)} - \theta_{s}^{(c)} = \frac{t-s}{T}(\vartheta - \theta_0^{(c)}) + (B_{t}^{(c)} - B_{s}^{(c)}) \, ,
    \end{equation}
    where the first term represents the drift, and the second is the increment of a $d$-dimensional Brownian bridge on $[0,T]$
    \citep[Ch.~I, Sec.~3, pp.~35--36]{revuz1999continuous}.

    By the law of total expectations, we can write
    \begin{equation}
        \mathbb{E} \| \theta_{t}^{(c)} - \theta_{s}^{(c)} \|
        = \mathbb{E} \left [ \mathbb{E} \left [ \| \theta_{t}^{(c)} - \theta_{s}^{(c)} \| \; | \; \theta_0, \vartheta \right ] \right ] \, .
    \end{equation}
    Substituting in Equation \eqref{eq:theta_increment}, using the triangle inequality, and Lemma \ref{lem:bridge_incr}, we have
    \begin{align}
        \mathbb{E} \| \theta_{t}^{(c)} - \theta_{s}^{(c)} \|
        & \le \mathbb{E} \left [ \frac{t-s}{T}  \| \vartheta - \theta_0^{(c)} \| + \mathbb{E} \left [ \| (B_{t}^{(c)} - B_{s}^{(c)}) \| \; | \; \theta_0, \vartheta \right ] \right ] \\
        & \le \frac{t-s}{T} \mathbb{E} \| \vartheta - \theta_0^{(c)} \| + \sqrt{d \, (t-s)}  \, .
    \end{align}
    We can now add and subtract $\bar{\theta}_0$ in the norm, and use triangle inequality to obtain
    \begin{align}
        \mathbb{E} \| \theta_{t}^{(c)} - \theta_{s}^{(c)} \|
        & \le \frac{t-s}{T} \mathbb{E} \left \| \vartheta - \bar{\theta}_0 \right \| + \frac{t-s}{T} \mathbb{E} \left \| \theta_0^{(c)} - \bar{\theta}_0 \right \| + \sqrt{d \, (t-s)} \, .
    \end{align}
    Using Lemma \ref{lem:theta_0}, and Lemma \ref{lem:theta_T}, we can bound the expression above by
    \begin{align}
        \mathbb{E} \| \theta_{t}^{(c)} - \theta_{s}^{(c)} \|
        & \le \frac{t-s}{T} \sqrt{d \, \frac{T}{C}} + \frac{t-s}{T} \sqrt{d} \,\kappa^{(c)} + \sqrt{d \, (t-s)} \\
        & = (t-s) \sqrt{\frac{d}{CT}} + \frac{t-s}{T} \sqrt{d} \, \kappa^{(c)} + \sqrt{d \, (t-s)} \, .
    \end{align}
\end{proof}

\begin{lemma} \label{lem:ex_mvt}
For any $x, y \in \mathbb{R}$,
\begin{equation}
    |\mathrm{e}^{-y} - \mathrm{e}^{-x}|
    \le
    \mathrm{e}^{-\min(x,y)} |y-x| \, .
\end{equation}
\end{lemma}

\begin{proof}
   For any real $x$ and $y$ we have:
    \begin{align}
    |e^{-y} - e^{-x}| 
    & = \left\vert \int_x^{y} \frac{\mathrm{d} (e^{-z})}{\mathrm{d}z} \, \mathrm{d}z \right\vert \\
    & = \left\vert \int_x^{y} - e^{-z} \, \mathrm{d}z \right\vert \\
    & \leq \sup_{w \in [x,y]} \left\vert \int_x^{y} e^{-w} \, \mathrm{d}z \right\vert \\
    & = \left\vert \int_x^{y} e^{-\inf\{w : w \in [x,y]\}} \, \mathrm{d}z \right\vert \\
    & = \exp(-\min(x,y)) |y-x| \, .
    \end{align}
\end{proof}

\begin{lemma} \label{lem:minAB}
    Let us define the following shorthand notation:
    \begin{align}
    A := \sum_{c=1}^C \int_0^T \phi^{(c)}(\theta_s^{(c)}) \, \mathrm{d} s \, , 
    \qquad 
    A_\Delta := \sum_{c=1}^C \int_0^T \phi_\Delta^{(c)}(s, \theta_{t_{0:n}}^{(c)}) \, \mathrm{d} s \, .
    \end{align}
    Under Assumption \ref{as:phi_min}, we have:
    \begin{equation}
        \min (A, A_\Delta) \ge C \bar{\Phi} T \, ,
    \end{equation}
    where $\bar{\Phi} := C^{-1} \sum_{c=1}^C \Phi^{(c)}$.
\end{lemma}

\begin{proof}
The function $\phi^{(c)}_\Delta$ interpolates points of $\phi^{(c)}$, and hence the infimum of $\phi^{(c)}_\Delta$ cannot be smaller than the infimum of $\phi^{(c)}$. 
Thus: 
\begin{align}
    \min (A, A_\Delta) & = \min \left ( \sum_{c=1}^C \int_0^T \phi^{(c)}(\theta_s^{(c)}) \, \mathrm{d} s, \quad \sum_{c=1}^C \int_0^T \phi_\Delta^{(c)}(s, \theta_{t_{0:n}}^{(c)}) \, \mathrm{d} s \right ) \\
    & \ge \sum_{c=1}^C \int_0^T \inf_\theta \phi^{(c)}(\theta) \, \mathrm{d} s  \\
    & \ge \sum_{c=1}^C \int_0^T \Phi^{(c)}  \, \mathrm{d}s \\
    & = \sum_{c=1}^C \Phi^{(c)}  \, T \\
    & = C \, \bar{\Phi} \, T \, ,
\end{align}
where $\Phi^{(c)}$ exists and is finite by Assumption \ref{as:phi_min}.
    
\end{proof}

\subsection{Proof of Proposition \ref{prop:strong_error} (Strong error bound)} \label{sec:strong_error_proof}

\begin{proof}[Proof of Proposition \ref{prop:strong_error}]
    Starting from the expression of the strong error, we can apply Jensen's inequality, and Fubini's theorem (relying on Assumption \ref{as:fubini}) to obtain
    \begin{align}
        \epsilon_\Delta 
        & = \mathbb{E} \left [ \left | \sum_{c=1}^C \int_0^T \phi^{(c)}(\theta_s^{(c)}) \, \mathrm{d} s - \sum_{c=1}^C \int_0^T \phi_\Delta^{(c)}(s, \theta_{t_{0:n}}^{(c)}) \, \mathrm{d} s \right | \right ] \\
        & = \mathbb{E}\left [ \left | \sum_{c=1}^C \int_0^T \left ( \phi^{(c)}(\theta_s^{(c)}) - \phi^{(c)}_\Delta(s, \theta_{t_{0:n}}^{(c)}) \right ) \, \mathrm{d} s\right | \right ] \\
        & \le \mathbb{E}\left [ \sum_{c=1}^C \int_0^T \left | \phi^{(c)}(\theta_s^{(c)}) - \phi^{(c)}_\Delta(s, \theta_{t_{0:n}}^{(c)}) \right | \, \mathrm{d}s \right ] \\
        & = \sum_{c=1}^C \int_0^T \mathbb{E} \left | \phi^{(c)}(\theta_s^{(c)}) - \phi^{(c)}_\Delta(s, \theta_{t_{0:n}}^{(c)}) \right | \, \mathrm{d}s \, .
    \end{align}
    Let us divide the integral on each subinterval on the partition $\mathcal{P}$, and apply a change  of variable to express the increments in terms of the variable $h$:
    \begin{align}
        \epsilon_\Delta 
        & \le \sum_{c=1}^C \int_0^T \mathbb{E} \left | \phi^{(c)}(\theta_s^{(c)}) - \phi^{(c)}_\Delta(s, \theta_{t_{0:n}}^{(c)}) \right | \, \mathrm{d}s \\
        & = \sum_{c=1}^C \sum_{j=0}^{n-1} \int_{t_{j}}^{t_{j+1}} \mathbb{E} \left | \phi^{(c)}(\theta_s^{(c)}) - \phi^{(c)}_\Delta(s, \theta_{t_{0:n}}^{(c)}) \right | \, \mathrm{d}s \\
        & = \sum_{c=1}^C \sum_{j=0}^{n-1} \int_0^\Delta \mathbb{E} \left | \phi^{(c)}(\theta_{t_{j}+h}^{(c)}) - \phi^{(c)}_\Delta(t_{j}+h, \theta_{t_{0:n}}^{(c)}) \right | \, \mathrm{d}h \, .
    \end{align}
    Substituting the definition of $\phi^{(c)}_\Delta(\theta_{t_{j}+h}^{(c)})$ from Definition \ref{def:trapezium} in the modulus above and using the triangle inequality, we obtain
    \begin{align}
        & \left | \phi^{(c)}(\theta_{t_{j}+h}^{(c)}) - \phi^{(c)}_\Delta(t_{j}+h, \theta_{t_{0:n}}^{(c)}) \right | \\
        & = \Big | \phi^{(c)}(\theta_{t_{j}+h}^{(c)}) - \phi^{(c)}(\theta_{t_{j}}^{(c)}) - \frac{h}{\Delta} \left ( \phi^{(c)}(\theta_{t_{j+1}}^{(c)}) - \phi^{(c)}(\theta_{t_{j}}^{(c)}) \right ) \Big | \\
        & \le \Big | \phi^{(c)}(\theta_{t_{j}+h}^{(c)}) - \phi^{(c)}(\theta_{t_{j}}^{(c)}) \Big | + \frac{h}{\Delta} \Big  | \phi^{(c)}(\theta_{t_{j+1}}^{(c)}) - \phi^{(c)}(\theta_{t_{j}}^{(c)}) \Big | \, .
    \end{align}
    Substituting this decomposition into the bound above we obtain
    \begin{align}
        \epsilon_\Delta & \le \sum_{c=1}^C \sum_{j=0}^{n-1} \int_0^\Delta \mathbb{E} \left | \phi^{(c)}(\theta_{t_{j}+h}^{(c)}) - \phi^{(c)}(\theta_{t_{j}}^{(c)}) \right | \, \mathrm{d}h \\
        & \qquad + \sum_{c=1}^C \sum_{j=0}^{n-1} \int_0^\Delta \frac{h}{\Delta} \mathbb{E} \left | \phi^{(c)}(\theta_{t_{j+1}}^{(c)}) - \phi^{(c)}(\theta_{t_{j}}^{(c)}) \right | \, \mathrm{d}h \, .
    \end{align}

    Using the Lipschitz property from Assumption \ref{as:phi_lip}, we can bound the two terms directly as
    \begin{align}
        \epsilon_\Delta
        & \le L \sum_{c=1}^C \sum_{j=0}^{n-1} \int_0^\Delta
        \mathbb{E} \| \theta_{t_j+h}^{(c)} - \theta_{t_j}^{(c)} \| \, \mathrm{d}h \\
        & \quad + L \sum_{c=1}^C \sum_{j=0}^{n-1} \int_0^\Delta
        \frac{h}{\Delta}
        \mathbb{E} \| \theta_{t_{j+1}}^{(c)} - \theta_{t_j}^{(c)} \| \, \mathrm{d}h \, .
    \end{align}
    Using Lemma \ref{lem:theta_increment}, with $h$ in the first term and with $h=\Delta$ in the second, we have
    \begin{align}
        \epsilon_\Delta
        & \le L \sum_{c=1}^C \sum_{j=0}^{n-1} \int_0^\Delta
        \left ( h \sqrt{\frac{d}{CT}} + \frac{h}{T}\sqrt{d} \, \kappa^{(c)} + \sqrt{d h} \right ) \, \mathrm{d}h \\
        & \quad + L \sum_{c=1}^C \sum_{j=0}^{n-1} \int_0^\Delta
        \frac{h}{\Delta}
        \left ( \Delta \sqrt{\frac{d}{CT}} + \frac{\Delta}{T}\sqrt{d} \, \kappa^{(c)} + \sqrt{d \Delta} \right ) \, \mathrm{d}h \\[1em]
        & = L \sum_{c=1}^C \sum_{j=0}^{n-1}
        \left ( \frac{\Delta^2}{2}\sqrt{\frac{d}{CT}} + \frac{\Delta^2}{2T}\sqrt{d} \, \kappa^{(c)} + \frac{2}{3}\sqrt{d}\Delta^{3/2} \right ) \\
        & \quad + L \sum_{c=1}^C \sum_{j=0}^{n-1}
        \left ( \frac{\Delta^2}{2}\sqrt{\frac{d}{CT}} + \frac{\Delta^2}{2T}\sqrt{d} \, \kappa^{(c)} + \frac{1}{2}\sqrt{d}\Delta^{3/2} \right ) \\[1em]
        & = L \sum_{c=1}^C \sum_{j=0}^{n-1}
        \left ( \Delta^2 \sqrt{\frac{d}{CT}} + \frac{\Delta^2}{T}\sqrt{d} \, \kappa^{(c)} + \frac{7}{6}\sqrt{d}\Delta^{3/2} \right ) \\
        & = L C n \sqrt{d}
        \left ( \frac{\Delta^2}{\sqrt{CT}} + \frac{\Delta^2}{T}\bar{\kappa} + \frac{7}{6}\Delta^{3/2} \right ) \\
        & = L C \sqrt{d}
        \left ( \Delta \sqrt{\frac{T}{C}} + \Delta \bar{\kappa} + \frac{7}{6} T \sqrt{\Delta} \right ) \\
        & := E_\Delta \, ,
    \end{align}
    where we resolved the sums in $C$ and $n$, substituted $n \Delta = T$, and set $\bar{\kappa} = C^{-1} \sum_{c=1}^C \kappa^{(c)}$. For small $\Delta$, the term corresponding to the Brownian fluctuations dominates.
\end{proof}

\subsection{Proof of Theorem \ref{th:bound_strong} (TV distance bound)} \label{sec:bound_strong_proof}

\begin{proof}[Proof of Theorem \ref{th:bound_strong}]
    Let us adopt the following shorthand notation:
    \begin{align}
        A := \sum_{c=1}^C \int_0^T \phi^{(c)}(\theta_s^{(c)}) \, \mathrm{d} s \, , 
        \qquad 
        A_\Delta := \sum_{c=1}^C \int_0^T \phi_\Delta^{(c)}(s, \theta_{t_{0:n}}^{(c)}) \, \mathrm{d} s \, ,
    \end{align}
	and define $\mathbb{F}$ and $\mathbb{F}_\Delta$ via their Radon--Nikodym derivatives with respect to $\mathbb{P}$:
    \begin{equation}
        \frac{\mathrm{d} \mathbb{F}}{\mathrm{d} \mathbb{P}} = \frac{\rho_0 \mathrm{e}^{-A}}{Z} \, ,
        \qquad 
        \frac{\mathrm{d} \mathbb{F}_\Delta}{\mathrm{d} \mathbb{P}} = \frac{\rho_0 \mathrm{e}^{-A_\Delta}}{Z_\Delta}
    \end{equation}
    where $\rho_0$, $A$ and $A_\Delta$ are regarded as functions of the path $\theta_{0:T}$, with $\rho_0 = \rho_0(\theta_0)$, $Z = \mathbb{E}[\rho_0 \mathrm{e}^{-A}]$, and $Z_\Delta = \mathbb{E} [\rho_0 \mathrm{e}^{-A_\Delta} ]$.
    The total variation distance between $\mathbb{F}$ and $\mathbb{F}_\Delta$
    can be written as 
    \begin{align}
      \| \mathbb{F} - \mathbb{F}_\Delta \|_{\text{TV}} 
        & = \frac{1}{2} \left \| \frac{\mathrm{d} \mathbb{F}}{\mathrm{d} \mathbb{P}} - \frac{\mathrm{d} \mathbb{F}_\Delta}{\mathrm{d} \mathbb{P}} \right \|_{L_1(\mathbb{P})}\\ 
        & = \frac{1}{2} \mathbb{E} \left | \frac{\rho_0 e^{-A}}{Z} - \frac{\rho_0 e^{-A_\Delta}}{Z_\Delta}\right | \\
        & = \frac{1}{2} \mathbb{E} \left [ \rho_0 \left | \frac{e^{-A}}{Z} - \frac{e^{-A_\Delta}}{Z_\Delta}\right | \right ] \\
        & = \frac{1}{2 Z} \mathbb{E} \left [ \rho_0 \left | e^{-A} - e^{-A_\Delta} \frac{Z}{Z_\Delta} \right | \right ] \, .
    \end{align}
    Adding and subtracting $\exp \{ - A_\Delta\}$ inside the absolute value, 
    \begin{align}
        \| \mathbb{F} - \mathbb{F}_\Delta \|_{\text{TV}} 
        & = \frac{1}{2 Z} \mathbb{E}\left [ \rho_0 \left | e^{-A} - e^{-A_\Delta} + e^{-A_\Delta} - e^{-A_\Delta} \frac{Z}{Z_\Delta} \right | \right ] \\
        & = \frac{1}{2 Z} \mathbb{E} \left [ \rho_0 \left | e^{-A} - e^{-A_\Delta} + \left ( 1 - \frac{Z}{Z_\Delta} \right ) e^{-A_\Delta} \right | \right ] \\
        & \le \frac{1}{2 Z} \mathbb{E} \left [ \rho_0 \left | e^{-A} - e^{-A_\Delta} \right | \right ] + \frac{1}{2 Z} \frac{| Z_\Delta - Z| }{Z_\Delta} \mathbb{E} \left [ \rho_0 e^{-A_\Delta} \right ] \, ,
    \end{align}
    where on the last line we used the triangle inequality. Using the definition of $Z$ and $Z_\Delta$,
    \begin{align}
        | Z_\Delta - Z | 
        & = \left | \mathbb{E} \left [ \rho_0 e^{-A_\Delta} - \rho_0 e^{-A} \right ] \right |
        \le \mathbb{E} \left [ \rho_0 \left | e^{-A} -  e^{-A_\Delta} \right | \right ] \, .
    \end{align}
    Therefore, we have
    \begin{align}
        \| \mathbb{F} - \mathbb{F}_\Delta \|_{\text{TV}}
        & \le \frac{1}{2 Z} \mathbb{E} \left [ \rho_0 \left | e^{-A} - e^{-A_\Delta} \right | \right ] + \frac{1}{2Z} \mathbb{E} \left [ \rho_0 \left | e^{-A} -  e^{-A_\Delta} \right | \right ] \\
        & = \frac{1}{Z} \mathbb{E} \left [ \rho_0 \left | e^{-A} - e^{-A_\Delta} \right | \right ] \\
        & \le \frac{1}{Z} \mathbb{E} \left [ \left | e^{-A} - e^{-A_\Delta} \right | \right ] \, ,
    \end{align}
    where on the last line we used the fact that $0 \leq \rho_0 \le 1$  almost surely, as the initial weight $\rho_0$ is the exponential of a negative number. 

    Using Lemma \ref{lem:ex_mvt} on the first line, and Lemma \ref{lem:minAB} on the second,
    \begin{align}
        \| \mathbb{F} - \mathbb{F}_\Delta \|_{\text{TV}}
        & \le \frac{1}{Z} \mathbb{E} \left [ \left | e^{-A} - e^{-A_\Delta} \right | \right ] \\
        & \le \frac{1}{Z} \mathbb{E} \left [ \mathrm{e}^{- \min (A, A_\Delta)}  \, |A-A_\Delta| \right ] \\
        & \le \frac{1}{Z} \mathbb{E} \left [ \mathrm{e}^{- C \,\bar{\Phi}  \, T}  \, |A-A_\Delta| \right ] \\
        & = \frac{1}{Z} \mathrm{e}^{- C \, \bar{\Phi}  \, T}  \, \mathbb{E} \left [ |A-A_\Delta| \right ] \\
        & = \frac{1}{Z} \mathrm{e}^{- C \, \bar{\Phi}  \, T}  \cdot \epsilon_\Delta \, .
    \end{align}
    Using Proposition \ref{prop:strong_error}, this expression can be bounded as
    \begin{align}
        \| \mathbb{F} - \mathbb{F}_\Delta \|_{\text{TV}}
        & \le \frac{1}{Z} \mathrm{e}^{- C \, \bar{\Phi}  \, T}  \cdot E_\Delta  \\
        & =  \frac{1}{Z} \mathrm{e}^{- C \, \bar{\Phi}  \, T} \; L C \sqrt{d} \left ( \Delta \sqrt{\frac{T}{C}} + \Delta \bar{\kappa} + \frac{7}{6} T \sqrt{\Delta} \right ) \, .
    \end{align}
\end{proof}

\subsection{Proof of Remark \ref{rem:Phi_neg} (Negative lower bounds)} \label{sec:Phi_neg_proof}

\begin{proof}[Proof of Remark \ref{rem:Phi_neg}]
    We first show that $\phi^{(c)}$ has mean zero under $f_c$. Let $\theta \sim f_c$. Then
    \begin{align}
    \mathbb{E} [ \phi^{(c)}(\theta) ]
    & = \int_{\mathbb{R}^d} f_c(\theta) \, \phi^{(c)}(\theta) \, \mathrm{d} \theta \\
    & = \int_{\mathbb{R}^d} f_c(\theta) \, \frac{1}{2} \sum_{r=1}^d \left \{ \left ( \frac{\partial}{\partial \theta_r} \log f_c(\theta) \right )^2 + \frac{\partial^2}{\partial \theta_r^2} \log f_c(\theta) \right \} \mathrm{d} \theta \\
    & = \int_{\mathbb{R}^d} f_c(\theta) \, \frac{1}{2} \sum_{r=1}^d \frac{ \frac{\partial^2}{\partial \theta_r^2} f_c(\theta) }{ f_c(\theta) } \, \mathrm{d} \theta \\
    & = \frac{1}{2} \sum_{r=1}^d \int_{\mathbb{R}^d} \frac{\partial^2}{\partial \theta_r^2} f_c(\theta) \, \mathrm{d} \theta \, ,
    \end{align}
    where the second line uses the definition of $\phi^{(c)}$ in Equation~\eqref{eq:phi}, the third uses $\frac{\partial^2}{\partial \theta_r^2} \log f_c = \frac{\partial^2}{\partial \theta_r^2} f_c / f_c - \big ( \frac{\partial}{\partial \theta_r} \log f_c \big )^2$, obtained by differentiating $\frac{\partial}{\partial \theta_r} \log f_c = \frac{\partial}{\partial \theta_r} f_c / f_c$, and the last cancels $f_c$. Integrating by parts each component $r$ on the last line above, and from Assumption \ref{as:smoothness} applying the property $\frac{\partial}{\partial \theta_r} f_c(\theta_r,\theta_{-r}) \to 0$,
    as $\theta_r \to \pm \infty$, it follows that $\forall r=1,\ldots,d$,
    \begin{equation}
        \int_{\mathbb{R}^d} \frac{\partial^2}{\partial \theta_r^2} f_c(\theta) \, \mathrm{d} \theta = 0 \, ,
    \end{equation}
    and therefore that
    \begin{equation}
        \mathbb{E} [ \phi^{(c)}(\theta) ] = 0 \, .
    \end{equation}
    Hence a lower bound $\Phi^{(c)}$ cannot be positive. It cannot be zero either. If $\phi^{(c)} \ge 0$, then $\phi^{(c)} = 0$ almost everywhere, because its mean is zero and $f_c > 0$, and so everywhere, because $\phi^{(c)}$ is continuous by Assumption~\ref{as:smoothness}. By the third line above, $\phi^{(c)} = \sum_{r=1}^d \frac{\partial^2}{\partial \theta_r^2} f_c / (2 f_c)$, so $f_c$ would be a positive harmonic function on $\mathbb{R}^d$. By Liouville's theorem for positive harmonic functions, $f_c$ would then be constant, and could not be a probability density. Therefore $\Phi^{(c)} < 0$ for every $c$, $\bar{\Phi} < 0$, and the exponent $- C \, \bar{\Phi} \, T$ is positive. The same argument applies to the preconditioned potential in Equation~\eqref{eq:phi_precond}, after the linear change of variables $\theta \mapsto \Lambda_c^{-1/2} \theta$.
\end{proof}

\subsection{Technical lemmas for the weak error bound}

\begin{lemma} \label{lem:int_incr_h}
    Let $\theta_t = \{ \theta_t^{(c)} \}_{c=1}^C$ follow the SDE from Equation \eqref{eq:sde} on $t \in [0,T]$. Then for $s \leq v \leq t$,
    \begin{equation}
        \int_s^t \mathbb{E} [\theta_{v}^{(c)} - \theta_{s}^{(c)} \, | \, \theta_{s}, \theta_{t} ] \, \mathrm{d}v = \frac{t-s}{2} (\theta_{t}^{(c)} - \theta_{s}^{(c)}) \, .
    \end{equation}
\end{lemma}

\begin{proof}
    Let us decompose the process $\theta_t^{(c)}$ as follows \citep[Ch.~I, Sec.~3, pp.~35--36]{revuz1999continuous}.
\begin{equation}
    \theta_{v}^{(c)} \, | \, \theta_{s}, \theta_{t} = \theta_{s}^{(c)} + \frac{v-s}{t-s} (\theta_{t}^{(c)} - \theta_{s}^{(c)}) + (B_{v}^{(c)} - B_{s}^{(c)}) \, ,
\end{equation}
where $B_{v}^{(c)}$ is a Brownian bridge over time interval $[s,t]$, with endpoints $B_{s}^{(c)} = B_{t}^{(c)} = \boldsymbol{0}$.
The increment of $\theta_v^{(c)}$ conditioned on the endpoints of the interval is then
\begin{equation}
    \theta_{v}^{(c)} - \theta_{s}^{(c)} \, | \, \theta_{s}, \theta_{t} = \frac{v-s}{t-s} (\theta_{t}^{(c)} - \theta_{s}^{(c)}) + (B_{v}^{(c)} - B_{s}^{(c)}) \, .
\end{equation}
As $(B_{v}^{(c)} - B_{s}^{(c)})$ is an increment of a Brownian bridge with zero as initial and final point, it has $\mathbb{E}[B_{v}^{(c)} - B_{s}^{(c)} \, | \, B_{s}, B_{t}] = \boldsymbol{0}$. Therefore,
\begin{equation}
    \mathbb{E} \left [ \theta_{v}^{(c)} - \theta_{s}^{(c)} \, | \, \theta_{s}, \theta_{t} \right ] = \frac{v-s}{t-s} (\theta_{t}^{(c)} - \theta_{s}^{(c)}) \, ,
\end{equation}
and 
\begin{align} \label{eq:expected_increment_h}
    \int_s^t \mathbb{E} \left [ \theta_{v}^{(c)} - \theta_{s}^{(c)} \, | \, \theta_{s}, \theta_{t} \right ] \, \mathrm{d} v
    & = \frac{1}{t-s} (\theta_{t}^{(c)} - \theta_{s}^{(c)}) \int_s^t (v-s) \, \mathrm{d} v \\
    & = \frac{t-s}{2} (\theta_{t}^{(c)} - \theta_{s}^{(c)}) \, .
\end{align}
\end{proof}

\begin{lemma} \label{lem:int_square_norm_incr_h}
    Let $\theta_t = \{ \theta_t^{(c)} \}_{c=1}^C$ follow the SDE from Equation \eqref{eq:sde} on $t \in [0,T]$. 
    For $s \leq v \leq t$,
    \begin{equation}
        \int_s^t \mathbb{E} [ \| \theta_{v}^{(c)} - \theta_{s}^{(c)} \|^2 \, | \, \theta_{s}, \theta_{t}] \, \mathrm{d}v = \frac{t-s}{3} \| \theta_{t}^{(c)} - \theta_{s}^{(c)} \|^2 + d \frac{(t-s)^2}{6} \, .
    \end{equation}
\end{lemma}

\begin{proof}
    As in the proof of Lemma \ref{lem:int_incr_h}, let us decompose the increment of $\theta_{v}^{(c)}$ as
    \begin{equation}
        \theta_{v}^{(c)} - \theta_{s}^{(c)} \, | \, \theta_{s}, \theta_{t} = \frac{v-s}{t-s} (\theta_{t}^{(c)} - \theta_{s}^{(c)}) + (B_{v}^{(c)} - B_{s}^{(c)}) \, ,
    \end{equation}
    with endpoints $B_{s}^{(c)} = B_{t}^{(c)} = \boldsymbol{0}$ \citep[Ch.~I, Sec.~3, pp.~35--36]{revuz1999continuous}. Then
    \begin{align}
        & \mathbb{E} [ \| \theta_{v}^{(c)} - \theta_{s}^{(c)} \|^2 \, | \, \theta_{s}, \theta_{t}] \\
        & = \mathbb{E} \left [ \left \| \frac{v-s}{t-s} ( \theta_{t}^{(c)} - \theta_{s}^{(c)} ) + (B_{v}^{(c)} - B_{s}^{(c)}) \right \|^2  \, | \, B_{s}, B_{t} \right ] \\[1em]
        & = \frac{(v-s)^2}{(t-s)^2} \| \theta_{t}^{(c)} - \theta_{s}^{(c)} \|^2 + \mathbb{E} \left [ \| B_{v}^{(c)} - B_{s}^{(c)} \|^2 \, | \, B_{s}, B_{t} \right ] \\
        & \qquad + 2 \frac{v-s}{t-s} ( \theta_{t}^{(c)} - \theta_{s}^{(c)} )^\top \, \mathbb{E} \left [ ( B_{v}^{(c)} - B_{s}^{(c)} ) \, | \, B_{s}, B_{t} \right ] \\[1em]
        & = \frac{(v-s)^2}{(t-s)^2} \| \theta_{t}^{(c)} - \theta_{s}^{(c)} \|^2 + \mathbb{E} \left [ \| B_{v}^{(c)} - B_{s}^{(c)} \|^2 \, | \, B_{s}, B_{t} \right ] \, ,
    \label{eq:squarednorm}
    \end{align}
    where the second equality follows from $\| a + b \|^2 =  (a+b)^\top (a+b) = \| a \|^2 + 2 a^\top b + \| b \|^2$, and the final equality uses the fact that the centred Brownian bridge fluctuation satisfies 
    \begin{equation}
        \mathbb{E} \left[ B_{v}^{(c)} - B_{s}^{(c)} \, \middle| \, B_{s}^{(c)}, B_{t}^{(c)} \right] = \boldsymbol{0} \, .
    \end{equation}
    Using standard properties of Brownian motion, we know that
    \begin{align}
        \mathrm{Var}(B_{v}^{(c)}) 
        & = (v-s) \frac{t-v}{t-s} I_d \, ,
    \end{align}
    and that
    \begin{align}
        & \mathbb{E} [ \| B_{v}^{(c)} - B_{s}^{(c)} \|^2 \, | \, B_{s}, B_{t}] \\
        & = \mathrm{tr} \left ( \frac{(v-s) (t - v)}{t-s} I_d  \right ) + \| \mathbb{E}[B_{v}^{(c)} - B_{s}^{(c)} \, | \, B_{s}, B_{t}] \|^2 \\
        & = \frac{(v-s) (t-v)}{t-s} d \, .
    \end{align}
    Therefore, substituting these results into Equation~\eqref{eq:squarednorm}, we have
    \begin{align}
        & \mathbb{E} [ \| \theta_{v}^{(c)} - \theta_{s}^{(c)} \|^2 \, | \, \theta_{s}, \theta_{t}] \\
        & = \frac{(v-s)^2}{(t-s)^2} \| \theta_{t}^{(c)} - \theta_{s}^{(c)} \|^2 + \mathbb{E} \left [ \| B_{v}^{(c)} - B_{s}^{(c)} \|^2 \, | \, B_{s}, B_{t} \right ] \\
        & = \frac{(v-s)^2}{(t-s)^2} \| \theta_{t}^{(c)} - \theta_{s}^{(c)} \|^2 + \frac{(v-s) (t-v)}{t-s} d \, .
    \end{align}
    Integrating this expression we have 
    \begin{align} \label{eq:int_square_norm_incr_h}
        & \int_s^t \mathbb{E} [ \| \theta_{v}^{(c)} - \theta_{s}^{(c)} \|^2 \, | \, \theta_{s}, \theta_{t}] \, \mathrm{d}v \\
        & = \int_s^t \frac{(v-s)^2}{(t-s)^2} \| \theta_{t}^{(c)} - \theta_{s}^{(c)} \|^2 \, \mathrm{d}v + \int_s^t (v-s) \frac{t-v}{t-s} d \, \mathrm{d}v \\
        & = \frac{1}{(t-s)^2} \| \theta_{t}^{(c)} - \theta_{s}^{(c)} \|^2 \int_s^t (v-s)^2 \, \mathrm{d}v + d \int_s^t (v-s) \, \mathrm{d}v - d \frac{1}{t-s} \int_s^t (v-s)^2 \, \mathrm{d} v \\
        & = \frac{t-s}{3} \| \theta_{t}^{(c)} - \theta_{s}^{(c)} \|^2 + d \frac{(t-s)^2}{2} - d \frac{(t-s)^2}{3} \\
        & = \frac{t-s}{3} \| \theta_{t}^{(c)} - \theta_{s}^{(c)} \|^2 + d \frac{(t-s)^2}{6}
    \end{align}
    
\end{proof}

\begin{lemma} \label{lem:theta_delta_square_norm}
    Let $\theta_t = \{ \theta_t^{(c)} \}_{c=1}^C$ follow the SDE from Equation \eqref{eq:sde} on $t \in [0,T]$. 
    For $0 \le s \leq t \le T$,
    \begin{equation}
        \mathbb{E} \left [ \| \theta_{t}^{(c)} - \theta_{s}^{(c)} \|^2  \, \big | \, \theta_0, \vartheta \right ] = \frac{(t-s)^2}{T^2} \| \vartheta - \theta_0^{(c)} \|^2 + \frac{(t-s) (T - t+s)}{T} d \, .
    \end{equation}
\end{lemma}

\begin{proof}
    Following the same steps as in Lemma \ref{lem:int_square_norm_incr_h}, let us write the increment as
    \begin{equation}
        \theta_{t}^{(c)} - \theta_{s}^{(c)} \, | \, \theta_0, \vartheta = \frac{(t-s)}{T} (\vartheta - \theta_0^{(c)}) + (B_{t}^{(c)} - B_{s}^{(c)}) \, ,
    \end{equation}
    with $B_{0}^{(c)} = B_T^{(c)} = \boldsymbol{0}$, $\forall c$ \citep[Ch.~I, Sec.~3, pp.~35--36]{revuz1999continuous}. Then by decomposing the square of the $2$-norm as $\| a+b \|^2  = \| a \|^2 + \| b \|^2 + 2 a^\top b$, and cancelling the mixed term due to the fact that $\mathbb{E} [ B_{t}^{(c)} - B_{s}^{(c)} \, | \, B_0, B_T] = \boldsymbol{0}$,
    \begin{align}
        & \mathbb{E}[ \| \theta_{t}^{(c)} - \theta_{s}^{(c)} \|^2 \, | \, \theta_0, \vartheta ] \\
        & = \frac{(t-s)^2}{T^2} \| \vartheta - \theta_0^{(c)} \|^2 + \mathbb{E} \left [ \| B_{t}^{(c)} - B_{s}^{(c)} \|^2 \, \big | \, B_0, B_T \right ] \, .
    \end{align}
    As seen in the proof of Lemma \ref{lem:bridge_incr}, from standard arguments we know that 
    \begin{equation}
        B_{t}^{(c)} - B_{s}^{(c)} \big | B_0, B_T \sim \mathcal{N} \left ( 0, \frac{(t-s) (T - t+s)}{T} I_d \right ) \, ,
    \end{equation}
    and therefore,
    \begin{equation}
        \mathbb{E} \left [ \| B_{t}^{(c)} - B_{s}^{(c)} \|^2 \, \big | \, B_0, B_T \right ] = \mathrm{tr} \left ( \frac{(t-s) (T - t+s)}{T} I_d \right ) = \frac{(t-s) (T - t+s)}{T} d \, .
    \end{equation}
    Putting everything together,
    \begin{align}
        \mathbb{E}[ \| \theta_{t}^{(c)} - \theta_{s}^{(c)} \|^2 \, | \, \theta_0, \vartheta ]
        & = \frac{(t-s)^2}{T^2} \| \vartheta - \theta_0^{(c)} \|^2 + \frac{(t-s) (T - t+s)}{T} d \, .
    \end{align}
\end{proof}

\begin{lemma} \label{lem:theta_T_square_norm}
    Let $\{\theta_{0}^{(c)}\}_{c=1}^C$ be a collection of mutually independent $d$-dimensional random vectors with marginal densities $\{f_c\}_{c=1}^C$, where $f_c$ is an unknown distribution with finite mean $\mu^{(c)}$ and covariance $Q^{(c)}$, as per Assumption~\ref{as:smoothness}.
    Further, let $\theta_t = \{ \theta_t^{(c)} \}_{c=1}^C$ follow the SDE from Equation \eqref{eq:sde} on $t \in [0,T]$. Then
    \begin{equation}
        \mathbb{E} \| \vartheta - \theta_0^{(c)} \|^2 \le d \frac{T}{C} + 2 d \lambda^{(c)} \, ,
    \end{equation}
    where we defined the constant $\lambda^{(c)}$ as
    \begin{equation}
        \lambda^{(c)} = d^{-1} \left ( \mathrm{tr}(Q^{(c)}) + \| \mu^{(c)} \|^2 + C^{-1} \mathrm{tr}(\bar{Q}) + \| \bar{\mu} \|^2 \right ) \, ,
    \end{equation}
    to explicitly show the linear dependence on $d$, and where the quantities $\mu^{(c)}$, $\bar{\mu}$, $Q^{(c)}$, and  $\bar{Q}$ are moments of $\theta_0$ and $\bar{\theta}_0$, and are defined as $\bar{\mu} := C^{-1} \sum_{c=1}^C \mu^{(c)}$, $\bar{Q} := C^{-1} \sum_{c=1}^C Q^{(c)}$, and $\bar{\theta}_0 = C^{-1} \sum_{c=1}^C \theta_0^{(c)}$, as in Lemma \ref{lem:theta_0}.
\end{lemma}

\begin{remark}
    Under mild assumptions, the trace terms appearing in $\lambda^{(c)}$ are expected to grow proportionally to the dimension $d$. The factor $d^{-1}$ keeps $\lambda^{(c)}$ of comparable magnitude as $d$ varies.
\end{remark}

\begin{proof}
    Following similar steps to the proof of Lemmas \ref{lem:theta_0} and \ref{lem:theta_T}, we can write
    \begin{align} \label{eq:expected_norm_increment_T}
        \mathbb{E} \| \vartheta - \theta_0^{(c)} \|^2 
        & = \mathbb{E} \| \vartheta - \bar{\theta}_0 + \bar{\theta}_0 - \theta_0^{(c)} \|^2 \\[1em]
        & = \mathbb{E} \| \vartheta - \bar{\theta}_0 \|^2 + \mathbb{E} \| \theta_0^{(c)} - \bar{\theta}_0 \|^2 \\
        & \qquad + 2 \mathbb{E} [ ( \vartheta - \bar{\theta}_0 )^\top ( \bar{\theta}_0 - \theta_0^{(c)} ) ] \\[1em]
        & = \mathbb{E} \| \vartheta - \bar{\theta}_0 \|^2 + \mathbb{E} \| \theta_0^{(c)} - \bar{\theta}_0 \|^2 \\
        & \qquad + 2 \mathbb{E} [ ( \bar{\theta}_0 - \theta_0^{(c)} )^\top \mathbb{E} [ ( \vartheta - \bar{\theta}_0 ) \, | \, \theta_0] ] \\[1em]
        & \le d \frac{T}{C} + 2 \mathbb{E} \| \theta_0^{(c)} \|^2 + 2 \mathbb{E} \| \bar{\theta}_0 \|^2 \\
        & = d \frac{T}{C} + 2 (\mathrm{tr}(Q^{(c)}) + \| \mu^{(c)} \|^2 ) + 2 (C^{-1} \mathrm{tr}(\bar{Q}) + \| \bar{\mu} \|^2 ) \\
        & = d \frac{T}{C} + 2 d \lambda^{(c)} \, ,
    \end{align}
    where we added and subtracted $\bar{\theta}_0$, decomposed the square of the $2$-norm as $\| a+b \|^2  = \| a \|^2 + \| b \|^2 + 2 a^\top b$, we cancelled the mixed term as $\mathbb{E} [ ( \vartheta - \bar{\theta}_0 ) \, | \, \theta_0] = \boldsymbol{0}$, we used
    the triangle inequality and the $C_p$ inequality, and we defined the constant $\lambda^{(c)}$ as
    \begin{equation}
        \lambda^{(c)} = d^{-1} \left ( \mathrm{tr}(Q^{(c)}) + \| \mu^{(c)} \|^2 + C^{-1} \mathrm{tr}(\bar{Q}) + \| \bar{\mu} \|^2 \right ) \, ,
    \end{equation}
    to explicitly show the linear dependence on $d$.
\end{proof}

\subsection{Proof of Proposition \ref{prop:weak_error} (Weak error bound)} \label{sec:weak_error_proof}

\begin{proof}[Proof of Proposition \ref{prop:weak_error}]
    Let us decompose the weak error as follows
    \begin{align}
        \eta_\Delta 
        & = \left | \mathbb{E} \left [ \sum_{c=1}^C \int_0^T \phi^{(c)}(\theta_s^{(c)}) \, \mathrm{d} s - \sum_{c=1}^C \int_0^T \phi_\Delta^{(c)}(s, \theta_{t_{0:n}}^{(c)}) \, \mathrm{d} s \right ] \right | \\
        & = \left | \sum_{c=1}^C \mathbb{E} \left [ \int_0^T \phi^{(c)}(\theta_s^{(c)}) \, \mathrm{d} s - \int_0^T \phi_\Delta^{(c)}(s, \theta_{t_{0:n}}^{(c)}) \, \mathrm{d} s \right ] \right | \\
        & = \left | \sum_{c=1}^C \mathbb{E} \left [ \int_0^T \phi^{(c)}(\theta_s^{(c)}) \, \mathrm{d} s - \sum_{j=0}^{n-1} \frac{\Delta}{2} ( \phi^{(c)}(\theta_{t_j}^{(c)}) + \phi^{(c)}(\theta_{t_{j+1}}^{(c)}) ) \right ] \right | \\
        & = \Bigg | \sum_{c=1}^C \sum_{j=0}^{n-1} \mathbb{E} \Bigg [ \mathbb{E} \Bigg [ \int_0^\Delta \phi^{(c)}(\theta_{t_j+h}^{(c)}) \, \mathrm{d} h \\
        & \qquad \qquad \qquad \qquad - \frac{\Delta}{2} ( \phi^{(c)}(\theta_{t_j}^{(c)}) + \phi^{(c)}(\theta_{t_{j+1}}^{(c)}) ) \; | \; \theta_{t_j}, \theta_{t_{j+1}} \Bigg ] \Bigg ] \Bigg  | \\
        & =: \left | \sum_{c=1}^C \sum_{j=0}^{n-1} \mathbb{E} \left [ K_j^{(c)} (\Delta) \right ] \right | \\
        & =: \left | G(\Delta) \right | \, ,
    \end{align}
    where $K_j^{(c)} (\Delta)$ is the expected integration error relative to a single shard $c$ and single subinterval from $t_j$ to $t_{j+1} = t_j + \Delta$:
    \begin{align}
        & K_j^{(c)} (\Delta) \\
        & := \mathbb{E} \left [ \int_0^\Delta \phi^{(c)}(\theta_{t_{j}+h}^{(c)}) \, \mathrm{d}h - \frac{\Delta}{2}(\phi^{(c)}(\theta_{t_{j}}^{(c)}) +\phi^{(c)}(\theta_{t_{j+1}}^{(c)})) \; \Big | \, \theta_{t_{j}}, \theta_{t_{j+1}} \right ] \\
        & = \mathbb{E} \left [ \int_0^\Delta \phi^{(c)}(\theta_{t_{j}+h}^{(c)}) \, \mathrm{d}h \; \Big | \, \theta_{t_{j}}, \theta_{t_{j+1}} \right ] - \frac{\Delta}{2}(\phi^{(c)}(\theta_{t_{j}}^{(c)}) +\phi^{(c)}(\theta_{t_{j+1}}^{(c)})). \label{eq:weak_err_interval}
    \end{align}
    We bound each $K_j^{(c)}(\Delta)$ in absolute value. 
    In order to do that, we write the Taylor expansions with Lagrange-form remainders of $\phi^{(c)}(\theta_{t_{j}+h}^{(c)})$ and $\phi^{(c)}(\theta_{t_{j+1}}^{(c)})$ about $\theta_{t_{j}}^{(c)}$, which are, respectively,
    \begin{align} \label{eq:taylor_h}
        \phi^{(c)}(\theta_{t_{j}+h}^{(c)}) 
        & = \phi^{(c)}(\theta_{t_{j}}^{(c)}) + \nabla \phi^{(c)}(\theta_{t_{j}}^{(c)})^\top (\theta_{t_{j}+h}^{(c)} - \theta_{t_{j}}^{(c)}) \\
        & \qquad + \frac{1}{2} (\theta_{t_{j}+h}^{(c)} - \theta_{t_{j}}^{(c)})^\top \nabla^2 \phi^{(c)} (\xi_h) (\theta_{t_{j}+h}^{(c)} - \theta_{t_{j}}^{(c)}) \, ,
    \end{align}
    and
    \begin{align} \label{eq:taylor_delta}
        \phi^{(c)}(\theta_{t_{j+1}}^{(c)}) 
        & = \phi^{(c)}(\theta_{t_{j}}^{(c)}) + \nabla \phi^{(c)}(\theta_{t_{j}}^{(c)})^\top (\theta_{t_{j+1}}^{(c)} - \theta_{t_{j}}^{(c)}) \\
        & \qquad + \frac{1}{2} (\theta_{t_{j+1}}^{(c)} - \theta_{t_{j}}^{(c)})^\top \nabla^2 \phi^{(c)} (\xi_\Delta) (\theta_{t_{j+1}}^{(c)} - \theta_{t_{j}}^{(c)}) \, ,
    \end{align}
    with $\xi_h \in [\theta_{t_{j}}^{(c)}, \theta_{t_{j}+h}^{(c)}]$, and $\xi_\Delta \in [\theta_{t_{j}}^{(c)}, \theta_{t_{j+1}}^{(c)}]$. 
    Let us now substitute them into the definition of $K_j^{(c)}(\Delta)$.
    \begin{align}
        & K_j^{(c)} (\Delta) \\
        & = \mathbb{E} \left [ \int_0^\Delta \phi^{(c)}(\theta_{t_{j}+h}^{(c)}) \, \mathrm{d}h \; \Big | \, \theta_{t_{j}}, \theta_{t_{j+1}} \right ] - \frac{\Delta}{2}(\phi^{(c)}(\theta_{t_{j}}^{(c)}) +\phi^{(c)}(\theta_{t_{j+1}}^{(c)})) \\[1em]
        & = \nabla \phi^{(c)} (\theta_{t_{j}}^{(c)})^\top \int_0^\Delta \mathbb{E} [\theta_{t_{j}+h}^{(c)} - \theta_{t_{j}}^{(c)} \, | \, \theta_{t_{j}}, \theta_{t_{j+1}} ] \, \mathrm{d}h \\
        & \quad + \frac{1}{2} \int_0^\Delta \mathbb{E} [ (\theta_{t_{j}+h}^{(c)} - \theta_{t_{j}}^{(c)})^\top \nabla^2 \phi^{(c)} (\xi_h) (\theta_{t_{j}+h}^{(c)} - \theta_{t_{j}}^{(c)}) \, | \, \theta_{t_{j}}, \theta_{t_{j+1}}] \, \mathrm{d} h \\
        & \quad - \frac{\Delta}{2} \nabla \phi^{(c)}(\theta_{t_{j}}^{(c)})^\top (\theta_{t_{j+1}}^{(c)} - \theta_{t_{j}}^{(c)}) \\
        & \quad - \frac{1}{2} \frac{\Delta}{2} (\theta_{t_{j+1}}^{(c)} - \theta_{t_{j}}^{(c)})^\top \nabla^2 \phi^{(c)} (\xi_\Delta) (\theta_{t_{j+1}}^{(c)} - \theta_{t_{j}}^{(c)}) \, ,
    \end{align}
    where we applied Fubini's theorem (Assumption \ref{as:fubini}), the zero-order term cancel out immediately, and we are only left with the first and second order terms. Using Lemma \ref{lem:int_incr_h}, the first order terms cancel out too, leaving us with only the second order terms, i.e.\
    \begin{align} \label{eq:E_second_order_terms}
        & K_j^{(c)} (\Delta) \\
        & = \frac{1}{2} \int_0^\Delta \mathbb{E} [ (\theta_{t_{j}+h}^{(c)} - \theta_{t_{j}}^{(c)})^\top \nabla^2 \phi^{(c)} (\xi_h) (\theta_{t_{j}+h}^{(c)} - \theta_{t_{j}}^{(c)}) \, | \, \theta_{t_{j}}, \theta_{t_{j+1}}] \, \mathrm{d} h \\
        & \quad - \frac{1}{2} \frac{\Delta}{2} (\theta_{t_{j+1}}^{(c)} - \theta_{t_{j}}^{(c)})^\top \nabla^2 \phi^{(c)} (\xi_\Delta) (\theta_{t_{j+1}}^{(c)} - \theta_{t_{j}}^{(c)}) \, .
    \end{align}
    
    By Assumption~\ref{as:weak_error_regularity}, for every $v, \xi \in \mathbb{R}^d$,
    \begin{equation}
        \| \nabla^2 \phi^{(c)}(\xi) \, v \| = \lim_{h \to 0} \frac{\| \nabla \phi^{(c)}(\xi + h v) - \nabla \phi^{(c)}(\xi) \|}{|h|} \le U \| v \| \, ,
    \end{equation}
    so that, by the Cauchy--Schwarz inequality, $| x^\top \nabla^2 \phi^{(c)}(\xi) \, x | \le \| x \| \, \| \nabla^2 \phi^{(c)}(\xi) \, x \| \le U \| x \|^2$ for every $x, \xi \in \mathbb{R}^d$. Applying this bound to both terms of Equation~\eqref{eq:E_second_order_terms}, together with the triangle inequality and $| \mathbb{E}[Y] | \le \mathbb{E} | Y |$, gives
    \begin{align}
        | K_j^{(c)} (\Delta) |
        & \le \frac{1}{2} U \int_0^\Delta \mathbb{E} [ \| \theta_{t_{j}+h}^{(c)} - \theta_{t_{j}}^{(c)} \|^2 \, | \, \theta_{t_{j}}, \theta_{t_{j+1}}] \, \mathrm{d} h + \frac{1}{2} U \frac{\Delta}{2} \| \theta_{t_{j+1}}^{(c)} - \theta_{t_{j}}^{(c)} \|^2 \\
        & = \frac{1}{2} U \frac{\Delta}{3} \| \theta_{t_{j+1}}^{(c)} - \theta_{t_{j}}^{(c)} \|^2 + \frac{1}{2} U \frac{\Delta^2}{6} d + \frac{1}{2} U \frac{\Delta}{2} \| \theta_{t_{j+1}}^{(c)} - \theta_{t_{j}}^{(c)} \|^2 \\
        & = \frac{5}{12} U \Delta \| \theta_{t_{j+1}}^{(c)} - \theta_{t_{j}}^{(c)} \|^2 + \frac{1}{12} U \Delta^2 d \, ,
    \end{align}
    where the second line uses Lemma~\ref{lem:int_square_norm_incr_h}, and the third collects the terms in $\| \theta_{t_{j+1}}^{(c)} - \theta_{t_{j}}^{(c)} \|^2$, using $\frac{1}{3} + \frac{1}{2} = \frac{5}{6}$. Taking expectations, and using Lemmas~\ref{lem:theta_delta_square_norm} and~\ref{lem:theta_T_square_norm} through the law of total expectation,
    \begin{align}
        \mathbb{E} | K_j^{(c)} (\Delta) |
        & \le \frac{5}{12} U \Delta \, \mathbb{E} \| \theta_{t_{j+1}}^{(c)} - \theta_{t_{j}}^{(c)} \|^2 + \frac{1}{12} U \Delta^2 d \\
        & = \frac{5}{12} U \Delta \, \mathbb{E} \left [ \frac{\Delta^2}{T^2} \| \vartheta - \theta_0^{(c)} \|^2 + \frac{\Delta (T - \Delta)}{T} d \right ] + \frac{1}{12} U \Delta^2 d \\
        & \le \frac{5}{12} U \Delta \left ( \frac{\Delta^2}{T^2} \left ( d \frac{T}{C} + 2 d \lambda^{(c)} \right ) + \frac{\Delta (T - \Delta)}{T} d \right ) + \frac{1}{12} U \Delta^2 d \\
        & = \frac{5}{12} U \Delta^2 d + \frac{1}{12} U \Delta^2 d + \mathcal{O}(\Delta^3) \\
        & = \frac{1}{2} U \Delta^2 d + \mathcal{O}(\Delta^3) \, ,
    \end{align}
    where the fourth line collects the terms of order $\Delta^3$, which include the terms in $T / C$ and in $\lambda^{(c)}$. Finally, summing over the $C$ shards and the $n$ intervals, and using $n \Delta = T$,
    \begin{align}
        \eta_\Delta
        & = | G(\Delta) | \\
        & \le \sum_{c=1}^C \sum_{j=0}^{n-1} \mathbb{E} | K_j^{(c)} (\Delta) | \\
        & \le C n \left ( \frac{1}{2} U \Delta^2 d + \mathcal{O}(\Delta^3) \right ) \\
        & = \frac{1}{2} C T \Delta d U + \mathcal{O}(\Delta^2) \\
        & =: H_\Delta \, .
    \end{align}
\end{proof}

\section{Relaxing Assumption~\ref{as:phi_lip} (Lipschitzness of \texorpdfstring{$\phi^{(c)}$}{phi(c)})} \label{app:weaker_assumptions}

The global Lipschitz condition in Assumption \ref{as:phi_lip} is a convenient way to control increments of the potential along the bridge, though it is not essential for the $\sqrt{\Delta}$ dependence of the strong-error bound. 
Its main convenience is that it allows us to keep the analysis as simple as possible. The downside is that it excludes models with a $\phi^{(c)}$ that grows quadratically, e.g.\ a Gaussian distribution.

In this appendix we replace Assumption \ref{as:phi_lip}, which postulates that $\phi^{(c)}$ is Lipschitz, with Assumption~\ref{as:weak_error_regularity} of Section~\ref{sec:controlling_error}, which postulates that the gradient of $\phi^{(c)}$ is $U$-Lipschitz, allowing more flexibility. This is the assumption under which Proposition~\ref{prop:weak_error} bounds the weak error.

Unlike Assumption \ref{as:phi_lip}, which gives a uniform bound on $\|\nabla\phi^{(c)}\|$, Assumption \ref{as:weak_error_regularity} allows $\|\nabla\phi^{(c)}(\theta)\|$ to grow at most linearly in $\|\theta\|$. We control this growth using moments of the Brownian bridge.

The proof proceeds in three steps. First, we obtain a uniform second-moment bound for the position of the bridge. We then combine this bound with the Lipschitz continuity of the gradient to control increments of $\phi^{(c)}$. Finally, we integrate this increment bound over the time partition to obtain a strong-error bound, which in turn yields a total variation bound.

Recall the constant $\lambda^{(c)}$ defined in Lemma \ref{lem:theta_T_square_norm}. For each $c = 1,\ldots,C$, define
\begin{align}
    \mathcal{G}^{(c)} & := d^{-1/2} \left\| \nabla \phi^{(c)}(0) \right\| \, , \\
    \mathcal{R}^{(c)} & := \frac{T}{C} + 2 \lambda^{(c)} \, , \\
    \mathcal{K}^{(c)} & := \mathcal{G}^{(c)} + U \left( \lambda^{(c)} + \frac{T}{C} + \frac{T}{4} \right)^{1/2} \, . \label{eq:relaxed_constants}
\end{align}
The quantity $\mathcal{G}^{(c)}$ controls the gradient at the origin, $\mathcal{R}^{(c)}$ controls increments of the bridge, and $\mathcal{K}^{(c)}$ controls the gradient along the bridge. For fixed $d$, $\mathcal{G}^{(c)}$ is finite under Assumption \ref{as:weak_error_regularity}.

We first bound the second moment of the position of the process, uniformly in time.

\begin{lemma}[Bridge position moment] \label{lem:bridge_position_moment_relaxed}
    Let $\theta_t$ follow the SDE from Equation \eqref{eq:sde}, and suppose Assumption \ref{as:smoothness} holds. Then, for all $c = 1,\ldots,C$ and $t \in [0,T]$,
    \begin{equation}
        \mathbb{E} \| \theta_t^{(c)} \|^2 \le d \left( \lambda^{(c)} + \frac{T}{C} + \frac{T}{4} \right) \, . \label{eq:bridge_position_moment_relaxed}
    \end{equation}
\end{lemma}

\begin{proof}
    Conditional on $\theta_0$ and $\vartheta$, the $c$th component can be written as a Brownian bridge from $\theta_0^{(c)}$ to $\vartheta$:
    \begin{equation} \label{eq:bridge_decomposition_relaxed}
        \theta_t^{(c)} = \left( 1 - \frac{t}{T} \right) \theta_0^{(c)} + \frac{t}{T} \vartheta + B_t^{(c)} \, ,
    \end{equation}
    where $B_0^{(c)} = B_T^{(c)} = \boldsymbol{0}$ and
    \begin{equation}
        \mathbb{E} \left[ B_t^{(c)} \mid \theta_0,\vartheta \right] = \boldsymbol{0} \, , \qquad \mathbb{E} \left[ \| B_t^{(c)} \|^2 \mid \theta_0,\vartheta \right] = d \frac{t(T-t)}{T} \, .
    \end{equation}
    Using Equation \eqref{eq:bridge_decomposition_relaxed}, we have
    \begin{align}
        \mathbb{E} \| \theta_t^{(c)} \|^2
        & = \mathbb{E} \left\| \left( 1 - \frac{t}{T} \right) \theta_0^{(c)} + \frac{t}{T} \vartheta + B_t^{(c)} \right\|^2 \\
        & = \mathbb{E} \left\| \left( 1 - \frac{t}{T} \right) \theta_0^{(c)} + \frac{t}{T} \vartheta \right\|^2 + \mathbb{E} \| B_t^{(c)} \|^2 \\
        & \le \left( 1 - \frac{t}{T} \right) \mathbb{E} \| \theta_0^{(c)} \|^2 + \frac{t}{T} \mathbb{E} \| \vartheta \|^2 + d \frac{t(T-t)}{T} \\
        & \le \mathbb{E} \| \theta_0^{(c)} \|^2 + \mathbb{E} \| \vartheta \|^2 + d \frac{T}{4} \, . \label{eq:bridge_position_decomposition_relaxed}
    \end{align}
    On the second line of Equation \eqref{eq:bridge_position_decomposition_relaxed}, the mixed term vanishes because the linear interpolation is a function of $\theta_0$ and $\vartheta$, whereas the centred bridge has conditional mean zero. On the third line, we used the convexity of $\theta \mapsto \| \theta \|^2$. On the final line, we used $1-t/T \le 1$, $t/T \le 1$, and $t(T-t) \le T^2/4$.

    It remains to control the two moment terms on the final line of Equation \eqref{eq:bridge_position_decomposition_relaxed}. As shown in Lemma \ref{lem:theta_T_conditional}, 
    \begin{equation}
        \vartheta \mid \theta_0 \sim \mathcal{N} \left( \bar{\theta}_0, \frac{T}{C} I_d \right) \, .
    \end{equation}
    Hence, for a standard $d$-dimensional normal random vector $\Upsilon$ independent of $\theta_0$, we have
    \begin{align}
        \mathbb{E} \| \vartheta \|^2
        & = \mathbb{E} \left\| \bar{\theta}_0 + \sqrt{\frac{T}{C}} \, \Upsilon \right\|^2 \\
        & = \mathbb{E} \| \bar{\theta}_0 \|^2 + 2 \sqrt{\frac{T}{C}} \, \mathbb{E} \left[ \bar{\theta}_0^\top \Upsilon \right] + \frac{T}{C} \mathbb{E} \| \Upsilon \|^2 \\
        & = \mathbb{E} \| \bar{\theta}_0 \|^2 + \frac{dT}{C} \, . \label{eq:theta_T_second_moment_relaxed}
    \end{align}
    On the final line, the mixed term vanishes because $\Upsilon$ is independent of $\theta_0$ and has mean zero, and we used $\mathbb{E} \| \Upsilon \|^2 = d$.

    Substituting Equation \eqref{eq:theta_T_second_moment_relaxed} into Equation \eqref{eq:bridge_position_decomposition_relaxed} gives
    \begin{equation}
        \mathbb{E} \| \theta_t^{(c)} \|^2 \le \mathbb{E} \| \theta_0^{(c)} \|^2 + \mathbb{E} \| \bar{\theta}_0 \|^2 + \frac{dT}{C} + \frac{dT}{4} \, .
    \end{equation}
    By the definition of $\lambda^{(c)}$ in Lemma \ref{lem:theta_T_square_norm},
    \begin{equation}
        \mathbb{E} \| \theta_0^{(c)} \|^2 + \mathbb{E} \| \bar{\theta}_0 \|^2 = d \lambda^{(c)} \, .
    \end{equation}
    Therefore,
    \begin{equation}
        \mathbb{E} \| \theta_t^{(c)} \|^2 \le d \left( \lambda^{(c)} + \frac{T}{C} + \frac{T}{4} \right) \, ,
    \end{equation}
    which proves the result.
\end{proof}

Lemma~\ref{lem:bridge_position_moment_relaxed} controls the size of the gradient along the bridge. Combining this with a second-moment bound for bridge increments gives the following bound on increments of the potential.

\begin{lemma}[Potential increment bound] \label{lem:potential_increment_bound_relaxed}
    Let $\theta_t$ follow the SDE from Equation \eqref{eq:sde}, let $0 < \Delta \le T$, let $n \Delta = T$, and let $t_j = j \Delta$. Under Assumptions \ref{as:smoothness} and \ref{as:weak_error_regularity}, for $c = 1,\ldots,C$, $j = 0,\ldots,n-1$, and $0 \le h \le \Delta$,
    \begin{align}
        & \mathbb{E} \left| \phi^{(c)}(\theta_{t_j+h}^{(c)}) - \phi^{(c)}(\theta_{t_j}^{(c)}) \right| \\
        & \le d \, \mathcal{K}^{(c)} \left( \sqrt{h} + \frac{h}{T} \sqrt{\mathcal{R}^{(c)}} \right) + \frac{Ud}{2} \left( h + \frac{h^2}{T^2} \mathcal{R}^{(c)} \right) \, . \label{eq:potential_increment_bound_relaxed}
    \end{align}
\end{lemma}

\begin{proof}
    For any $\theta,\theta'\in\mathbb{R}^d$, Taylor's formula with Lagrange-form remainder gives
    \begin{align}
        \phi^{(c)}(\theta') - \phi^{(c)}(\theta)
        & = \nabla \phi^{(c)}(\theta)^\top (\theta'-\theta) \\
        & \quad + \frac{1}{2} (\theta'-\theta)^\top \nabla^2 \phi^{(c)}(\xi) (\theta'-\theta) \, ,
    \end{align}
    for some $\xi$ on the line segment between $\theta$ and $\theta'$.
    Since by Assumption \ref{as:weak_error_regularity} $\nabla\phi^{(c)}$ is $U$-Lipschitz and $\phi^{(c)} \in C^2(\mathbb R^d)$, the quadratic remainder is bounded in absolute value by $\frac{U}{2} \|\theta'-\theta\|^2$. Therefore,
    \begin{equation} \label{eq:phi_increment_hessian_bound_relaxed}
        \left | \phi^{(c)}(\theta') - \phi^{(c)}(\theta) \right | \le \| \nabla \phi^{(c)}(\theta) \| \, \| \theta'-\theta \| + \frac{U}{2} \| \theta'-\theta \|^2 \, .
    \end{equation}
    Substituting $\theta_{t_j+h}^{(c)},  \theta_{t_j}^{(c)}$ for $\theta',\theta$ and adding the expectation, we have
    \begin{align}
        & \mathbb{E} \left | \phi^{(c)}(\theta_{t_j+h}^{(c)}) - \phi^{(c)}(\theta_{t_j}^{(c)}) \right | \\
        & \le \mathbb{E} \left[ \|\nabla \phi^{(c)}(\theta_{t_j}^{(c)}) \| \, \|\theta_{t_j+h}^{(c)}-\theta_{t_j}^{(c)}\| \right] + \frac{U}{2} \mathbb{E} \|\theta_{t_j+h}^{(c)}-\theta_{t_j}^{(c)}\|^2 \\
        & \le \left( \mathbb{E} \| \nabla \phi^{(c)}(\theta_{t_j}^{(c)}) \|^2 \right)^{1/2} \left( \mathbb{E} \|\theta_{t_j+h}^{(c)} - \theta_{t_j}^{(c)}\|^2 \right)^{1/2} + \frac{U}{2} \mathbb{E} \|\theta_{t_j+h}^{(c)} - \theta_{t_j}^{(c)}\|^2 \, , \label{eq:exp_incr_phi}
    \end{align}
    where on the second line we used the Cauchy--Schwarz inequality on the product in the first term.
    We need to bound each term on the RHS of this equation.

    \vspace{1em}
    \noindent \textbf{Controlling the gradient}

    \noindent First, we control the gradient term. Assumption \ref{as:weak_error_regularity} implies that $\nabla \phi^{(c)}$ is $U$-Lipschitz.
    Taking one endpoint to be $0$, $\forall t \in [0,T]$,
    \begin{equation}
        \| \nabla \phi^{(c)}(\theta_t^{(c)}) - \nabla \phi^{(c)}(0) \| \le U \| \theta_t^{(c)} \| \, ,
    \end{equation}
    and by the triangle inequality,
    \begin{equation} \label{eq:bound_gradient}
        \| \nabla \phi^{(c)}(\theta_t^{(c)}) \| \le \| \nabla \phi^{(c)}(0) \| + U \| \theta_t^{(c)} \| \, .
    \end{equation}
    In order to bound the relevant quantity, we take the $L_2$ norm of both sides. By Minkowski's inequality,
    \begin{align}
        \left( \mathbb{E}\| \nabla \phi^{(c)}(\theta_t^{(c)}) \|^2 \right)^{1/2}
        & \le \left( \mathbb{E} \left[ \left( \|\nabla \phi^{(c)}(0)\| + U\|\theta_t^{(c)}\| \right)^2 \right] \right)^{1/2} \\
        & \le \left( \mathbb{E} \| \nabla \phi^{(c)}(0) \|^2 \right)^{1/2} + U \left( \mathbb{E}\|\theta_t^{(c)}\|^2 \right)^{1/2} \\
        & = \| \nabla \phi^{(c)}(0) \| + U \left( \mathbb{E}\|\theta_t^{(c)}\|^2 \right)^{1/2} \, ,
    \end{align}
    where the last line follows as $\| \nabla \phi^{(c)}(0) \|$ is a deterministic quantity.
    Using Equation \eqref{eq:relaxed_constants} and Lemma \ref{lem:bridge_position_moment_relaxed}, we can write
    \begin{align} 
        \left( \mathbb{E} \| \nabla \phi^{(c)}(\theta_t^{(c)}) \|^2 \right)^{1/2}
        & \le \| \nabla \phi^{(c)}(0) \| + U \left( \mathbb{E}\|\theta_t^{(c)}\|^2 \right)^{1/2} \\
        & \le \sqrt{d} \mathcal{G}^{(c)} +  U \sqrt{d} \left( \lambda^{(c)} + \frac{T}{C} + \frac{T}{4} \right)^{1/2} \\
        & = \sqrt{d} \left[ \mathcal{G}^{(c)} + U \left( \lambda^{(c)} + \frac{T}{C} + \frac{T}{4} \right)^{1/2} \right] \\
        & = \sqrt{d}\,\mathcal{K}^{(c)} \, . \label{eq:bound_gradient_minkowski}
    \end{align}

    \vspace{1em}
    \noindent \textbf{Controlling the increment}

    \noindent We next control the increments of the bridge. By the law of total expectation, Lemma~\ref{lem:theta_delta_square_norm} with $s = t_j$ and $t = t_j + h$, and Lemma~\ref{lem:theta_T_square_norm},
    \begin{align}
        \mathbb{E}\|\theta_{t_j+h}^{(c)} - \theta_{t_j}^{(c)}\|^2
        & = \frac{h^2}{T^2} \mathbb{E}\|\vartheta-\theta_0^{(c)}\|^2 + d h \frac{T-h}{T} \\
        & \le \frac{h^2}{T^2} \, d \left( \frac{T}{C} + 2 \lambda^{(c)} \right) + d h \frac{T-h}{T} \\
        & \le d \left( \frac{h^2}{T^2}\mathcal{R}^{(c)} + h \right) \, , \label{eq:bound_diff_2norm_sq_exp}
    \end{align}
    where the first line takes the expectation of Lemma~\ref{lem:theta_delta_square_norm}, the second uses Lemma~\ref{lem:theta_T_square_norm}, and the third uses the definition of $\mathcal{R}^{(c)}$ in Equation~\eqref{eq:relaxed_constants} and $(T-h)/T \le 1$.
    Hence, using the subadditivity of the square root,
    \begin{equation} \label{eq:increment_l2_bound}
        \left( \mathbb{E}\|\theta_{t_j+h}^{(c)} - \theta_{t_j}^{(c)}\|^2 \right)^{1/2} 
        \le \sqrt{d} \left( \sqrt{h} + \frac{h}{T} \sqrt{\mathcal{R}^{(c)}} \right) \, .
    \end{equation}
    These are our bounds for the increment. \\[1em]

    \noindent Finally, substituting Equation \eqref{eq:bound_gradient_minkowski} with $t=t_j$, \eqref{eq:bound_diff_2norm_sq_exp}, and \eqref{eq:increment_l2_bound} into Equation \eqref{eq:exp_incr_phi} we obtain
    \begin{align}
        & \mathbb{E} \left | \phi^{(c)}(\theta_{t_j+h}^{(c)}) - \phi^{(c)}(\theta_{t_j}^{(c)}) \right | \\
        & \le \left( \mathbb{E} \| \nabla \phi^{(c)}(\theta_{t_j}^{(c)}) \|^2 \right)^{1/2} \cdot \left( \mathbb{E} \|\theta_{t_j+h}^{(c)} - \theta_{t_j}^{(c)}\|^2 \right)^{1/2} + \frac{U}{2} \mathbb{E} \|\theta_{t_j+h}^{(c)} - \theta_{t_j}^{(c)}\|^2 \\
        & \le \sqrt{d}\,\mathcal{K}^{(c)} \cdot \sqrt{d} \left( \sqrt{h} + \frac{h}{T}\sqrt{\mathcal{R}^{(c)}} \right) + \frac{U}{2} d \left( h + \frac{h^2}{T^2}\mathcal{R}^{(c)} \right) \\
        & = d \, \mathcal{K}^{(c)} \left( \sqrt{h} + \frac{h}{T}\sqrt{\mathcal{R}^{(c)}} \right) + \frac{Ud}{2} \left( h + \frac{h^2}{T^2}\mathcal{R}^{(c)} \right) \, .
    \end{align}
\end{proof}

The preceding lemma replaces the direct potential-increment estimate obtained from Assumption~\ref{as:phi_lip} in the original strong-error proof. Applying it on every interval of the partition gives the following analogue of Proposition~\ref{prop:strong_error}.

\begin{proposition}[Strong error bound] \label{prop:strong_error_relaxed}
    Let the strong error be defined as
    \begin{align}
        \epsilon_\Delta := \mathbb{E} \left | \sum_{c=1}^C \int_0^T \phi^{(c)}(\theta_s^{(c)}) \, \mathrm{d} s - \sum_{c=1}^C \int_0^T \phi_\Delta^{(c)}(s, \theta_{t_{0:n}}^{(c)}) \, \mathrm{d} s \right | \, .
    \end{align}
    Let $0<\Delta\le T$ and $n\Delta=T$, and let $\mathcal{R}^{(c)}$ and $\mathcal{K}^{(c)}$ be defined as in Equation~\eqref{eq:relaxed_constants}.
    Then, under Assumptions \ref{as:smoothness}, \ref{as:fubini}, and \ref{as:weak_error_regularity}, there exists a bound $E_\Delta \ge \epsilon_\Delta$ such that
    \begin{align} \label{eq:strong_error_bound_relaxed}
        E_\Delta
        =
        d \sum_{c=1}^C \left[ \mathcal{K}^{(c)} \left( \frac{7}{6}T\sqrt{\Delta} + \Delta\sqrt{\mathcal{R}^{(c)}} \right) + U \left( \frac{1}{2}T\Delta + \frac{5}{12T}\mathcal{R}^{(c)}\Delta^2 \right) \right] \, .
    \end{align}
\end{proposition}

\begin{proof}[Proof of Proposition \ref{prop:strong_error_relaxed}]
    Starting from the expression of the strong error, we can apply the triangle inequality
    and Fubini's theorem, relying on Assumption \ref{as:fubini}, to obtain
    \begin{align}
        \epsilon_\Delta 
        & = \mathbb{E} \left [ \left | \sum_{c=1}^C \int_0^T \phi^{(c)}(\theta_s^{(c)}) \, \mathrm{d} s - \sum_{c=1}^C \int_0^T \phi_\Delta^{(c)}(s, \theta_{t_{0:n}}^{(c)}) \, \mathrm{d} s \right | \right ] \\
        & = \mathbb{E}\left [ \left | \sum_{c=1}^C \int_0^T \left ( \phi^{(c)}(\theta_s^{(c)}) - \phi^{(c)}_\Delta(s, \theta_{t_{0:n}}^{(c)}) \right ) \, \mathrm{d} s\right | \right ] \\
        & \le \mathbb{E}\left [ \sum_{c=1}^C \int_0^T \left | \phi^{(c)}(\theta_s^{(c)}) - \phi^{(c)}_\Delta(s, \theta_{t_{0:n}}^{(c)}) \right | \, \mathrm{d}s \right ] \\
        & = \sum_{c=1}^C \int_0^T \mathbb{E} \left | \phi^{(c)}(\theta_s^{(c)}) - \phi^{(c)}_\Delta(s, \theta_{t_{0:n}}^{(c)}) \right | \, \mathrm{d}s \, .
    \end{align}
    Let us divide the integral on each subinterval on the partition $\mathcal{P}$, and apply a change of variable to express the increments in terms of the variable $h$:
    \begin{align}
        \epsilon_\Delta 
        & \le \sum_{c=1}^C \int_0^T \mathbb{E} \left | \phi^{(c)}(\theta_s^{(c)}) - \phi^{(c)}_\Delta(s, \theta_{t_{0:n}}^{(c)}) \right | \, \mathrm{d}s \\
        & = \sum_{c=1}^C \sum_{j=0}^{n-1} \int_{t_{j}}^{t_{j+1}} \mathbb{E} \left | \phi^{(c)}(\theta_s^{(c)}) - \phi^{(c)}_\Delta(s, \theta_{t_{0:n}}^{(c)}) \right | \, \mathrm{d}s \\
        & = \sum_{c=1}^C \sum_{j=0}^{n-1} \int_0^\Delta \mathbb{E} \left | \phi^{(c)}(\theta_{t_{j}+h}^{(c)}) - \phi^{(c)}_\Delta(t_{j}+h, \theta_{t_{0:n}}^{(c)}) \right | \, \mathrm{d}h \, .
    \end{align}
    Substituting the definition of $\phi^{(c)}_\Delta$ from Definition \ref{def:trapezium} in the modulus above and using the triangle inequality, we obtain
    \begin{align}
        & \left | \phi^{(c)}(\theta_{t_{j}+h}^{(c)}) - \phi^{(c)}_\Delta(t_{j}+h, \theta_{t_{0:n}}^{(c)}) \right | \\
        & =
        \Big |
            \phi^{(c)}(\theta_{t_{j}+h}^{(c)})
            -
            \phi^{(c)}(\theta_{t_{j}}^{(c)})
            -
            \frac{h}{\Delta}
            \left (
                \phi^{(c)}(\theta_{t_{j+1}}^{(c)})
                -
                \phi^{(c)}(\theta_{t_{j}}^{(c)})
            \right )
        \Big | \\
        & \le
        \Big |
            \phi^{(c)}(\theta_{t_{j}+h}^{(c)})
            -
            \phi^{(c)}(\theta_{t_{j}}^{(c)})
        \Big |
        +
        \frac{h}{\Delta}
        \Big |
            \phi^{(c)}(\theta_{t_{j+1}}^{(c)})
            -
            \phi^{(c)}(\theta_{t_{j}}^{(c)})
        \Big | .
    \end{align}
    Substituting this decomposition into the bound above gives
    \begin{align}
        \epsilon_\Delta
        & \le
        \sum_{c=1}^C \sum_{j=0}^{n-1}
        \int_0^\Delta
        \mathbb{E}
        \left |
            \phi^{(c)}(\theta_{t_{j}+h}^{(c)})
            -
            \phi^{(c)}(\theta_{t_{j}}^{(c)})
        \right |
        \, \mathrm{d}h \\
        & \quad
        +
        \sum_{c=1}^C \sum_{j=0}^{n-1}
        \int_0^\Delta
        \frac{h}{\Delta}
        \mathbb{E}
        \left |
            \phi^{(c)}(\theta_{t_{j+1}}^{(c)})
            -
            \phi^{(c)}(\theta_{t_{j}}^{(c)})
        \right |
        \, \mathrm{d}h \, .
    \end{align}

    Applying Lemma \ref{lem:potential_increment_bound_relaxed} with $h$, and again with $h=\Delta$, gives
    \begin{align}
        & \mathbb{E} \left | \phi^{(c)}(\theta_{t_j+h}^{(c)}) - \phi^{(c)}(\theta_{t_j}^{(c)}) \right | \\
        & \le d \, \mathcal{K}^{(c)} \left( \sqrt{h} + \frac{h}{T}\sqrt{\mathcal{R}^{(c)}} \right) + \frac{Ud}{2} \left( h + \frac{h^2}{T^2}\mathcal{R}^{(c)} \right) \, ,
    \end{align}
    and
    \begin{align}
        & \mathbb{E} \left | \phi^{(c)}(\theta_{t_{j+1}}^{(c)}) - \phi^{(c)}(\theta_{t_j}^{(c)}) \right | \\
        & \le d \, \mathcal{K}^{(c)} \left( \sqrt{\Delta} + \frac{\Delta}{T}\sqrt{\mathcal{R}^{(c)}} \right) + \frac{Ud}{2} \left( \Delta + \frac{\Delta^2}{T^2}\mathcal{R}^{(c)} \right) \, .
    \end{align}

    Substituting these two bounds into the previous decomposition of $\epsilon_\Delta$ gives
    \begin{align}
        \epsilon_\Delta
        & \le d \sum_{c=1}^C \sum_{j=0}^{n-1} \int_0^\Delta \left[ \mathcal{K}^{(c)} \left( \sqrt{h} + \frac{h}{T}\sqrt{\mathcal{R}^{(c)}} \right) + \frac{U}{2} \left( h + \frac{h^2}{T^2}\mathcal{R}^{(c)} \right) \right] \, \mathrm{d}h \\
        & \quad + d \sum_{c=1}^C \sum_{j=0}^{n-1} \int_0^\Delta \frac{h}{\Delta} \left[ \mathcal{K}^{(c)} \left( \sqrt{\Delta} + \frac{\Delta}{T}\sqrt{\mathcal{R}^{(c)}} \right) + \frac{U}{2} \left( \Delta + \frac{\Delta^2}{T^2}\mathcal{R}^{(c)} \right) \right] \, \mathrm{d}h .
    \end{align}
    Evaluating the integrals,
    \begin{align}
        \epsilon_\Delta
        & \le d \sum_{c=1}^C \sum_{j=0}^{n-1} \left[ \mathcal{K}^{(c)} \left( \frac{2}{3}\Delta^{3/2} + \frac{\Delta^2}{2T}\sqrt{\mathcal{R}^{(c)}} \right) + U \left( \frac{\Delta^2}{4} + \frac{\mathcal{R}^{(c)}\Delta^3}{6T^2} \right) \right] \\
        & \quad + d \sum_{c=1}^C \sum_{j=0}^{n-1} \left[ \mathcal{K}^{(c)} \left( \frac{1}{2}\Delta^{3/2} + \frac{\Delta^2}{2T}\sqrt{\mathcal{R}^{(c)}} \right) + U \left( \frac{\Delta^2}{4} + \frac{\mathcal{R}^{(c)}\Delta^3}{4T^2} \right) \right] \\
        & = d \sum_{c=1}^C \sum_{j=0}^{n-1} \left[ \mathcal{K}^{(c)} \left( \frac{7}{6}\Delta^{3/2} + \frac{\Delta^2}{T}\sqrt{\mathcal{R}^{(c)}} \right) + U \left( \frac{\Delta^2}{2} + \frac{5\mathcal{R}^{(c)}\Delta^3}{12T^2} \right) \right].
    \end{align}
    Since $n\Delta=T$,
    \begin{align}
        \epsilon_\Delta
        & \le d \sum_{c=1}^C \left[ \mathcal{K}^{(c)} \left( \frac{7}{6}T\sqrt{\Delta} + \Delta\sqrt{\mathcal{R}^{(c)}} \right) + U \left( \frac{1}{2}T\Delta + \frac{5}{12T}\mathcal{R}^{(c)}\Delta^2 \right) \right] \\
        & := E_\Delta \, .
    \end{align}

\end{proof}

In particular, when $C$, $d$, $T$, $U$, $\mathcal{G}^{(c)}$, and $\lambda^{(c)}$ are fixed, the leading term in $E_\Delta$ is of order $\sqrt{\Delta}$. Thus the same convergence rate in the step size is obtained under Assumption~\ref{as:weak_error_regularity} as under the global Lipschitz condition in Assumption~\ref{as:phi_lip}.

We can now combine this strong-error bound with the change-of-measure argument used in Theorem~\ref{th:bound_strong}.

\begin{theorem}[TV distance bound] \label{th:bound_strong_relaxed}
    Let Assumptions \ref{as:smoothness}, \ref{as:fubini}, \ref{as:phi_min}, and \ref{as:weak_error_regularity} hold.
    Let $0<\Delta\le T$, $n\Delta=T$, and let $\mathcal{R}^{(c)}$ and $\mathcal{K}^{(c)}$ be defined as in Proposition \ref{prop:strong_error_relaxed}.
    Let $\mathbb{F}$ be the fusion measure on $[0,T]$, and let $\mathbb{F}_\Delta$ be the approximate fusion measure on $[0,T]$ obtained using the numerical approximation to calculate the weights. The total variation distance between these distributions can be bounded as follows.
    \begin{align}
        & \| \mathbb{F} - \mathbb{F}_\Delta \|_{\mathrm{TV}} \\
        & \le \frac{1}{Z} \mathrm{e}^{- C \, \bar{\Phi}  \, T} \; d \sum_{c=1}^C \left[ \mathcal{K}^{(c)} \left( \frac{7}{6}T\sqrt{\Delta} + \Delta\sqrt{\mathcal{R}^{(c)}} \right) + U \left( \frac{1}{2}T\Delta + \frac{5}{12T}\mathcal{R}^{(c)}\Delta^2 \right) \right] \, ,
    \end{align}
    with $\bar{\Phi} = C^{-1} \sum_{c=1}^C \Phi^{(c)}$, where $\Phi^{(c)}$ is the lower bound from Assumption \ref{as:phi_min}, and $Z$ is the normalising constant of $\mathbb{F}$,
    \begin{equation}
        Z = \mathbb{E} \! \left[ \rho_0(\theta_0) \exp \left\{ -\sum_{c=1}^C\int_0^T\phi^{(c)}(\theta_s^{(c)}) \, \mathrm{d}s \right\} \right] \, .
    \end{equation}
    Note that the exponent $- C \, \bar{\Phi}  \, T$ is always positive, as shown in Remark \ref{rem:Phi_neg}.
\end{theorem}

\begin{proof}[Proof of Theorem \ref{th:bound_strong_relaxed}]
    The proof follows the same steps as the proof of Theorem~\ref{th:bound_strong}, except for the final step.
    Using Proposition \ref{prop:strong_error_relaxed},
    \begin{align}
        & \| \mathbb{F} - \mathbb{F}_\Delta \|_{\mathrm{TV}} \\
        & \le \frac{1}{Z} \mathrm{e}^{- C \, \bar{\Phi}  \, T}  \cdot \epsilon_\Delta \\
        & \le \frac{1}{Z} \mathrm{e}^{- C \, \bar{\Phi}  \, T}  \cdot E_\Delta  \\
        & = \frac{1}{Z} \mathrm{e}^{- C \, \bar{\Phi}  \, T} \; d \sum_{c=1}^C \left[ \mathcal{K}^{(c)} \left( \frac{7}{6}T\sqrt{\Delta} + \Delta\sqrt{\mathcal{R}^{(c)}} \right) + U \left( \frac{1}{2}T\Delta + \frac{5}{12T}\mathcal{R}^{(c)}\Delta^2 \right) \right] \, .
    \end{align}
\end{proof}

\section{Detailed results for the logistic regression} \label{app:logit_details}

This appendix reports the detailed results of the logistic regression experiment of Section~\ref{sec:logistic_comparison}. Table~\ref{tab:logistic_iad_comparison} lists, for each method and target $\zeta_\star$, the number of SMC steps $n$ selected by the bisection, together with the mean and range of the IAD and of the runtime across the ten replicas. Tuning $T$ and $n$ took between $5$ and $14$ seconds for ABF, and between $44$ and $51$ seconds for BF.

\begin{table}[htbp]
    \centering
    \begin{tabular}{lcccc}
        \toprule
        Method & $\zeta_\star$ & $n$ & IAD & Runtime (seconds) \\
        \midrule
        ABF & $0.1$ & $5$ & $0.0444\;(0.0317,\,0.0600)$ & $7.19\;(5.08,\,8.10)$ \\
        ABF & $0.2$ & $7$ & $0.0293\;(0.0254,\,0.0359)$ & $9.57\;(7.35,\,10.6)$ \\
        ABF & $0.5$ & $11$ & $0.0316\;(0.0256,\,0.0430)$ & $14.6\;(11.2,\,16.1)$ \\
        ABF & $0.8$ & $21$ & $0.0314\;(0.0227,\,0.0536)$ & $26.0\;(20.0,\,28.8)$ \\
        ABF & $0.9$ & $32$ & $0.0264\;(0.0204,\,0.0306)$ & $39.2\;(30.2,\,43.0)$ \\
        ABF & $0.99$ & $108$ & $0.0281\;(0.0240,\,0.0393)$ & $132\;(112,\,143)$ \\
        \midrule
        BF & $0.1$ & $32$ & $0.0739\;(0.0534,\,0.100)$ & $958\;(926,\,1{,}005)$ \\
        BF & $0.2$ & $52$ & $0.0578\;(0.0445,\,0.0751)$ & $1{,}175\;(1{,}148,\,1{,}215)$ \\
        BF & $0.3$ & $75$ & $0.0517\;(0.0411,\,0.0642)$ & $1{,}445\;(1{,}392,\,1{,}616)$ \\
        BF & $0.4$ & $101$ & $0.0471\;(0.0402,\,0.0600)$ & $1{,}771\;(1{,}714,\,1{,}812)$ \\
        \midrule
        CMC & -- & -- & $0.0607\;(0.0596,\,0.0620)$ & $0.172\;(0.146,\,0.265)$ \\
        \bottomrule
    \end{tabular}
    \caption{Accuracy and per-replica runtime for the logistic regression experiment, where $n$ is the number of SMC steps and $\Delta = T / n$ with $T = 0.0353$. Each entry reports the mean across ten replicas, followed by the observed range in parentheses.}
    \label{tab:logistic_iad_comparison}
\end{table}

The bisection selected $n = 32$ for both ABF with $\zeta_\star = 0.9$ and BF with $\zeta_\star = 0.1$, so the two methods can be compared on a common mesh. On this mesh, the ESS after the first transition is $0.90 N$ for ABF and $0.14 N$ for BF, which illustrates the greater variability of the BF weights. ABF has a mean IAD of $0.026$, against $0.074$ for BF, and runs about $25$ times faster.

As the mesh is refined, the IAD of BF decreases, but even its finest mesh remains, on average, less accurate than every ABF configuration. The IAD of ABF is largest on the coarsest mesh and roughly constant for $n \ge 7$. CMC is improved upon by ABF at every $\zeta_\star$, and by BF for $\zeta_\star \ge 0.2$.

\section{Detailed results for the nycflights13 experiment} \label{app:nycflights_details}

This appendix gives the settings of the Stan samplers for the nycflights13 data of Section~\ref{sec:nycflights}, and compares ABF with CMC. Figure~\ref{fig:nycflights_iad_compare_cmc} shows the results of Figure~\ref{fig:nycflights_iad_compare} together with those of CMC, which uses the same values of $N$ as the number of draws per shard.

\begin{figure}[!htb]
    \centering
    \includegraphics[width=1\linewidth]{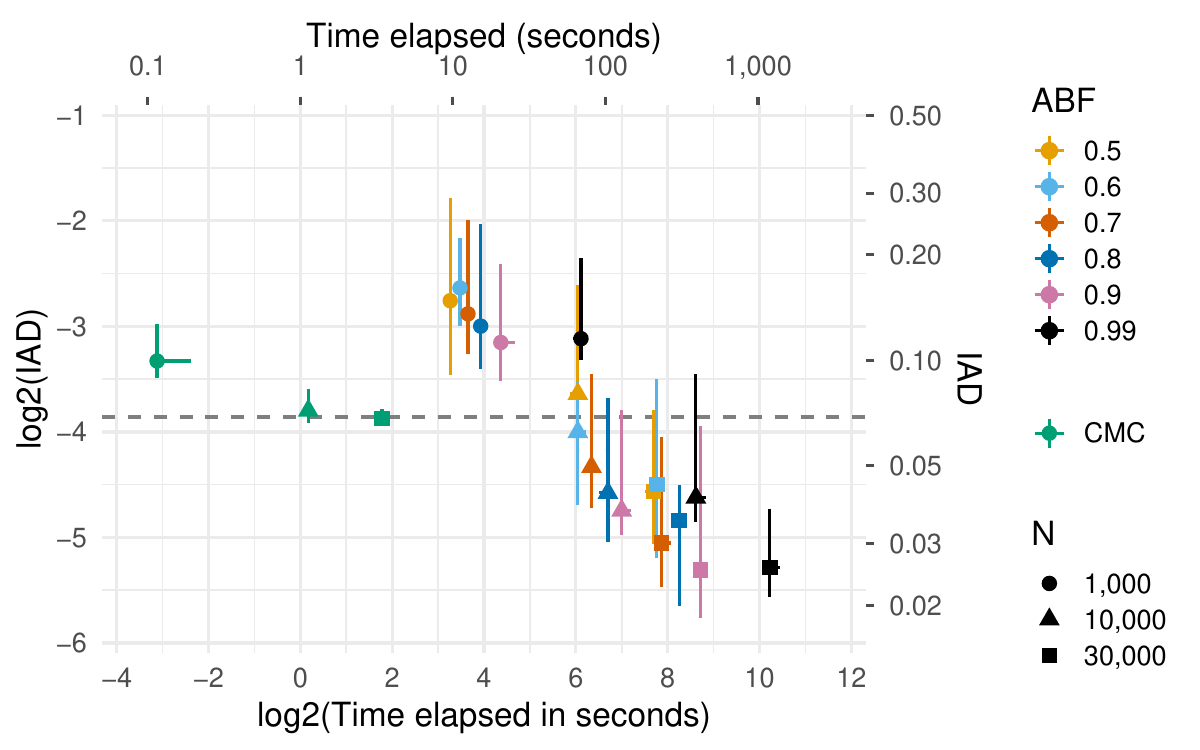}
    \caption{IAD against runtime of ABF and CMC for the logistic regression on the nycflights13 data, on the $\log_2$ scales of \citet[Figure 7]{chan2023dcf}. Colours indicate the target $\zeta_\star$ of ABF, with CMC in a separate colour, and marker shapes the number of particles $N$. The markers indicate the median runtime and the median IAD over ten replicas, and the bars extend from the minimum to the maximum. As in \citet[Figure 7]{chan2023dcf}, the dashed horizontal line marks the smallest mean IAD of CMC amongst all $N$. Points nearer the lower-left corner are preferable.}
    \label{fig:nycflights_iad_compare_cmc}
\end{figure}

CMC is faster than ABF by about two orders of magnitude at every $N$. Its IAD decreases from $N = 1{,}000$ to $N = 10{,}000$ but changes little beyond, which indicates that its remaining error is mostly systematic. ABF is less accurate than CMC for $N = 1{,}000$, but it is already more accurate for $N = 10{,}000$, its IAD keeps decreasing as $N$ increases, and for $N = 30{,}000$ almost every ABF replica is more accurate than every CMC replica.

\subsection{Shard and reference samples}

For each shard, a single Stan chain \citep{carpenter2017} is run for $1{,}000$ warm-up iterations, followed by enough iterations for the bulk and tail effective sample sizes of every coordinate to reach $300{,}000$, and the chain is then thinned to $300{,}000$ draws. These form ten blocks of $30{,}000$ draws: every ABF replica uses the first $N$ draws of the first block, while CMC replica $k$ uses the first $N$ draws of block $k$. The reference sample consists of $256{,}000$ draws from $64$ Stan chains run in parallel on the full data, each with $1{,}000$ warm-up iterations followed by $4{,}000$ kept iterations and no thinning.

\section{Details of the proportional hazards model} \label{sec:ph_model_prior}

This appendix describes the data of the proportional hazards model of Section~\ref{sec:ph} and the settings of the Stan samplers, gives its likelihood, prior and fractional prior, explains the choice of centred coordinates, shows that $\phi^{(c)}$ is not bounded below, and reports the runs with longer steps.

\subsection{Data}
The \texttt{hds} data set arranges the records of each drive in the start--stop format used in survival analysis \citep{christoffersen2021, hds-doc}. After keeping one record per drive and the $G = 4$ families with the most drives, $m = 57{,}954$ drives remain, of which $1{,}717$ failed during the observation period. Three of the families have a capacity of $4$ terabytes and one of $8$ terabytes, and the number of failures per family ranges from $17$ to $1{,}565$. Across the $C = 16$ shards, each shard contains $3{,}622$ or $3{,}623$ drives and $107$ or $108$ failures.

\subsection{Shard and reference samples}

The shard and reference samples are obtained as for the nycflights13 data (Appendix~\ref{app:nycflights_details}).

\subsection{Likelihood}
For drive $i$, let $a_i \geq 0$ be the age at which it enters the data, $b_i > a_i$ the age at which it is last observed, both in years of operation, and $\delta_i \in \{0, 1\}$ indicate a failure at $b_i$. Write $\eta_i = \beta_0 + \beta_1 x_i + u_{g_i}$. By Equation~\eqref{eq:ph_hazard}, the cumulative hazard of drive $i$ is $\int_0^t h_i(s) \, \mathrm{d} s = e^{\eta_i} t^{\gamma}$, and its survival function is $\mathcal{S}_i(t) = \exp(- e^{\eta_i} t^{\gamma})$. The drive enters the data only because it survived to $a_i$, so its likelihood is $h_i(b_i)^{\delta_i} \mathcal{S}_i(b_i) / \mathcal{S}_i(a_i)$: the density of a failure at $b_i$ if $\delta_i = 1$, or the probability of surviving beyond $b_i$ if the drive is censored, conditional on survival to $a_i$. Its logarithm is
\begin{align}
    \mathcal{L}_i & = \delta_i \log h_i(b_i) + \log \mathcal{S}_i(b_i) - \log \mathcal{S}_i(a_i) \nonumber \\
    & = \delta_i \left \{ \log \gamma + (\gamma - 1) \log b_i + \eta_i \right \} - e^{\eta_i} \left ( b_i^\gamma - a_i^\gamma \right ) \, , \label{eq:ph_loglik_row}
\end{align}
where we took the logarithm of each factor, and substituted the logarithm of Equation~\eqref{eq:ph_hazard} and the expression of $\mathcal{S}_i$. All drives of a family have the same capacity, so $x_i = x_{g_i}$ and $\eta_i = \eta_{g_i}$, with $\eta_g = \beta_0 + \beta_1 x_g + u_g$. Summing Equation~\eqref{eq:ph_loglik_row} over the drives of shard $c$, and taking out of the sum over each family the factors that do not depend on $i$, gives
\begin{equation}\label{eq:ph_grouped_loglik}
    \mathcal{L}_c(q) = \sum_{g=1}^G \left \{ E_g (\log \gamma + \eta_g) + (\gamma - 1) S_g - e^{\eta_g} D_g(\gamma) \right \} \, ,
\end{equation}
where, if $\mathcal{I}_{c,g}$ is the set of drives of family $g$ in shard $c$,
\begin{equation}
    E_g = \sum_{i \in \mathcal{I}_{c,g}} \delta_i \, , \qquad S_g = \sum_{i \in \mathcal{I}_{c,g}} \delta_i \log b_i \, , \qquad D_g(\gamma) = \sum_{i \in \mathcal{I}_{c,g}} \left ( b_i^\gamma - a_i^\gamma \right ) \, .
\end{equation}
The counts $E_g$ and $S_g$ are computed once, and drives with the same $(a_i, b_i)$ contribute identical terms to $D_g$, so each distinct pair is evaluated once with its multiplicity.

\subsection{Prior and fractional prior}
The prior is
\begin{align}
    \beta_0 & \sim \mathcal{N}(-4, 2^2) \, , & \log \gamma & \sim \mathcal{N}(\log 1.2, 0.7^2) \, , \nonumber \\
    \beta_1 & \sim \mathcal{N}(0, 1) \, , & \omega & \sim \operatorname{HalfNormal}(0, 0.8^2) \, ,
\end{align}
together with $u_g \mid \omega \overset{\mathrm{iid}}{\sim} \mathcal{N}(0, \omega^2)$. The half-normal prior has positive density at $\omega = 0$, so it does not rule out that all families have the same hazard. Its scale is consistent with $16$ further families of the data set, which the analysis does not use: the standard deviation of the logarithms of their crude failure rates is $1.04$, and the prior gives $\omega > 1$ probability $0.21$. Following \citet{chan2023dcf}, shard $c$ targets $f_c(q) \propto \exp \{ \mathcal{L}_c(q) \} \{ \pi_q(q) \}^{1/C}$, where $\pi_q$ is the prior density of $q$, so that $\prod_{c=1}^C f_c$ is proportional to the posterior density. Write $u = (u_1, \ldots, u_G)^\top$. The density of $\log \omega$ is $\omega \exp \{ - \omega^2 / (2 \cdot 0.8^2) \}$, where the factor $\omega$ is the Jacobian of the map from $\omega$ to $\log \omega$, and the density of $u_g$ given $\omega$ is proportional to $\omega^{-1} \exp \{ - u_g^2 / (2 \omega^2) \}$. The terms of $\log f_c$ that depend on $\log \omega$ or on $u$, apart from $\mathcal{L}_c$, are therefore
\begin{equation}\label{eq:ph_prior_lambda_u}
    p_c(\log \omega, u) = \frac{1}{C} \left ( \log \omega - \frac{\omega^2}{2 \cdot 0.8^2} - G \log \omega - \frac{\| u \|^2}{2 \omega^2} \right ) \, .
\end{equation}
In the fusions of the divide-and-conquer tree, the exponent $1/C$ is replaced by the share of the prior carried by each child.

\subsection{Centred coordinates and the funnel}
For the families with many failures, the data determine $\eta_g$ closely, even when $\beta_0$, $\beta_1$ and $\omega$ are poorly determined. In the coordinates $q$, $\eta_g$ is linear, so the set on which it is constant is a hyperplane, and the mean of a Brownian bridge between two points of this set, which is the segment between them, stays in it. In the non-centred coordinates, with $u_g = \omega z_g$, this set is curved, the bridges leave it, and the data give the corresponding paths very small weights. This agrees with the general finding that centred parameterisations are preferable when the data are informative about each group effect \citep{papaspiliopoulos2007}. Family $4$, with $17$ failures, is the exception: for it, the centred coordinates produce the funnel between $u_4$ and $\log \omega$.

\subsection{No lower bound on \texorpdfstring{$\phi^{(c)}$}{phi(c)}}
Differentiating Equation~\eqref{eq:ph_prior_lambda_u}, using $\mathrm{d} \omega^2 / \mathrm{d} \log \omega = 2 \omega^2$ and $\mathrm{d} \omega^{-2} / \mathrm{d} \log \omega = - 2 \omega^{-2}$, gives
\begin{align}
    \frac{\partial p_c}{\partial u_g} & = - \frac{u_g}{C \omega^2} \, , & \frac{\partial p_c}{\partial \log \omega} & = \frac{1}{C} \left ( 1 - G - \frac{\omega^2}{0.8^2} + \frac{\| u \|^2}{\omega^2} \right ) \, , \nonumber \\
    \frac{\partial^2 p_c}{\partial u_g^2} & = - \frac{1}{C \omega^2} \, , & \frac{\partial^2 p_c}{\partial (\log \omega)^2} & = - \frac{1}{C} \left ( \frac{2 \omega^2}{0.8^2} + \frac{2 \| u \|^2}{\omega^2} \right ) \, ,
\end{align}
together with $\partial^2 p_c / \partial \log \omega \, \partial u_g = 2 u_g / (C \omega^2)$ and $\partial^2 p_c / \partial u_g \, \partial u_h = 0$ for $g \neq h$. Fix $\beta_0$, $\beta_1$ and $\gamma$, and set $u = 0$. Then the gradient of $p_c$ is bounded as $\omega \to 0$, and its Hessian is diagonal. The log-likelihood $\mathcal{L}_c$ depends on $\omega$ and $u$ only through $\eta_g$, which does not depend on $\omega$, so at $u = 0$ its gradient and Hessian do not depend on $\omega$, and neither do the priors on $\beta_0$, $\beta_1$ and $\log \gamma$. In Equation~\eqref{eq:phi_precond}, the first term is therefore bounded as $\omega \to 0$, and in the second term only the entries $\partial^2 p_c / \partial u_g^2$ grow without bound. Hence, if $\Lambda_{c,uu}$ denotes the block of $\Lambda_c$ that corresponds to $u$,
\begin{equation}\label{eq:ph_phi_neck}
    \phi^{(c)}(q) = - \frac{\operatorname{Tr} ( \Lambda_{c,uu} )}{2 C \omega^2} + \mathcal{O}(1) \to - \infty \quad \text{as } \omega \to 0 \text{ with } u = 0 \, ,
\end{equation}
because $\Lambda_c$ is positive definite, so $\operatorname{Tr} ( \Lambda_{c,uu} ) > 0$.

\subsection{Step size tests}
Table~\ref{tab:ph_zeta_N} reports the median IAD of ABF over ten replicas for all the values of $\zeta_\star$ that we ran. For $\zeta_\star \in \{0.95, 0.99\}$, the IAD decreases as $N$ increases. For $\zeta_\star \in \{0.85, 0.9\}$, it decreases from $N = 1{,}000$ to $N = 10{,}000$, and then increases: the bias caused by the longer steps grows with $N$. A particle that lands at a grid point deep in the neck of the funnel receives a very large weight (Section~\ref{sec:ph}), and the more particles there are, the more likely it is that one of them lands there. Shorter steps reduce this weight, which is why, for $\zeta_\star \in \{0.95, 0.99\}$, the IAD keeps decreasing up to $N = 30{,}000$.

\begin{table}[htbp]
    \centering
    \begin{tabular}{lccc}
        \toprule
        $\zeta_\star$ & $N = 1{,}000$ & $N = 10{,}000$ & $N = 30{,}000$ \\
        \midrule
        $0.85$ & $0.114$ & $0.054$ & $0.116$ \\
        $0.9$ & $0.108$ & $0.051$ & $0.067$ \\
        $0.95$ & $0.139$ & $0.045$ & $0.039$ \\
        $0.99$ & $0.159$ & $0.040$ & $0.031$ \\
        \bottomrule
    \end{tabular}
    \caption{Median IAD of ABF over ten replicas for the proportional hazards model.}
    \label{tab:ph_zeta_N}
\end{table}

To check that the bias comes from the time discretisation, we reran ABF with $\zeta_\star = 0.85$ and $N = 30{,}000$ as follows. Each fusion of the tree first runs the adaptive algorithm as a pilot, and is then rerun on the pilot's grid with each step split into $k$ equal parts. Table~\ref{tab:ph_finegrid} compares the results with those of the adaptive runs. Splitting each step into four removes the bias of $\omega$ and brings the IAD below that of $\zeta_\star = 0.99$ in Table~\ref{tab:ph_zeta_N}.

Rerunning on the pilot's grid without splitting it, $k = 1$, is much worse than the adaptive run. The rerun uses new random numbers, so its particles reach the neck at different times from those of the pilot, and the fixed grid cannot shorten the steps in response, as the adaptive rule does. The resulting errors are then passed up the tree: at the fusions above the leaves, the pilots, which start from the output of these reruns, needed up to $4 \times 10^4$ steps, whereas no fusion of the adaptive runs needed more than $25$.

\begin{table}[htbp]
    \centering
    \begin{tabular}{lccc}
        \toprule
        Grid & Replicas & Error of the mean of $\omega$ & IAD \\
        \midrule
        Adaptive & $10$ & $-0.50$ & $0.116$ \\
        Pilot grid, $k = 1$ & $4$ & $-2.28$ & $0.687$ \\
        Pilot grid, $k = 2$ & $5$ & $-0.28$ & $0.062$ \\
        Pilot grid, $k = 4$ & $4$ & $0.07$ & $0.019$ \\
        \bottomrule
    \end{tabular}
    \caption{ABF with $\zeta_\star = 0.85$ and $N = 30{,}000$ on the proportional hazards model, on the adaptive grid and on the grid of an adaptive pilot run with each step split into $k$ equal parts. The error of the weighted mean of $\omega$ is in units of the standard deviation of the reference sample. Medians over the replicas.}
    \label{tab:ph_finegrid}
\end{table}

\end{document}